\documentclass[11pt]{article}
\usepackage[margin=1in]{geometry}

\usepackage[english]{babel}
\usepackage[utf8]{inputenc}

\usepackage{amsmath,amsfonts,amsthm,amssymb}
\usepackage{graphicx}
\usepackage[colorlinks=true, allcolors=blue]{hyperref}
\usepackage[capitalise]{cleveref}
\usepackage{caption}
\usepackage{subcaption}
\usepackage{todonotes}
\usepackage{authblk}
\usepackage{xspace}
\usepackage{thmtools, thm-restate}
\declaretheorem{theorem}
\usepackage{tikz}
\usetikzlibrary{calc, shapes, backgrounds, graphs, arrows.meta, bending,decorations,snakes}
\usepackage{xcolor}
\crefname{algocf}{Algorithm}{Algorithms}
\Crefname{algocf}{Algorithm}{Algorithms}
\usepackage{verbatim}
\usepackage[noline,noend]{algorithm2e}
\usepackage{subcaption}
\usepackage{thm-restate}
\usepackage{float}
\usepackage{enumerate}

\counterwithin{theorem}{section}
\newtheorem{lemma}[theorem]{Lemma}

\newtheorem{definition}{Definition}
\newtheorem{example}{Example}
\newtheorem{claim}{Claim}
\newtheorem{fact}[theorem]{Fact}
\newtheorem{remark}[theorem]{Remark}
\newtheorem{proposition}[theorem]{Proposition}
\newtheorem{corollary}[theorem]{Corollary}
\newtheorem{invariant}[theorem]{Invariant}

\def\dd{\mathinner{.\,.}}
\newcommand{\cO}{\mathcal{O}}
\newcommand{\Oh}{\cO}

\newcommand{\per}{\textsf{per}}
\newcommand{\LCE}{\textsf{LCP}}
\newcommand{\suc}{\textsf{succ}}
\newcommand{\csucc}{\textsf{lex-succ}}

\newcommand{\cpred}{\textsf{lex-pred}}

\newcommand{\problemtwo}{\textsc{Two String Families LCP Problem}\xspace}
\newcommand{\problemlambda}{$\lambda$ \textsc{String Families LCP Problem}\xspace}

\newcommand{\F}{\mathcal{F}}
\newcommand{\M}{\mathcal{M}}
\renewcommand{\L}{\mathcal{L}}
\newcommand{\R}{\mathcal{R}}
\newcommand{\T}{\mathcal{T}}
\newcommand{\TT}{\mathcal{T}}
\newcommand{\maxPairLCP}{\mathrm{maxPairLCP}}
\newcommand{\maxMultiLCP}{\mathrm{maxMultiLCP}}
\newcommand{\maxLCP}{\mathrm{maxLCP}}
\renewcommand{\P}{\mathcal{P}}
\newcommand{\Q}{\mathcal{Q}}
\newcommand{\LCP}{\mathsf{LCP}}
\newcommand{\LCPs}{\mathsf{LCPs}}
\newcommand{\sub}{\subseteq}
\newcommand{\sm}{\setminus}
\newcommand{\U}{\mathcal{U}}
\newcommand{\V}{\mathcal{V}}
\newcommand{\G}{\mathcal{G}}
\newcommand{\Gam}{\mathbf{G}}
\newcommand{\val}{\mathsf{val}}
\newcommand{\Pam}{\mathbf{P}}
\newcommand{\Zp}{\mathbb{Z}_+}

\newcommand{\nnext}{\mathsf{next}}
\newcommand{\pprev}{\mathsf{prev}}

\newcommand{\LeftMis}{\mathsf{LeftMisper}}
\newcommand{\RightMis}{\mathsf{RightMisper}}
\newcommand{\Misp}{\mathsf{Misper}}

\newcommand\twocol[2]{%
\begin{center}%
\begin{minipage}[t]{0.5\textwidth}%
{#1}%
\end{minipage}\hfill%
\begin{minipage}[t]{0.5\textwidth}%
{#2}%
\end{minipage}%
\end{center}}

\definecolor{darkblue}{rgb}{0.2,0.2,0.6}

\newcommand{\defproblem}[3]{
  \vspace{2mm}
\noindent\fbox{
  \begin{minipage}{0.96\textwidth}
  \textsc{#1}\\
  {\bf{Input:}} #2  \\
  {\bf{Output:}} #3
  \end{minipage}
  }
  \vspace{2mm}
}

\usepackage{mathtools}
\DeclarePairedDelimiter{\floor}{\lfloor}{\rfloor}

\title{Faster Algorithms for Longest Common Substring}

\author[1]{Panagiotis Charalampopoulos}
\author[2]{Tomasz Kociumaka}
\author[3,4]{Jakub Radoszewski}
\author[5,6]{Solon~P.~Pissis}

\affil[1]{King's College London, UK\\
    \texttt{p.charalampopoulos@kcl.ac.uk}}
\affil[2]{Max Planck Institute for Informatics, Saarland Informatics Campus, Germany\\
    \texttt{tomasz.kociumaka@mpi-inf.mpg.de}}
\affil[3]{University of Warsaw, Poland\\
    \texttt{jrad@mimuw.edu.pl}}
\affil[4]{Samsung R\&D Institute Poland}
\affil[5]{CWI, Amsterdam, The Netherlands\\
    \texttt{solon.pissis@cwi.nl}}
\affil[6]{Vrije Universiteit, Amsterdam, The Netherlands}
    
\date{\vspace{-5ex}}

\begin{document}
\maketitle
\thispagestyle{empty}

\begin{abstract}
In the classic longest common substring (LCS) problem, we are given two strings $S$ and $T$ of total length $n$ over an alphabet of size $\sigma$, and we are asked to find a longest string occurring as a fragment of both $S$ and~$T$.
Weiner, in his seminal paper that introduced the suffix tree, presented an $\cO(n \log \sigma)$-time algorithm for the LCS problem [SWAT 1973]. For polynomially-bounded integer alphabets, the linear-time construction of suffix trees by Farach yielded an $\cO(n)$-time algorithm for the LCS problem [FOCS 1997]. 
However, for small alphabets, this is not necessarily optimal for the LCS problem in the word RAM model of computation, in which  the strings can be stored in $\cO(n \log \sigma/\log n )$ space and read in $\cO(n \log \sigma/\log n )$ time.
We show that we can compute an LCS of two strings in time $\cO(n \log \sigma / \sqrt{\log n})$ in the word RAM model, which is sublinear in~$n$ if $\sigma=2^{o(\sqrt{\log n})}$ (in particular, if $\sigma=\Oh(1)$), using optimal space $\cO(n \log \sigma/\log n)$. 
In fact, it was recently shown that this result is conditionally optimal [Kempa and Kociumaka, STOC 2025]. The same complexity can be achieved for computing an LCS of $\lambda = \Oh(\sqrt{\log n}/\log \log n)$ input strings of total length $n$.

We then lift our ideas to the problem of computing a $k$-mismatch LCS, which has received considerable attention in recent years. In this problem, the aim is to compute a longest substring of $S$ that occurs in~$T$ with at most $k$ mismatches. Flouri et al.~showed how to compute a 1-mismatch LCS in $\Oh(n \log n)$ time [IPL 2015]. Thankachan et al.~extended this result to computing a $k$-mismatch LCS in $\Oh(n \log^k n)$ time for $k=\Oh(1)$ [J.~Comput.~Biol.~2016]. 
We show an $\Oh(n \log^{k-0.5} n)$-time algorithm, for any constant integer $k>0$ and \emph{irrespective} of the alphabet size, using $\Oh(n)$ space as the previous approaches.
We thus notably break through the well-known $n \log^k n$ barrier, which stems from a recursive heavy-path decomposition technique that was first introduced in the seminal paper of Cole et al.\ for string indexing with $k$ errors [STOC 2004]. As a by-product, we improve upon the algorithm of Charalampopoulos et al.\ [CPM 2018] for computing a $k$-mismatch LCS in the case when the output $k$-mismatch LCS is sufficiently long.
\end{abstract}

\clearpage
\setcounter{page}{1}

\section{Introduction}

In the classic longest common substring (LCS) problem, we are given two strings $S$ and $T$ of total length $n$ over an alphabet of size $\sigma$, and we are asked to compute a longest string occurring as a fragment of both $S$ and $T$ (see \cref{fig:0-1-LCS} for an example). The problem was conjectured by Knuth to require $\Omega(n\log n)$ time until Weiner, in his seminal paper introducing the suffix tree~\cite{DBLP:conf/focs/Weiner73}, showed that the LCS problem can be solved in $\cO(n)$ time when $\sigma$ is constant via constructing the suffix tree of string $S\#T$ for a sentinel letter $\#$. 
Later, Farach showed that even if $\sigma$ is not constant, the suffix tree can be constructed in linear time in addition to the time required for sorting the characters of $S$ and $T$~\cite{DBLP:conf/focs/Farach97}. This yielded an $\cO(n)$-time algorithm for the LCS problem in the word RAM model for polynomially-bounded integer alphabets. While Farach's algorithm for suffix tree construction is \emph{optimal} for all alphabets (the suffix tree by definition has size $\Theta(n)$), optimality does not carry over to the LCS problem. We were thus motivated to answer the following basic question: 

\begin{center}
\emph{Can the LCS problem be solved in $o(n)$ time when $\log\sigma=o(\log n)$?} 
\end{center}

\begin{figure}[ht]
\centering
\begin{tikzpicture}
\begin{scope}[xscale=0.7]
  \begin{scope}
    \foreach \x/\c in {1/c,2/b,8/a,9/b,10/c,11/a,12/b,13/c,14/b,15/a}{
      \draw (\x*0.5,0) node[above] {\texttt{\c}};
    }    
    \foreach \x/\c in {3/b,4/a,5/a,6/b,7/c}{
      \draw (\x*0.5,0) node[above] {\textcolor{blue}{\texttt{\c}}};
    }    
    \draw[yshift=0.2cm] (1.25,0) -- (1.25,-0.2) -- (3.75,-0.2) -- (3.75,0);
  \end{scope}
  \begin{scope}[yshift=-1cm]
    \foreach \x/\c in {1/d,2/a,3/a,9/b,10/b,11/a}{
      \draw (\x*0.5,0) node[above] {\texttt{\c}};
    }    
    \foreach \x/\c in {4/b,5/a,6/a,7/b,8/c}{
      \draw (\x*0.5,0) node[above] {\textcolor{blue}{\texttt{\c}}};
    }    
    \draw[xshift=0.5cm,yshift=0.2cm] (1.25,0) -- (1.25,-0.2) -- (3.75,-0.2) -- (3.75,0);
  \end{scope}
  \begin{scope}[xshift=10cm]
    \foreach \x/\c in {1/c,2/b,3/b,4/a,5/a,6/b,7/c,15/a}{
      \draw (\x*0.5,0) node[above] {\texttt{\c}};
    }
    \foreach \x/\c in {8/a,9/b,11/a,12/b,13/c,14/b}{
      \draw (\x*0.5,0) node[above] {\textcolor{blue}{\texttt{\c}}};
    }
    \foreach \x/\c in {10/c}{
      \draw (\x*0.5,0) node[above] {\textcolor{red}{{\texttt{\c}}}};
    }
    \draw[yshift=0.2cm] (3.75,0) -- (3.75,-0.2) -- (7.25,-0.2) -- (7.25,0);
  \end{scope}
  \begin{scope}[yshift=-1cm,xshift=10cm]
    \foreach \x/\c in {1/d,2/a,10/b,11/a}{
      \draw (\x*0.5,0) node[above] {\texttt{\c}};
    }
    \foreach \x/\c in {3/a,4/b,6/a,7/b,8/c,9/b}{
      \draw (\x*0.5,0) node[above] {\textcolor{blue}{\texttt{\c}}};
    }
    \foreach \x/\c in {5/a}{
      \draw (\x*0.5,0) node[above] {\textcolor{red}{{\texttt{\c}}}};
    }
    \draw[yshift=0.2cm,xshift=-2.5cm] (3.75,0) -- (3.75,-0.2) -- (7.25,-0.2) -- (7.25,0);
  \end{scope}
\end{scope}    
\end{tikzpicture}
\caption{Left: The LCS of the two strings has length 5. Right: The 1-LCS (1-mismatch LCS) of the same two strings has length 7; the mismatching letters are shown in red.}\label{fig:0-1-LCS}
\end{figure}

We consider the word RAM model and assume an alphabet $[0\dd \sigma)=\{0,1,\ldots,\sigma-1\}$. Any string of length $n$ can then be stored in $\cO(n \log \sigma/\log n )$ space and read in $\cO(n \log \sigma/\log n )$ time.
We answer the above question positively when $\log \sigma=o(\sqrt{\log n})$: 

\begin{theorem}\label{thm:sublinear}
    Given two strings $S$ and $T$ of total length $n$ over an alphabet $[0\dd \sigma)$, the LCS problem can be solved in
    $\cO(n\log\sigma/\sqrt{\log n})$ time using $\Oh(n \log\sigma/ \log n)$ space.
\end{theorem}

In a work subsequent to ours, Kempa and Kociumaka~\cite{DBLP:conf/stoc/KempaK25} showed that a variant of the dictionary matching problem underlies the hardness of a plethora of problems in the word RAM model. In particular, they showed that the time complexity in \cref{thm:sublinear} is optimal conditioned to said problem for all non-unary alphabets.
Notably, as a byproduct of their hardness reduction, they obtain a reduction from any instance of the LCS problem with strings of total length $\cO(n)$ under alphabet $[0\dd \sigma)$ to an instance of the LCS problem with strings of total length $\cO(n \log \sigma)$ under the binary alphabet.

The classic solution for the LCS problem over integer alphabets allows to compute an LCS of $\Oh(1)$ strings of total length $n$ in $\Oh(n)$ time; with extra care, it can also be extended to arbitrarily many strings in the same time complexity~\cite{DBLP:conf/cpm/Hui92}. We show that the result of \cref{thm:sublinear} can be achieved also for computing an LCS of $\lambda=\Oh(\sqrt{\log n}/\log \log n)$ input strings of total length $n$ (cf.\ \cref{thm:sublinear2}).

We also consider the following generalisation of the LCS problem that allows for mismatches (see \cref{fig:0-1-LCS} for an example).

\defproblem{$k$-Mismatch Longest Common Substring ($k$-LCS)}{
Two strings $S$ and $T$ of total length $n$ and an integer $k>0$.}{
A pair $S'$, $T'$ of substrings of $S$ and $T$, respectively, with Hamming distance (i.e., number of mismatches) at most $k$ and maximal length $|S'|=|T'|$.
}

Flouri et al.~presented an $\Oh(n \log n)$-time algorithm for the 1-LCS problem~\cite{DBLP:journals/ipl/FlouriGKU15}.
(Earlier work on this problem includes~\cite{DBLP:journals/poit/BabenkoS11}.)
This was generalised by Thankachan et al.~\cite{DBLP:journals/jcb/ThankachanAA16} to an algorithm for the $k$-LCS problem that works in $\Oh(n \log^k n)$ time if $k=\Oh(1)$. Both algorithms use $\Oh(n)$ space.  In~\cite{DBLP:conf/cpm/Charalampopoulos18}, Charalampopoulos et al.~presented an $\cO(n + n \log^{k+1} n/\sqrt{\ell})$-time algorithm  for $k$-LCS with $k=\Oh(1)$, where $\ell$ is the length of a $k$-LCS.\@
For general $k$, Flouri et al.~presented an $\Oh(n^2)$-time algorithm that uses $\cO(1)$ additional space~\cite{DBLP:journals/ipl/FlouriGKU15}.
Grabowski~\cite{DBLP:journals/ipl/Grabowski15} presented two algorithms with running times $\Oh (n ((k+1) (\ell_0+1))^k)$ and $\Oh (n^2 k/\ell_k)$, where $\ell_0$ and $\ell_k$ are, respectively, the length of an LCS of $S$ and~$T$ and the length of a $k$-LCS of $S$ and $T$.
Abboud et al.~\cite{Abboud:2015:MAP:2722129.2722146} employed the polynomial method to obtain a $k^{1.5} n^2 / 2^{\Omega(\sqrt{\log n/k})}$-time randomised algorithm.
In~\cite{DBLP:journals/algorithmica/KociumakaRS19}, Kociumaka et al.~showed that, assuming the Strong Exponential Time Hypothesis (SETH)~\cite{DBLP:journals/jcss/ImpagliazzoP01,DBLP:journals/jcss/ImpagliazzoPZ01}, no strongly subquadratic-time solution for $k$-LCS exists for $k=\Omega(\log n)$. The authors of~\cite{DBLP:journals/algorithmica/KociumakaRS19} additionally presented an $\Oh(n^{1.5}\log^2 n)$-time 2-approximation algorithm for $k$-LCS for general $k$; it was improved in~\cite{DBLP:conf/cpm/GourdelKRS20} to run in $\Oh(n^{4/3+o(1)})$ time for strings over a constant-sized alphabet.

Analogously to Weiner's solution to the LCS problem via the suffix tree, the algorithm of Thankachan et al.~\cite{DBLP:journals/jcb/ThankachanAA16} builds upon the ideas of the $k$-errata tree, which was introduced by Cole et al.~\cite{DBLP:conf/stoc/ColeGL04} in their seminal paper for indexing a string of length $n$ with the aim of answering pattern matching queries with up to $k$ mismatches (that is, under the Hamming distance) or errors (that is, under the edit distance).
For constant $k$, the size of the $k$-errata tree is $\cO(n \log^k n)$.
Let us stress that computing a $k$-LCS using the $k$-errata tree directly is not straightforward, as opposed to computing an LCS using the suffix tree.

We show the following result, breaking through the $n \log^k n$ barrier for all constant integers $k>0$ irrespectively of the alphabet size. Recall that, in the word RAM model, the letters of $S$ and~$T$ can be renumbered in $\Oh(n \log\log n)$ time~\cite{DBLP:journals/jal/Han04} so that they belong to $[0\dd \sigma)$, where $\sigma$ is the total number of distinct letters in $S$ and $T$.

\begin{restatable}{theorem}{klcs}\label{thm:klcs}
    Given two strings $S$ and $T$ of total length $n$ and a constant integer $k>0$, the $k$-LCS problem can be solved in
    $\cO(n\log^{k-1/2} n)$ time using $\cO(n)$ space.
\end{restatable}

Notably, on the way to proving the above theorem, we improve upon~\cite{DBLP:conf/cpm/Charalampopoulos18}
by showing an $\cO(n + \frac{n}{\ell} \log^{k+1} n)$-time algorithm for $k$-LCS with $k=\Oh(1)$, where $\ell$ is the length of a $k$-LCS; see \cref{cor:CPMbetter}. (Our second summand is smaller by a $\sqrt{\ell}$ multiplicative factor compared to~\cite{DBLP:conf/cpm/Charalampopoulos18}.)

\subsection{Our Techniques}
At the heart of our approaches lies the following \problemtwo. Here, the length of the longest common prefix of two strings $U$ and $V$ is denoted by $\LCP(U, V)$; for a precise definition of compacted tries, see \cref{sec:prelim}.

\defproblem{\problemtwo}{%
A compacted trie $\T(\F)$ of $\F\sub \Sigma^*$
and two sets $\P,\Q\sub \F^2$ whose elements are represented by pairs of pointers to nodes of $\T(\F)$, with $|\P|,|\Q|,|\F| \le N$.}
 {$\maxPairLCP(\P,\Q)=\max\{\LCP(P_1,Q_1)+\LCP(P_2,Q_2) : (P_1,P_2)\in \P,(Q_1,Q_2)\in \Q\}$.}

This abstract problem was introduced in~\cite{DBLP:conf/cpm/Charalampopoulos18}. The solution encapsulated in the following lemma
is directly based on a technique that was used in~\cite{DBLP:conf/cpm/BrodalP00,DBLP:journals/tcs/CrochemoreIMS06} and then in~\cite{DBLP:journals/ipl/FlouriGKU15} to devise an $\Oh(n \log n)$-time solution for 1-LCS.

\begin{lemma}[{\cite[Lemma 3]{DBLP:conf/cpm/Charalampopoulos18}}]\label{lem:problem}
The \problemtwo can be solved in $\Oh(N\log N)$ time and $\Oh(N)$ space.\footnote{The original formulation of~\cite[Lemma 3]{DBLP:conf/cpm/Charalampopoulos18} does not include a claim about the space complexity. However, it can be readily verified that the underlying algorithm, described in~\cite{DBLP:journals/tcs/CrochemoreIMS06,DBLP:journals/ipl/FlouriGKU15}, uses only linear space.}
\end{lemma}

In particular, \cref{lem:problem} implies an $\Oh(n \log n)$-time algorithm for 1-LCS; see the following example.

\begin{example}
For a string $T$, let $T[i\dd j]$ denote the fragment of $T$ that starts at position $i$ (1-based) and ends at position $j$, and let $T^R$ be $T$ reversed. For two strings $S$ and $T$, we define the following sets of string pairs:
\begin{align*}
    \P&=\{((S[1\dd i-1])^R,S[i+1 \dd |S|])\,:\,i \in [1\dd|S|]\}\\
    \Q&=\{((T[1\dd i-1])^R,T[i+1 \dd |T|])\,:\,i \in [1\dd|T|]\}.
\end{align*}
Either the 1-LCS of $S$ and $T$ is simply their LCS, or the length of the 1-LCS of $S$ and $T$ equals $\maxPairLCP(\P,\Q)$. The underlying set $\F$ in the instance of \problemtwo is the set of suffixes of $S$, $T$, $S^R$, and $T^R$. The compacted trie $\T(\F)$ can be inferred in $\cO(n)$ time from the suffix tree of $S\#_1T\#_2S^R\#_3T^R$ for sentinel letters $\#_1,\#_2,\#_3$. Therefore, indeed \cref{lem:problem} implies an $\Oh(n \log n)$-time algorithm for 1-LCS.
\end{example}

In the algorithm underlying~\cref{lem:problem}, for each node $v$ of $\T(\F)$, we try to identify a pair of elements, one from $\P$ and one from $\Q$, whose first components are descendants of $v$ and the $\LCP$ of their second components is maximised. The algorithm traverses $\T(\F)$ bottom-up and uses mergeable height-balanced trees with $\Oh(N \log N)$ total merging time to store elements of pairs; see~\cite{DBLP:journals/jacm/BrownT79}.

No $o(N \log N)$ time solution to the \problemtwo is known, and devising such an algorithm seems difficult. The key ingredient of our algorithms is an efficient solution to a special case of the problem where we have bounds on the lengths of the involved strings. We say that a family $\mathcal{X}$ of string pairs is an \emph{$(\alpha,\beta)$-family} if each $(U,V) \in \mathcal{X}$ satisfies $|U| \le \alpha$ and $|V| \le \beta$.

\begin{restatable}{lemma}{lemmain}\label{lem:main}
An instance of the \problemtwo in which $\P$ and $\Q$ are $(\alpha,\beta)$-families can be solved in time 
$\cO(N(\alpha+\log N)(\log \beta +\sqrt{\log N})/ \log N)$ and space $\cO(N+N\alpha/\log N)$.
\end{restatable}

Observe that the algorithm of \cref{lem:main} works in $\Oh(N\sqrt{\log N})$ time if $\alpha = \Oh(\log N)$ and $\log \beta = \Oh(\sqrt{\log N})$.
The algorithm uses a wavelet tree (cf.\ \cite{DBLP:conf/soda/GrossiGV03}) of the first components of $\P \cup \Q$.

\paragraph{Solution to LCS}
For the LCS problem, we design three different algorithms and pick one depending on the length of the solution.
For short LCS ($\le \tfrac13\log_\sigma n$), we employ a simple tabulation technique.
For long LCS ($\ge \log^4 n$), we obtain a small instance of the \problemtwo by employing difference covers~\cite{DBLP:journals/ipl/ColbournL00,DBLP:journals/tocs/Maekawa85}, similarly to the solution of Charalampopoulos et al.~\cite{DBLP:conf/cpm/Charalampopoulos18} for the long $k$-LCS problem.
The main modification in this part is the usage of the sublinear longest common extension (LCE) data structure of Kempa and Kociumaka~\cite{DBLP:conf/stoc/KempaK19}. Our solutions for short LCS and long LCS work in $\Oh(n/\log_\sigma n)$ time.

As for medium-length LCS ($\ge \tfrac13\log_\sigma n$ and $\le \beta = 2^{\Oh(\sqrt{\log n})}$), let us first consider a case when the strings do not contain highly periodic fragments. In this case, we use the string synchronising sets of Kempa and Kociumaka~\cite{DBLP:conf/stoc/KempaK19} to select a set of $\cO(\frac{n}{\tau})$ anchors over $S$ and $T$, where $\tau=\Theta(\log_\sigma n)$, such that, for any common substring $U$ of $S$ and $T$ of length $\ell\geq 3\tau-1$, there exist occurrences $S[i^S \dd j^S]$ and $T[i^T \dd j^T]$ of $U$, for which we have anchors $a^S \in [i^S\dd j^S]$ and $a^T\in [i^T\dd j^T]$ with $a^S-i^S=a^T-i^T\le \tau$.
For each anchor $a$ in $S$, we add a string pair ${((S[a-\tau\dd a))^R,S[a \dd a+\beta))}$, possibly truncated at the ends of the string, to $\P$ (and similarly for $T$ and~$\Q$).
This lets us apply \cref{lem:main} with $N=\Oh(n/\tau)$, $\alpha=\Oh(\tau)$, and $\beta = 2^{\Oh(\sqrt{\log n})}$.
In the periodic case, we cannot guarantee that $a^S-i^S=a^T-i^T$ is small, but we can obtain a different set of anchors based on runs (maximal repetitions; cf.\ \cite{DBLP:conf/focs/KolpakovK99}) that yields multiple instances of the \problemtwo, which have extra structure leading to a linear-time solution.

\paragraph{Generalisation to many strings.}
We allow $\lambda=\Oh(\sqrt{\log n}/\log \log n)$ input strings.
Here, we also consider three cases based on the LCS length. The solution to short LCS is basically the same. For long LCS, we generalise difference covers from two to many strings, which might be a result of independent interest.
More precisely, we say that $(D,h)$ is a ($\lambda, d$)-\emph{cover} if $D \sub \mathbb{Z}_+$ is a set and $h : \Zp^\lambda \to [0\dd d)$ is an efficiently computable function such that, for any $i_1,\dots,i_\lambda\in \mathbb{Z}_+$ and $t\in [1\dd \lambda]$, we have $i_t+h(i_1,\dots,i_\lambda)\in D$.
The construction of ($\lambda, d$)-covers generalises the original construction behind $d$-covers~\cite{DBLP:journals/ipl/ColbournL00,DBLP:journals/tocs/Maekawa85}.

\begin{restatable}{theorem}{lambdadcover}\label{thm:lambda_dcover}
For every positive integers $\lambda\ge 2$ and $d$, there is a $(\lambda,d)$-cover $(D,h)$
such that the set $D\cap [1\dd n]$ is of size $\Oh(\lambda \cdot n/\sqrt[\lambda]{d})$.
After $\Oh(\log d)$-time initialization, the set $D\cap [1\dd n]$ can be constructed in $\Oh(\lambda \cdot n/\sqrt[\lambda]{d})$ time
and the function $h(i_1,\ldots,h_\lambda)$ can be evaluated in $\Oh(\lambda)$ time.
\end{restatable}

Further, we generalise the \problemtwo to $\lambda$ string families and obtain a counterpart of \cref{lem:problem} for this problem with a simpler but slower $\Oh(n \log^{\Oh(1)} n)$-time solution; this is still sufficient if the threshold for a long LCS is adjusted properly. We achieve this by opening the black box behind \cref{lem:problem} and extending a greedy property that stands behind it from two to~$\lambda$ input strings. 

For medium-length LCS, the application of string synchronising sets and runs does not require any essential changes. \cref{lem:main} can be carefully extended to $\lambda$ strings using said greedy property. Our solution to the special instances of the \problemtwo that arise due to periodicity in the $\lambda=\Oh(\sqrt{\log n}/\log \log n)$ input strings requires dynamic predecessor data structures. We keep the space $\Oh(n/\log_\sigma n)$ (and the solution deterministic) by applying a space-efficient version of van Emde Boas trees~\cite{DBLP:journals/ipl/Boas77} described in~\cite{DBLP:journals/jda/LagogiannisMT06,DBLP:journals/siamcomp/Willard00}.

\paragraph{Solution to $k$-LCS}
In this case, we also obtain a set of $\Oh(n/\ell)$ anchors, where $\ell$ is the length of a $k$-LCS. If the common substring is far from highly periodic, we use a string synchronising set for $\tau=\Theta(\ell)$, and otherwise we generate anchors using a technique of misperiods that was initially introduced for $k$-mismatch pattern matching~\cite{DBLP:conf/soda/BringmannWK19,DBLP:journals/jcss/Charalampopoulos21}.

Now, the families $\P,\Q$ need to consist not simply of substrings of $S$ and $T$ but rather of modified substrings generated by an approach that resembles $k$-errata trees~\cite{DBLP:conf/stoc/ColeGL04}. This requires combining the ideas of Thankachan et al.~\cite{DBLP:journals/jcb/ThankachanAA16}---who showed how to transform a family of strings into a family of modified strings such that $\LCE_k$ values for the former correspond to (maximal) $\LCE$ values for the latter---and (indirectly) Charalampopoulos et al.~\cite{DBLP:conf/cpm/Charalampopoulos18}---who lifted this idea to families consisting of pairs of strings that stem from ``sitting on anchors''.
Limiting the space usage of the algorithm to $\cO(n)$ proved to be quite technically challenging; note that, for example, this was not an issue in~\cite{DBLP:conf/cpm/Charalampopoulos18} where only a long enough LCS was sought.
Finally, we apply \cref{lem:problem} or \cref{lem:main} depending on the length $\ell$, breaking through the $n \log^k n$ barrier for $k$-LCS.

\subsection{Related Work} 

\subsubsection{Longest Common Substring}
Other results on the LCS problem include trade-offs between the time and the working space for computing an LCS of two strings~\cite{DBLP:conf/cpm/Nun0KK20,DBLP:conf/esa/KociumakaSV14,DBLP:conf/cpm/StarikovskayaV13}, internal LCS queries~\cite{DBLP:journals/algorithmica/AmirCPR20}, computing the LCS of two compressed strings~\cite{DBLP:conf/cccg/GagieGN13,DBLP:journals/tcs/MatsubaraIISNH09}, the dynamic maintenance of an LCS~\cite{DBLP:conf/cpm/AmirB18,DBLP:conf/spire/AmirCIPR17,DBLP:journals/algorithmica/AmirCPR20,DBLP:conf/icalp/Charalampopoulos20}, and the so-called heaviest induced ancestors queries that are a common theme of compressed and dynamic LCS~\cite{DBLP:journals/algorithmica/AbedinHGT22,DBLP:conf/cpm/Charalampopoulos23}.
The (long) LCS problem has also been recently considered in the quantum setting~\cite{DBLP:journals/corr/abs-2010-12122}, with almost optimal algorithms presented in~\cite{DBLP:journals/algorithmica/AkmalJ23,DBLP:conf/soda/JinN23}.

A version of the $k$-LCS problem in which it suffices to compute an LCS with \emph{approximately} $k$ mismatches has also been considered~\cite{DBLP:conf/cpm/GourdelKRS20,DBLP:journals/algorithmica/KociumakaRS19}.
The $k$-LCS problem has also been studied under edit distance. In this case, an $\Oh(n \log^{k} n)$-time algorithm is known~\cite{DBLP:conf/recomb/ThankachanACA18} (cf.~\cite{DBLP:conf/spire/AyadBCIP18} for an efficient average-case algorithm). It remains open if our approach can be used to improve upon the algorithm from~\cite{DBLP:conf/recomb/ThankachanACA18}. This would require substantial changes in the algorithm; a promising approach would be to use locked fragments (see \cite{DBLP:journals/siamcomp/ColeH02,DBLP:conf/focs/Charalampopoulos20}) instead of misperiods for computing anchors (see \cref{prp:anchors}).

\subsubsection{Stringology in the Word RAM Model}
A large body of work has been devoted to exploiting bit parallelism in the word RAM model for pattern matching~\cite{DBLP:journals/spe/Baeza-Yates89,DBLP:journals/spe/TarhioP97,DBLP:conf/cpm/NavarroR98,DBLP:conf/spire/Fredrikson02,DBLP:journals/ipl/Fredriksson03,DBLP:conf/wia/KleinB07a,DBLP:conf/cpm/Bille09,DBLP:conf/sofsem/CantoneF09,DBLP:conf/iwoca/Belazzougui10,DBLP:journals/jda/Bille11,DBLP:conf/cpm/BreslauerGG12,DBLP:journals/ipl/GiaquintaGF13,DBLP:journals/tcs/Ben-KikiBBGGW14}.
The problem of indexing a string of length $n$ over an alphabet $[0\dd \sigma)$ in the word RAM model, with the aim of efficiently answering pattern matching queries, has attracted significant attention~\cite{DBLP:conf/cpm/BilleGS17,DBLP:journals/algorithmica/BarbayCGNN14,DBLP:journals/talg/BelazzouguiN14,DBLP:journals/talg/BelazzouguiN15,DBLP:conf/stoc/KempaK19,DBLP:conf/cpm/MunroNN20,DBLP:conf/soda/MunroNN17,DBLP:journals/algorithmica/MunroNN20,DBLP:journals/siamcomp/NavarroN17}. Since, by definition, the suffix tree occupies $\Theta(n)$ space, alternative indexes have been sought. 
The state of the art is an index that occupies $\Oh(n \log \sigma /\log n)$ space and can be constructed in $\Oh(n \log \sigma / \sqrt{\log n})$ time; see Kempa and Kociumaka~\cite{DBLP:conf/stoc/KempaK19} and Munro et al.~\cite{DBLP:conf/cpm/MunroNN20}.  Interestingly, the running time of our algorithm (\cref{thm:sublinear}) matches the construction time of this index.
Note that, in contrast to suffix trees, such indexes \emph{cannot be used directly} for computing an LCS.
Intuitively, these indexes sample suffixes of the string to be indexed, and upon a pattern matching query, they treat separately the first $\cO(\log_\sigma n)$ letters of the pattern.

Bit parallelism was also considered in the context of other problems in stringology~\cite{DBLP:conf/cpm/Charalampopoulos22,DBLP:conf/mfcs/Charalampopoulos25,DBLP:conf/esa/BannaiE23,DBLP:conf/spire/Ellert23,DBLP:conf/focs/KempaK24}.

\subsubsection{\texorpdfstring{$k$}{k}-Mismatch Stringology Problems}
There are other problems in computational biology that allow for $k$ mismatches and for which $\Oh(n \log^{k} n)$-time algorithms are known. This includes the $k$-mismatch average common substring \cite{DBLP:journals/jcb/ThankachanAA16} and $k$-mappability \cite{DBLP:journals/algorithmica/Charalampopoulos22}.
It remains open if our techniques can be applied to these problems as well. Intuitively, the difficulty lies in that these problems seem to require computing maximal approximate $\LCE$ values for every suffix of the text separately.

As for $k$-mismatch indexing, for $k=\Oh(1)$, a $k$-errata tree occupies $\Oh(n \log^k n)$ space, can be constructed in $\Oh(n \log^k n)$ time, and supports pattern matching queries with at most $k$ mismatches in $\Oh(m+ \log^k n \log \log n+\textit{occ})$ time, where $m$ is the length of the pattern and $\textit{occ}$ is the number of the reported pattern occurrences. 
Other trade-offs for this problem, in which the product of space and query time is $\Omega(n\log^{2k} n)$, were shown in~\cite{DBLP:journals/jda/ChanLSTW11,DBLP:journals/jda/Tsur10} (a product $\Omega(n \log^{2k-1} n)$ of space and query time can be achieved for strings over constant-sized alphabets~\cite{DBLP:journals/algorithmica/ChanLSTW10}), and solutions with $\Oh(n)$ space but $\Omega(\min\{n,\sigma^km^{k-1}\})$-time queries were presented in~\cite{DBLP:journals/algorithmica/ChanLSTW10,DBLP:conf/cpm/Cobbs95,DBLP:journals/tcs/HuynhHLS06,DBLP:conf/cpm/Ukkonen93}. More efficient solutions for $k=1$ are known (see~\cite{DBLP:journals/algorithmica/Belazzougui15} and references therein).
Cohen-Addad et al.~\cite{DBLP:conf/soda/Cohen-AddadFS19} showed that, under SETH, for $k=\Theta(\log n)$ any indexing data structure that can be constructed in polynomial time cannot have $\Oh(n^{1-\delta})$ query time for any $\delta > 0$. They also showed that in the pointer machine model, for the reporting version of the problem with $k=o(\log n)$, exponential dependency on $k$ cannot be avoided in both the space and the query time. 
We hope that our techniques can fuel further progress in $k$-mismatch indexing.

\subsection{Conclusions}
We present $\cO(n \log \sigma / \sqrt{\log n})$-time and $\Oh(n \log^{k-0.5} n)$-time (assuming $k=\Oh(1)$) algorithms for computing the LCS and the $k$-LCS of two strings of total length $n$ over an alphabet of size $\sigma$, respectively. The bottleneck of both algorithms is the construction of a wavelet tree~\cite{DBLP:conf/soda/BabenkoGKS15,DBLP:journals/tcs/MunroNV16}. 
Improving the wavelet tree construction time could be a hard task; in particular, it would also improve upon the $\Oh(n \sqrt{\log n})$-time algorithm by Chan and Pătrașcu \cite{DBLP:conf/soda/ChanP10} for counting inversions of a permutation; cf.~\cite{DBLP:conf/stoc/KempaK19}.
As already mentioned, there is a matching conditional lower bound for LCS; see~\cite{DBLP:conf/stoc/KempaK25}.
However, in the case of the $k$-LCS problem, we are not aware of such a lower bound.

\subsection{Preliminary Version}
A preliminary version of this work was presented in~\cite{DBLP:conf/esa/Charalampopoulos21}. The current full version extends the sublinear-time LCS computation to many input strings, has (significantly) improved exposition, and includes all the proofs.

\section{Preliminaries}\label{sec:prelim}
\paragraph{Strings.} Let $T=T[1]T[2]\cdots T[n]$ be a \emph{string} (or \emph{text}) of length $n=|T|$ over an alphabet $\Sigma=[0\dd \sigma)$. The elements of $\Sigma$ are called \textit{letters}. By $\Sigma^*$ and $\Sigma^n$ we denote the sets of all finite strings and of all length-$n$ strings, respectively.

By $\varepsilon$ we denote the \emph{empty string}. For two positions $i$ and $j$ of $T$, we denote by $T[i\dd j]$ the \emph{fragment} of $T$ that starts at position $i$ and ends at position $j$ (the fragment is empty if $i>j$). A fragment of $T$ is represented using $\cO(1)$ space by specifying the indices $i$ and $j$. We define $T[i \dd j)=T[i \dd j-1]$ and $T(i \dd j]=T[i+1 \dd j]$. The fragment $T[i\dd j]$ is an \emph{occurrence} of the underlying \emph{substring} $P=T[i]\cdots T[j]$. We say that $P$ occurs at \emph{position} $i$ in $T$. A \emph{prefix} of $T$ is a fragment of $T$ of the form $T[1\dd j]$ and a \emph{suffix} of $T$ is a fragment of $T$ of the form $T[i\dd n]$.
For strings $U_1,U_2, \ldots, U_k$, by $\LCP(U_1,U_2,\ldots,U_k)$ we denote the length of the longest common prefix of $U_1,U_2, \ldots, U_k$. We denote the \emph{reverse string} of $T$ by $T^R$, i.e., $T^R=T[n]T[n-1]\cdots T[1]$. By $UV$ we denote the \emph{concatenation} of two strings $U$ and $V$, i.e., $UV=U[1]U[2]\cdots U[|U|]V[1]V[2]\cdots V[|V|]$.

The lexicographic order on strings is denoted by $\le$. We use the following simple fact that has extensive applications, e.g., in computations using suffix arrays~\cite{DBLP:journals/siamcomp/ManberM93}.

\begin{fact}\label{fact:trivial}
For strings $U_1 \le U_2 \le U_3$, $\LCP(U_1,U_3) = \min(\LCP(U_1,U_2),\LCP(U_2,U_3))$.
\end{fact} 

A positive integer $p$ is called a \emph{period} of a string $T$ if $T[i] = T[i + p]$ for all $i \in [1\dd  |T| - p]$.  We refer to the smallest period as \emph{the period} of the string, and denote it by $\per(T)$. A string $T$ is called \emph{periodic} if $\per(T)\leq |T|/2$ and \emph{aperiodic} otherwise. A \emph{run} in $T$ is a periodic fragment that cannot be extended (to the left nor to the right) without an increase of its smallest period. All runs in a string can be computed in linear time~\cite{DBLP:journals/siamcomp/BannaiIINTT17,DBLP:conf/focs/KolpakovK99}, even over arbitrary ordered alphabets~\cite{DBLP:conf/icalp/Ellert021}.

\begin{lemma}[Periodicity Lemma (weak version)~\cite{FW:periodicity-lemma}] 
If a string $S$ has periods $p$ and $q$ such that $p+q \leq |S|$, then $\gcd(p, q)$ is also a period of $S$.\label{lem:perlemma}
\end{lemma}

 \paragraph{Tries.}~Let $\M$ be a finite set containing $m>0$ strings over $\Sigma$. The \emph{trie} of $\M$, denoted by~$\R(\M)$, contains a node for every distinct prefix of a string in $\M$; the root node is $\varepsilon$; the set of terminal nodes is $\M$; and edges are of the form $(u,\alpha,u\alpha)$, where $u$ and $u\alpha$ are nodes and $\alpha\in\Sigma$. The \emph{compacted trie} of $\M$, denoted by $\TT(\M)$, contains the root, the branching nodes, and the terminal nodes of $\R(\M)$. 
Each maximal branchless path segment from $\R(\M)$ is replaced by a single edge, and a fragment of a string $M\in \M$ is used to represent the label of this edge in $\Oh(1)$ space. The best-known example of a compacted trie is the suffix tree~\cite{DBLP:conf/focs/Weiner73}. Throughout our algorithms, $\M$ always consists of a set of substrings or modified substrings with $k=\Oh(1)$ modifications (see~\cref{sec:kLCP} for a definition) of a reference string. 
The value $\val(u)$ of a node $u$ is the concatenation of labels of edges on the path from the root to $u$, and the \emph{string-depth} of $u$ is the length of $\val(u)$. The size of $\TT(\M)$ is $\Oh(m)$.
We use the following well-known construction (cf.~\cite{DBLP:books/daglib/0020103}).

\begin{lemma}\label{lem:ctrie}
Given a sorted list of $N$ strings and the longest common prefixes between pairs of consecutive strings in the list, the compacted trie of the strings can be constructed in $\Oh(N)$ time.
\end{lemma}

 \paragraph{Packed strings.}~We assume the unit-cost word RAM model with word size $w= \Theta(\log n)$ and a standard instruction set including arithmetic operations, bitwise Boolean operations, and shifts. We count the space complexity of our algorithms in machine words.
 The \emph{packed representation} of a string $T \in [0\dd \sigma)^n$ is a list obtained by storing $\Theta(\log_\sigma n)$ letters per machine word thus representing~$T$ in $\Oh(n/\log_{\sigma}n)=\Oh(n \log \sigma/\log n)$ machine words. If $T$ is given in the packed representation, we simply say that $T$ is a \emph{packed string}.

\paragraph{Computations on packed data.}~We use the following abstract proposition to conveniently perform sequential computations on data given in a packed form. We consider any algorithm that interacts via $\Oh(1)$ input and output streams. In the proof, for every possible starting state of the algorithm's memory, we memoise the result of the algorithm after any $\tau$ single-bit read or write instructions. A full proof is given in \cref{app:prpstream}.

\begin{restatable}{proposition}{prpstream}\label{prp:streaming}
Let $\mathsf{A}$ be a deterministic streaming algorithm that reads $\Oh(1)$ input streams, writes $\Oh(1)$ output streams, uses $s$ bits of working space, and executes at most $t$ instructions (each in $\Oh(1)$ time) between subsequent read/write instructions.
The streams are represented in a packed form with word size $w$.
For every integer $\tau \in [1\dd w]$, if $s \le w$, after an $\Oh(2^{\tau+s}(\tau t+s))$-time preprocessing, any execution of $\mathsf{A}$ can be simulated in $\Oh(1+L/\tau)$ time, where $L$ is the total length (in bits) of all input and output streams.
\end{restatable}

For example, we will be using the proposition as follows. If for some $n$ the algorithm from the proposition uses $s=o(\log n)$ bits of space and executes $t=\Oh(1)$ instructions between subsequent reads and writes, then after $o(n)$ preprocessing, it can be simulated in $\Oh(1+L/\log n)$ time (for $\tau=\floor{\frac12 \log n}$).

\paragraph{String synchronising sets.}~Our solution uses the string synchronising sets that were introduced by Kempa and Kociumaka~\cite{DBLP:conf/stoc/KempaK19}.
Informally, in the simpler case that $T$ is cube-free, a $\tau$-synchronising set of $T$ is a small set of \emph{synchronising positions} in $T$ such that each length-$\tau$ fragment of $T$ contains at least one synchronising position, and the leftmost synchronising positions within two sufficiently long matching fragments of $T$ are consistent.

Formally, for a length-$n$ string $T$ and a positive integer $\tau \leq \frac12n$, a set $A\subseteq [1\dd  n-2\tau+1]$  is a \emph{$\tau$-synchronising set} of $T$ if it satisfies the following two conditions (see \cref{fig:SSS}):
\begin{enumerate}
    \item\label{item1} Consistency: If $T[i\dd i+2\tau)$ matches $T[j\dd j+2\tau)$, then $i\in A$ if and only if $j\in A$.
    \item\label{item2} Density: For $i\in [1\dd  n-3\tau +2]$, $A\cap [i\dd i+\tau)=\emptyset$ if and only if $\per(T[i\dd i+3\tau -2 ])\leq  \frac13\tau$.
\end{enumerate}

\begin{figure}[htpb]
\centering
\begin{tikzpicture}
  \begin{scope}
    \draw (12,0.3*0.6) node[right] {$\tau=3$};
    \foreach \x in {0,2,5,11,12,15}{
      \filldraw[xshift=0.25cm] (\x*0.5+0.25,-0.1) circle (0.07cm);
    }
    \foreach \x in {1,...,22}{
      \filldraw[xshift=0.25cm] (\x*0.5-0.25,-0.1) node[below] {\small \x};
    }
    \draw[xshift=0.5cm,yshift=1.2cm,red] (0.5*0.5,0.1*0.6) rectangle node {$T[2 \dd 8)$} (6.5*0.5,0.8*0.6);
    \draw[red,yshift=1.2cm,xshift=8cm] (0.5*0.5,0.1*0.6) rectangle node {$T[17 \dd 23)$} (6.5*0.5,0.8*0.6);
    \draw[green!50!black,yshift=0.6cm] (0.5*0.5,0.1*0.6) rectangle node {$T[1 \dd 7)$} (6.5*0.5,0.8*0.6);
    \draw[green!50!black,xshift=7.5cm,yshift=0.6cm] (0.5*0.5,0.1*0.6) rectangle node {$T[16 \dd 22)$} (6.5*0.5,0.8*0.6);
    \foreach \x/\c in {1/a,2/b,3/a,4/a,5/b,6/c,7/a,8/a,9/a,10/a,11/a,12/a,13/a,14/a,15/a,16/a,17/b,18/a,19/a,20/b,21/c,22/a}{
      \draw (\x*0.5,0) node[above] {\texttt{\c}};
    }    
    \draw[snake=brace,xshift=0.75cm] (1.4,-0.7*0.6-0.1) -- node[below] {\textcolor{green!50!black}{Exists synch.}} (0.1,-0.7*0.6-0.1);
    \draw[snake=brace,xshift=3.25cm] (2.4,-0.7*0.6-0.1) -- node[below] {\textcolor{red}{No synch.}} (0.1,-0.7*0.6-0.1);
    \foreach \x in {3,3.5,...,7.5}{
      \draw[xshift=\x cm] (0.25,0.6*0.6) .. controls (0.4,0.9*0.6) and (0.6,0.9*0.6) .. (0.75,0.6*0.6);
    }
    \draw (5.5,0.9*0.6) node[above] {$\tau$-run};
  \end{scope}
  \end{tikzpicture}
  \caption{An example of a $\tau$-synchronising set ($\tau=3$) of this string is $A=\{1,3,6,12,13,16\}$. Fragments $T[1 \dd 7)$ and $T[16 \dd 22)$ match and thus $1$, $16$ are both in $A$. Fragments $T[2 \dd 8)$ and $T[17 \dd 23)$ match and thus $2$, $17$ are both \emph{not} in $A$. Among every 3 consecutive positions (that are sufficiently far from the end of the string) there is a synchronising position, except for the positions $7,\ldots,11$ which imply a long fragment with period $\frac13\tau=1$ (a so-called $\tau$-run).}\label{fig:SSS}
\end{figure}

\begin{theorem}[{\cite[Proposition 8.10, Theorem 8.11]{DBLP:conf/stoc/KempaK19}}]\label{thm:synch_packed}
For a string $T \in [0\dd \sigma)^n$ with $\sigma=n^{\Oh(1)}$ and $\tau\le \frac12n$, there exists a $\tau$-synchronising set of size $\cO(n/\tau)$ that can be constructed in $\Oh(n)$ time or, if $\tau \leq \frac15 \log_\sigma n$, in $\cO(n/\tau)$ time if $T$ is given in a packed representation.
\end{theorem}

As in~\cite{DBLP:conf/stoc/KempaK19}, for a $\tau$-synchronising set $A$,
let 
$\suc_A(i) := \min\{j \in A \cup \{n - 2\tau + 2\} : j \geq i\}$.

\begin{lemma}[{\cite[Fact 3.2]{DBLP:conf/stoc/KempaK19}}]\label{lem:fact3.2}
If $p=\per(T[i\dd i+3\tau -2 ])\leq  \frac13\tau$, then $T[i \dd \suc_A(i)+2\tau-1)$ is the longest prefix of $T[i\dd|T|]$ with period $p$.
\end{lemma}

\begin{lemma}[{\cite[Fact 3.3]{DBLP:conf/stoc/KempaK19}}]\label{lem:synch}
If a string $U$ with $|U|\geq 3\tau-1$ and $\per(U)>\frac13\tau$ occurs at positions $i$ and $j$ in $T$, then $\suc_A(i) - i = \suc_A(j) - j \leq |U| - 2\tau$.
\end{lemma}

A \emph{$\tau$-run} $R$ is a run of length at least $3\tau-1$ with period at most $\frac13\tau$. The \emph{Lyndon root} of $R$ is the lexicographically smallest cyclic shift of $R[1\dd \per(R)]$.
The proof of the following lemma resembles an argument given in~\cite[Section 6.1.2]{DBLP:conf/stoc/KempaK19}; it is presented in \cref{app:taurons} for completeness.

\begin{restatable}{lemma}{lemtauruns}\label{lem:tauruns}
For a positive integer $\tau$, a string $T \in [0\dd \sigma)^n$ contains $\Oh(n/\tau)$ $\tau$-runs. 
Moreover, if $\tau\leq \frac{1}{4} \log_{\sigma} n$, given a packed representation of $T$, we can compute all $\tau$-runs in $T$ and group them by their Lyndon roots in $\Oh(n/\tau)$ time.
Within the same complexities, for each $\tau$-run, we can compute the two leftmost occurrences of its Lyndon root.
\end{restatable}

\begin{theorem}[{\cite[Theorem 4.3]{DBLP:conf/stoc/KempaK19}}]
Given a packed representation of a string $T \in [0\dd \sigma)^n$ and a $\tau$-synchronising set $A$ of $T$ of size $\Oh(n/\tau)$ for $\tau=\Oh(\log_\sigma n)$, we can compute in $\Oh(n/\tau)$ time the lexicographic order of all suffixes of $T$ starting at positions in $A$.\label{thm:suffixsort}
\end{theorem}

We often want to preprocess $T$ to be able to answer queries of the form $\LCP(T[i\dd n], T[j\dd n])$.
For this case, there exists an optimal data structure that is based on synchronising sets.

\begin{theorem}[{\cite[Theorem 5.4]{DBLP:conf/stoc/KempaK19}}]\label{thm:lce}
Given a packed representation of a string $T \in [0\dd \sigma)^n$, $\LCE$ queries on suffixes of $T$ can be answered in $\cO(1)$ time after $\cO(n/ \log_\sigma n)$-time preprocessing.
\end{theorem}

\section{Sublinear-Time LCS}\label{sec:0LCS}

We provide different solutions depending on the length $\ell$ of an LCS.~\cref{lem:long,lem:short,lem:medium} directly yield~\cref{thm:sublinear}. 

\subsection{Solutions for Short and Long LCS}

\begin{lemma}[Short LCS]\label{lem:short}
The LCS problem can be solved in $\cO(n/\log_\sigma n)$ time if $\ell\leq \frac13 \log_\sigma n$.
\end{lemma}
\begin{proof}
Let us set $m=\lfloor\frac13 \log_\sigma n\rfloor$. 
Given the packed representations of strings $S$ and $T$, we can encode any length-$2m$ substring of the strings as a number in $[0\dd \floor{n^{2/3}})$ in $\Oh(1)$ time.
Sorting $\Oh(n/m)$ such numbers via Counting Sort takes $\Oh(n/m+n^{2/3})$ time.
We use the so-called standard trick: we split both $S$ and $T$ into $\Oh(n/m)$ fragments, each of length $2m$ (perhaps apart from the last one), starting at positions equivalent to $1 \pmod m$. For each of the strings, we obtain at most $\sigma^{2m}=\cO(n^{2/3})$ \emph{distinct} substrings corresponding to these fragments.
For each of the two strings, we can compute the distinct such substrings in $\Oh(n/m+n^{2/3})$ time by sorting all considered substrings.

\begin{figure}[htpb]
    \centering
    \begin{tikzpicture}[xscale=0.5,yscale=0.9]
  \begin{scope}
    \draw (0,0) rectangle (9.5,0.5);
    \draw (0,0.25) node[left] {$S$};
    \foreach \x in {0,2,4,6}{
      \draw[xshift = \x cm] (0,0.5) rectangle (2,1);
    }
    \draw[xshift = 8 cm] (0,0.5) rectangle (1.5,1);
    \foreach \x in {1,3,5,7}{
      \draw[xshift = \x cm] (0,1) rectangle (2,1.5);
    }
    \draw[latex-latex,densely dotted] (1,1.7) -- node[above] {$2m$} (3,1.7);
      \filldraw[blue!50!white] (0.2,0) rectangle (1.2,0.5);
      \draw[blue] (0.2,0) rectangle (1.2,0.5);
      \draw (1,0.75) node {\small $X_1$};
      \draw (2,1.25) node {\small $X_2$};
      \draw (3,0.75) node {\small \textcolor{black!50!white}{$X_2$}};
      \draw (4,1.25) node {\small \textcolor{black!50!white}{$X_2$}};
      \draw (5,0.75) node {\small $X_3$};
      \draw (6,1.25) node {\small \textcolor{black!50!white}{$X_1$}};
      \draw (7,0.75) node {\small \textcolor{black!50!white}{$X_3$}};
      \draw (8,1.25) node {\small \textcolor{black!50!white}{$X_2$}};
      \draw (8.75,0.75) node {\small $X_4$};
  \end{scope}
  \begin{scope}[yshift=-2.2cm]
    \draw (0,0) rectangle (9.8,0.5);
    \draw (0,0.25) node[left] {$T$};
    \foreach \x in {0,2,4,6}{
      \draw[xshift = \x cm] (0,0.5) rectangle (2,1);
    }
    \draw[xshift = 8 cm] (0,0.5) rectangle (1.8,1);
    \foreach \x in {1,3,5,7}{
      \draw[xshift = \x cm] (0,1) rectangle (2,1.5);
    }
      \filldraw[blue!50!white] (5.4,0) rectangle (6.4,0.5);
      \draw[blue] (5.4,0) rectangle (6.4,0.5);
      \draw (1,0.75) node {\small $Y_1$};
      \draw (2,1.25) node {\small \textcolor{black!50!white}{$Y_1$}};
      \draw (3,0.75) node {\small \textcolor{black!50!white}{$Y_1$}};
      \draw (4,1.25) node {\small $Y_2$};
      \draw (5,0.75) node {\small $Y_3$};
      \draw (6,1.25) node {\small \textcolor{black!50!white}{$Y_2$}};
      \draw (7,0.75) node {\small $Y_4$};
      \draw (8,1.25) node {\small \textcolor{black!50!white}{$Y_3$}};
      \draw (8.9,0.75) node {\small $Y_5$};
  \end{scope}

    \begin{scope}[xshift=13cm]
    \draw (0,0.25) node[left] {$X$};
    \draw (0,0) rectangle (2,0.5);
    \draw (1,0.5) node[above] {\small $X_1$};
    \filldraw[blue!50!white] (0.2,0) rectangle (1.2,0.5);
      \draw[blue] (0.2,0) rectangle (1.2,0.5);
    \draw (3,0) rectangle (5,0.5);
    \draw (4,0.5) node[above] {\small $X_2$};
    \draw (6,0) rectangle (8,0.5);
    \draw (7,0.5) node[above] {\small $X_3$};
    \draw (9,0) rectangle (10.5,0.5);
    \draw (9.75,0.5) node[above] {\small $X_4$};
    \draw (2.5,0.25) node {\footnotesize $\#_1$};
    \draw (5.5,0.25) node {\footnotesize $\#_2$};
    \draw (8.5,0.25) node {\footnotesize $\#_3$};
  \end{scope}
  \begin{scope}[yshift=-2.2cm,xshift=13cm]
    \draw (0,0.25) node[left] {$Y$};
    \draw (0,0) rectangle (2,0.5);
    \draw (1,0.5) node[above] {\small $Y_1$};
    \draw (3,0) rectangle (5,0.5);
    \draw (4,0.5) node[above] {\small $Y_2$};
    \filldraw[blue!50!white] (3.3,0) rectangle (4.3,0.5);
    \draw[blue] (3.3,0) rectangle (4.3,0.5);
    \draw (6,0) rectangle (8,0.5);
    \draw (7,0.5) node[above] {\small $Y_3$};
    \draw (9,0) rectangle (11,0.5);
    \draw (10,0.5) node[above] {\small $Y_4$};
    \draw (12,0) rectangle (13.8,0.5);
    \draw (12.9,0.5) node[above] {\small $Y_5$};
    \draw (2.5,0.25) node {\footnotesize $\$_1$};
    \draw (5.5,0.25) node {\footnotesize $\$_2$};
    \draw (8.5,0.25) node {\footnotesize $\$_3$};
    \draw (11.5,0.25) node {\footnotesize $\$_4$};
  \end{scope}
    \end{tikzpicture}
    \caption{Left: partitioning of $S$ and $T$ into fragments of length (up to) $2m$. All distinct substrings in $S$ and in $T$ are $X_1,\ldots,X_4$ and $Y_1,\ldots,Y_5$, respectively. The LCS of length $m$ (in blue) is a substring of $X_1$ in $S$ and of $Y_2$ in~$T$. Right: strings $X$ and $Y$; their LCS (shown in blue) is the same as the LCS of $S$ and $T$.}
    \label{fig:shortLCS}
\end{figure}

Let $X_1,\dots,X_p$ and $Y_1,\dots,Y_q$ be the resulting distinct substrings.
Then, the problem can be reduced to computing the LCS of $X=X_1\,\#_1\,X_2\,\#_2\,\dots \#_{p-1}\,X_p$ and $Y=Y_1\,\$_1\,Y_2\,\$_2\,\dots \$_{q-1}\,Y_q$ for distinct letters $\#_1,\dots,\#_{p-1},\$_1,\dots,\$_{q-1}$; see \cref{fig:shortLCS}. Strings $X$ and $Y$ have length $\Oh(n^{2/3}m)$ and their LCS
can be computed in linear time by constructing the suffix tree of $X\#Y$, where $\#$ is a letter that does not occur in $XY$~\cite{DBLP:conf/focs/Farach97}.
Thus, the overall time complexity is $\Oh(n/m+n^{2/3}m) = \cO(n/\log_\sigma n)$.
\end{proof}

The proof of the following lemma, for the case when an LCS is long, i.e., of length $\ell=\Omega(\frac{\log^4 n}{\log^2 \sigma})$, uses difference covers and the $\cO(N \log N)$-time solution to the \problemtwo.
This proof closely follows~\cite{DBLP:conf/cpm/Charalampopoulos18}.

\begin{lemma}[Long LCS]\label{lem:long}
The LCS problem can be solved in $\cO(n/\log_\sigma n)$ time if $\ell=\Omega(\frac{\log^4 n}{\log^2 \sigma})$.
\end{lemma}

Before proceeding to the proof, we need to introduce difference covers. We say that $(D,h)$ is a $d$-\emph{cover} if $D\sub \mathbb{Z}_+$ and $h$ is a (constant-time computable) function such that for positive integers $i,j$ we have $0\le h(i,j)< d$ and $i+h(i,j), j+h(i,j)\in D$.
The following fact synthesises a well-known construction implicitly used in~\cite{BurkhardtEtAl2003}, for example.
\begin{theorem}[\cite{DBLP:journals/ipl/ColbournL00,DBLP:journals/tocs/Maekawa85}]\label{thm:dcover}
For each positive integer $d$ there is a $d$-cover $(D,h)$
such that $D \cap [1\dd n]$ is of size $\Oh(\frac{n}{\sqrt{d}})$ and can be constructed in $\Oh(\frac{n}{\sqrt{d}})$ time.
\end{theorem}

\begin{proof}[Proof of~\cref{lem:long}]
Let us assume that the answer to LCS is of length $\ell \geq d$ for some parameter~$d$. 
We first use \cref{thm:dcover} to compute in $\cO(n/\sqrt{d})$ time a set $D=D' \cap [1\dd n]$ such that $(D',h)$ is a $d$-cover for a function $h$.
Let us now consider a position $i$ from $S$ and a position $j$ from $T$ such that $S[i \dd i+\ell)=T[j \dd j+\ell)$. 
We have that $0 \le h(i,j) < d \le \ell$, so for $i'=i+h(i,j)$, $j'=j+h(i,j)$ we have:
$$S[i \dd i')=T[j \dd j'),\ S[i' \dd i+\ell)=T[j' \dd j+\ell),\text{ and }i',j' \in D.$$
In particular (see \cref{fig:dcover} for an example),
$$\LCP((S[1 \dd i'))^R,(T[1 \dd j'))^R)+\LCP(S[i' \dd |S|], T[j' \dd |T|])=\ell.$$

\begin{figure}[htpb]
\centering
\begin{tikzpicture}
\begin{scope}
  \begin{scope}
    \foreach \x/\c in {1/c,2/b,8/a,9/b,10/c,11/a,12/b,13/c,14/b,15/a}{
      \draw (\x*0.35,0) node[above] {\texttt{\c}};
    }
    \foreach \x/\c in {3/b,4/a,5/a,6/b,7/c}{
      \draw (\x*0.35,0) node[above] {\textcolor{blue}{\texttt{\c}}};
    }    
    \draw[yshift=0.12cm] (1.25*0.7,0) -- (1.25*0.7,-0.2) -- (3.75*0.7,-0.2) -- (3.75*0.7,0);
    \foreach \x in {1,...,15}{
        \draw (\x*0.35,-0.2) node[below] {\small \x};
    }
    \foreach \x in {1,2,4,6,7,9,11,12,14}{
        \filldraw (\x*0.35,-0.2) circle (0.05cm);
    }
    \filldraw[red] (6*0.35,-0.2) circle (0.07cm);
    \draw[-latex,red] (2.8*0.35,0.9) -- node[above] {$h(3,4)=3$} (6.2*0.35,0.9);
    \draw (3*0.35,0.6) node {$i$};
    \draw (6*0.35,0.6) node {$i'$};
    \draw[-latex,blue] (5.8*0.35,-0.7) -- node[below] {$\LCP$} (7.2*0.35,-0.7);
    \draw[-latex,blue] (5.2*0.35,-0.7) -- node[below] {$\LCP^R$} (2.8*0.35,-0.7);
  \end{scope}
  \begin{scope}[xshift=7.5cm]
    \foreach \x/\c in {1/d,2/a,3/a,9/b,10/b,11/a}{
      \draw (\x*0.35,0) node[above] {\texttt{\c}};
    }    
    \foreach \x/\c in {4/b,5/a,6/a,7/b,8/c}{
      \draw (\x*0.35,0) node[above] {\textcolor{blue}{\texttt{\c}}};
    }    
    \draw[xshift=0.35cm,yshift=0.2cm] (1.25*0.7,0) -- (1.25*0.7,-0.2) -- (3.75*0.7,-0.2) -- (3.75*0.7,0);
    \foreach \x in {1,...,11}{
        \draw (\x*0.35,-0.2) node[below] {\small \x};
    }
    \foreach \x in {1,2,4,6,7,9,11}{
        \filldraw (\x*0.35,-0.2) circle (0.05cm);
    }
    \filldraw[red] (7*0.35,-0.2) circle (0.07cm);
    \draw[-latex,red] (3.8*0.35,0.9) -- node[above] {$h(3,4)=3$} (7.2*0.35,0.9);
    \draw (4*0.35,0.6) node {$j$};
    \draw (7*0.35,0.6) node {$j'$};
    \draw[-latex,blue] (6.8*0.35,-0.7) -- node[below] {$\LCP$} (8.2*0.35,-0.7);
    \draw[-latex,blue] (6.2*0.35,-0.7) -- node[below] {$\LCP^R$} (3.8*0.35,-0.7);
  \end{scope}
\end{scope}    
\end{tikzpicture}
\caption{Strings from \cref{fig:0-1-LCS} with their LCS $S[3 \dd 8)=T[4 \dd 9)$ of length $\ell=5$. The dots correspond to elements of a 5-cover $(D,h)$ where $D=\{1,2,4,\,6,7,9,\,11,12,14,\ldots\}$. For $i=3$, $j=4$, we have $i'=i+h(i,j)=6$, $j'=j+h(i,j)=7$, and $\LCP((S[1 \dd i'))^R,(T[1 \dd j'))^R)+\LCP(S[i' \dd |S|], T[j' \dd |T|])=\ell=5$.}\label{fig:dcover}
\end{figure}

Hence, if we want to find an LCS whose length is at least $d$, we can use the elements of the $d$-cover as anchors between occurrences of the sought LCS, the length of which equals
\begin{equation}\label{eq:max}
\max_{i,j \in D} \{\LCP((S[1 \dd i))^R,(T[1 \dd j))^R)+\LCP(S[i \dd |S|], T[j \dd |T|])\}.
\end{equation}

As also done in~\cite{DBLP:conf/cpm/Charalampopoulos18}, we set:
\begin{itemize}
    \item $\P:= \{((S[1 \dd i))^R, S[i \dd |S|]) : i \in D\}$, $\Q:= \{((T[1 \dd i))^R, T[i \dd |T|]) : i \in D\}$, and
    \item $\F := \{U : (U,V) \in \P \cup \Q \text{ or } (V,U) \in \P \cup \Q \text{ for some string }V\}$.
\end{itemize}

It is then readily verified that, after building the compacted trie $\TT(\F)$, evaluating~\cref{eq:max} reduces to solving the \problemtwo for $\P$, $\Q$, and $\F$ with $N=\cO(|D \cap [1\dd n]|)=\cO(n/\sqrt{d})$.
In order to build $\TT(\F)$, due to~\cref{lem:ctrie}, it suffices to sort the elements of~$\F$ lexicographically and answer $\LCE$ queries between consecutive elements in the resulting sorted list.
To this end, we first preprocess $S \#_1 S^R \#_2 T \#_3 T^R$, where $\#_i \not\in \Sigma$ for $i \in \{1,2,3\}$ are distinct letters, in $\cO(n/\log_\sigma n)$ time as in~\cref{thm:lce}, in order to allow for $\cO(1)$-time $\LCE$ queries.
Next, we employ merge sort in order to sort the elements of~$\F$, performing each comparison using an $\LCE$ query. This step requires $\cO(|\F| \log |\F|)$ time. Then, answering $\LCE$ queries for consecutive elements requires $\cO(|\F|)$ time in total.
We finally employ~\cref{lem:problem} to solve the instance of \problemtwo.
We have that each of $\P$, $\Q$, and $\F$ is of size at most $N=\cO(n/\sqrt{d})$. The overall time complexity is 
\[\cO\left(N \log N  + n/\log_\sigma n\right) = \cO(n \log n/\sqrt{d} + n/\log_\sigma n),\]
which is $\cO(n/\log_\sigma n)$ if $d=\Omega(\frac{\log^4 n}{\log^2 \sigma})$.
\end{proof}

\begin{remark}
The only modification compared to the solution of~\cite{DBLP:conf/cpm/Charalampopoulos18} lies in the construction of $\TT(\F)$. In~\cite{DBLP:conf/cpm/Charalampopoulos18}, $\TT(\F)$ is extracted from the generalised suffix tree of strings $S$, $T$, $S^R$, and~$T^R$, which we cannot afford to construct, as its construction requires $\Omega(n)$ time. Here we employ \cref{thm:lce} instead.
\end{remark}

\subsection{Solution for Medium-Length LCS}\label{sec:medium}
\newcommand{\II}{\mathit{II}}
\newcommand{\III}{\mathit{III}}
We now give a solution to the LCS problem for $\ell$ such that $\frac13\log_\sigma n \le \ell \le 2^{\sqrt{\log n}}$. We first construct three subsets of positions in $S\$T$, where $\$ \not \in \Sigma$, of size $\Oh(n/\log_\sigma n)$ as follows.
For $\tau=\floor{\frac{1}{9}\log_\sigma n}$, let $A_I$ be a $\tau$-synchronising set of $S\$T$. 
For each $\tau$-run in $S\$T$, we insert to $A_{\II}$ the starting positions of the first two occurrences of the Lyndon root of the $\tau$-run and to $A_{\III}$ the last position of the $\tau$-run.
The elements of $A_\II$ and $A_\III$ store the $\tau$-run they originate from.
Finally, we denote $A^S_j=A_j \cap [1\dd |S|]$ and $A^T_j=\{a-|S|-1\,:\,a \in A_j,\,a > |S|+1\}$ for $j\in \{I,\II,\III\}$. The following lemma shows that there exists an LCS of $S$ and $T$ for which $A_I \cup A_\II \cup A_\III$ is a set of \emph{anchors} that satisfies certain distance requirements.

\begin{lemma}\label{lem:anchors}
If an LCS of $S$ and $T$ has length $\ell \ge 3\tau$, then there exist positions $i^S \in [1\dd |S|]$, $i^T \in [1\dd |T|]$, a shift $\delta \in [0\dd  \ell)$, and $j \in \{I,\II,\III\}$ such that $S[i^S \dd i^S+\ell) = T[i^T \dd i^T+\ell)$, $i^S+\delta\in A^S_j$, $i^T+\delta\in A^T_j$, and
\begin{itemize}
    \item if $j=I$, then we can choose a $\delta\in [0\dd \tau)$;
    \item if $j=\II$, then $S[i^S \dd i^S+\ell)$ is contained in the $\tau$-run from which $i^S+\delta\in A^S$ originates;
    \item if $j=\III$, then $S[i^S \dd i^S+\delta]$ is a suffix of the $\tau$-run from which $i^S+\delta\in A^S$ originates.
\end{itemize}
\end{lemma}
\begin{proof}
By the assumption, there exist $i^S \in [1\dd |S|]$ and $i^T \in [1\dd |T|]$ such that $S[i^S \dd i^S+\ell) = T[i^T \dd i^T+\ell)$. Let us choose any such pair $(i^S,i^T)$ minimising the sum $i^S+i^T$. We have the following cases.
\begin{enumerate}
    \item If $\per(S[i^S \dd i^S+3\tau-2])>\frac13\tau$, then, by the definition of a $\tau$-synchronising set, there exist some elements $a^S \in A_I^S \cap [i^S\dd i^S+\tau)$ and $a^T \in A_I^T \cap [i^T\dd i^T+\tau)$. Let us choose the smallest such elements. By~\cref{lem:synch}, we have $a^S-i^S=a^T-i^T \in [0\dd \tau)$.
    \item Else, $p=\per(S[i^S \dd i^S+3\tau-2])\le\frac13\tau$. We have two subcases.
    \begin{enumerate}
        \item If $p=\per(S[i^S \dd i^S+\ell))$, then, by the choice of $i^S$ and $i^T$, there exists a $\tau$-run $R_S$ in~$S$ that starts at position in $(i^S-p \dd i^S]$ and a $\tau$-run $R_T$ in $T$ that starts at a position in $(i^T-p \dd i^T]$. Moreover, by \cref{lem:perlemma}, the Lyndon roots of the two runs are equal. 
        For each $X\in \{S,T\}$, let us choose $a^X$ as the leftmost starting position of a Lyndon root of $R_X$ that is weakly to the right of $i^X$. We have $a^S-i^S=a^T-i^T \in [0\dd  \lfloor\frac13\tau\rfloor)$. Each position $a^X$ will be the starting position of the first or the second occurrence of the Lyndon root of $R_S$, so $a^S \in A_{\II}^S$ and $a^T \in A_{\II}^T$.
        \item Else, $p\ne \per(S[i^S \dd i^S+\ell))$. We have $d:=\min\{b \ge p\,:\,S[i^S+b] \ne S[i^S+b-p]\}<\ell$ (and $d \ge 3\tau-1$). In this case, $a^S=i^S+d-1$ and $a^T=i^T+d-1$ are the ending positions of $\tau$-runs with period $p$ in $S$ and $T$, respectively, so $a^S \in A^S_{\III}$ and $a^T\in A^T_{\III}$.
    \end{enumerate}
\end{enumerate}
In each case, we set $\delta:=a^S-i^S=a^T-i^T$.
\end{proof}

The case when $j=I$ in the above lemma corresponds to the \problemtwo with $\P$ and $\Q$ being $(\tau, 2^{\sqrt{\log n}})$-families.
In this case, we use \cref{lem:main}.
Let us introduce a variant of the \problemtwo that intuitively corresponds to the case when $j \in \{\II,\III \}$.
A family of string pairs $\P$ is called a \emph{prefix family} if there exists a string $Y$ such that, for each $(U,V) \in \P$, $U$ is a prefix of $Y$. 

\begin{lemma}\label{lem:suffix}
An instance of the \problemtwo in which $\P \cup \Q$ is a prefix family can be solved in $\Oh(N)$ time.
\end{lemma}
\begin{proof}
By traversing $\T(\F)$ we can compute in $\Oh(N)$ time a list $\mathcal{R}$ being a union of sets $\P$ and~$\Q$ in which the second components are ordered lexicographically.

\newcommand{\maxel}{\gamma}

Consider an element $e=(U,V) \in \P$.
Let $\cpred(e)=(Y_1,Y_2)$ be the predecessor of $e$ in~$\mathcal{R}$ that originates from $\Q$ and satisfies $|Y_1| \geq |U|$.
If there is no such predecessor, we assume that $\cpred(e)$ is undefined.
Similarly, let $\csucc(e)=(Z_1,Z_2)$ be the successor of $e$ in~$\mathcal{R}$ that originates from $\Q$ and satisfies $|Z_1| \geq |U|$ (inspect Figure~\ref{fig:predsucc}).
Further, let $\maxel(e)=\max\{\LCE(U,Y_1)+\LCE(V,Y_2), \LCE(U,Z_1)+\LCE(V,Z_2)\}$; if any of the pairs $(Y_1,Y_2)$, $(Z_1,Z_2)$ is undefined, we set the corresponding component of the $\max$ operation to 0.
We also define the same notations for $e \in \Q$ with $\P$ and $\Q$ swapped.

\begin{claim}\label{clm:gamma}
$\maxPairLCP(\P,\Q)=\max_{e \in \P \cup \Q} \maxel (e)$.
\end{claim}
\begin{proof}
First, we clearly have $\maxPairLCP(\P,\Q)\geq \max_{e \in \P \cup \Q} \maxel (e)$.

Let $(P_1,P_2) \in \P$, $(Q_1,Q_2) \in \Q$ be such that $\LCE(P_1,Q_1)+\LCE(P_2,Q_2)=\maxPairLCP(\P,\Q)$.
Without loss of generality, let us assume that $|P_1| \leq |Q_1|$. We further assume that $(Q_1,Q_2)$ precedes $(P_1,P_2)$ in $\mathcal{R}$.

Let $(Y_1,Y_2)=\cpred((P_1,P_2))$. 
Now, we have $\LCE(P_1,Y_1)=|P_1|\geq \LCE(P_1,Q_1)$ since ${|P_1|\leq |Y_1|}$ and $P_1$ and~$Y_1$ are prefixes of the same string (as $\P \cup \Q$ is a prefix family).
By \cref{fact:trivial} we have $\LCE(P_2,Y_2)\geq \LCE(P_2,Q_2)$ since $(Q_1,Q_2)$, $(Y_1,Y_2)$, and $(P_1,P_2)$ appear in $\mathcal{R}$ in this order.
Thus,
\[\maxPairLCP(\P,\Q)=\LCE(P_1,Q_1)+\LCE(P_2,Q_2) \leq \LCE(P_1,Y_1)+\LCE(P_2,Y_2) \leq \max_{e \in \P \cup \Q} \maxel (e).
\]
The case when $(P_1,P_2)$ precedes $(Q_1,Q_2)$ in $\mathcal{R}$ is symmetric.
\end{proof}

\begin{figure}[ht]
    \centering
    \includegraphics[width=.5\linewidth]{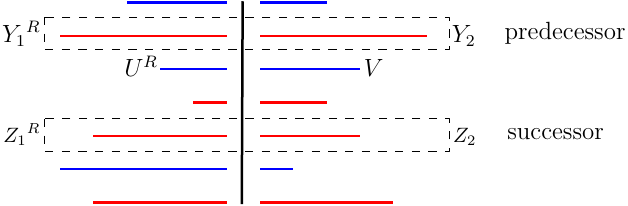}
    \caption{The setting in Lemma~\ref{lem:suffix} on list $\mathcal{R}$. With red color, we denote the elements of $\P$ and with blue color the elements of $\Q$. For element $e=(U,V)$ from $\Q$, we have $\cpred(e) = (Y_1,Y_2)$ and $\csucc(e) = (Z_1,Z_2)$.}\label{fig:predsucc}
\end{figure}

To compute $\csucc(e)$ and $\cpred(e)$ for each $e \in \P \cup \Q$ we proceed as follows.
We process the list $\mathcal{R}$ in the order of non-decreasing lengths of the first components. This order can be computed in $\Oh(N)$ time using the compacted trie $\T(\F)$.
After processing all elements $e$ of some length, we remove them from $\mathcal{R}$.

We maintain the list $\mathcal{R}$ using the data structure of Gabow and Tarjan~\cite{DBLP:journals/jcss/GabowT85} for a special case of the union-find problem, in which the elements correspond to numbers in $[1 \dd |\mathcal{R}|]$ representing the subsequent elements of the initial list $\mathcal{R}$ and the sets are formed by consecutive integers. 
Let us think of elements of $\mathcal{R}$ originating from $\P$ to be colored by red and elements originating from~$\Q$ to be colored by blue.
The sets the union-find data structure is initialised with correspond to maximal sequences of elements of the same color in $\mathcal{R}$. 
Each set has as an id its smallest element. 
Each set also maintains the following satellite information: pointers to the head and to the tail of a list of all non-deleted elements in this set.
Every element of $\mathcal{R}$ stores a pointer to the element in the list in which it is represented. It can thus be deleted at any moment in $\cO(1)$ time. When the satellite list of a set~$S_i$ becomes empty, we union $S_{i-1}$, $S_i$, and $S_{i+1}$ (provided that the respective sets exist). To find $\cpred(e)$ we perform the following; the procedure of $\csucc(e)$ is analogous. 
We find the set $S_i$ to which $e$ belongs using a find operation. Let the id of $S_i$ be $\alpha$. By using another find operation to find $\alpha-1$, we find the set $S_{i-1}$. The tail of the list of $S_{i-1}$ is $\cpred(e)$; in the analogous procedure, the head of the list of $S_{i+1}$ is $\csucc(e)$. The algorithm is correct because we process $\mathcal{R}$ in the order of non-decreasing lengths. Each union or find operation requires $\cO(1)$ amortized time~\cite{DBLP:journals/jcss/GabowT85}.
Thus, this procedure takes $\cO(N)$ time in total.

Finally, we preprocess the compacted trie $\T(\F)$ in $\cO(N)$ time to be able to answer lowest common ancestor (LCA) queries in $\cO(1)$ time~\cite{DBLP:conf/latin/BenderF00}.
For any $e=(U,V)$, given $\cpred(e)$ and $\csucc(e)$ we can compute $\maxel(e)$ in $\cO(1)$ time by answering $\LCE$ queries using the LCA data structure.
\end{proof}

We are now ready to state the main result of this subsection. The proof of \cref{lem:main} is deferred to \cref{sec:proof_lem_main}.

\begin{lemma}[Medium-length LCS]\label{lem:medium}
The LCS problem can be solved in $\cO(n \log \sigma /\sqrt{\log n})$ time
using $\Oh(n / \log_{\sigma } n)$ space if $\frac13\log_\sigma n\le\ell\le 2^{\sqrt{\log n}}$.
\end{lemma}
\begin{proof}
Recall that $\tau=\floor{\frac{1}{9}\log_\sigma n}$.
The set of anchors $A=A_I \cup A_\II \cup A_\III$ consists of a $\tau$-synchronising set and of $\Oh(1)$ positions per each $\tau$-run in $S\$T$. Hence, $|A| = \Oh(n/\tau)$ and $A$ can be constructed in $\Oh(n/\tau)$ time by~\cref{thm:synch_packed,lem:tauruns}. We also use \cref{lem:tauruns} to group all $\tau$-runs by their Lyndon roots.

We construct sets of pairs of substrings of $X=STS^RT^R$. First, for $\Delta=\floor{2^{\sqrt{\log n}}}$:
\[
    \P_I = \{((S[a-\tau\dd a))^R,S[a \dd a+\Delta))\,:\,a \in A_I^S\}.
\]

Then, for each group $\G$ of $\tau$-runs in $S$ and $T$ with equal Lyndon root, we construct the following set of string pairs:
\[
    \P_\II^\G = \{((S[x\dd a))^R,S[a \dd y])\,:\,a \in A_\II^S\text{ that originates from }\tau\text{-run }S[x \dd y] \in \G\}.
\]

We define the \emph{tail} of a $\tau$-run $S[x \dd y]$ with period $p$ and Lyndon root $S[x' \dd x'+p)$ as $(y+1-x') \bmod p$ (and same for $\tau$-runs in $T$).
For each group of $\tau$-runs in $S$ and $T$ with equal Lyndon roots, we group the $\tau$-runs belonging to it by their tails. This can be done in $\Oh(n/\tau)$ time using Radix Sort since the tail values are up to $\frac13\tau$.
For each group $\G$ of $\tau$-runs in $S$ and $T$ with equal Lyndon root and tail, we construct the following set of string pairs:
\[
\P_\III^\G = \{((S[x \dd y))^R,S[y \dd |S|])\,:\,S[x \dd y] \in \G\}.
\]
Simultaneously, we create sets $\Q_I$, $\Q^\G_\II$ and $\Q^\G_\III$ defined with $T$ instead of $S$.

By~\cref{lem:anchors}, it suffices to output the maximum of \[\{\maxPairLCP(\P_I,\Q_I), \maxPairLCP(\P_\II^\G,\Q_\II^\G), \maxPairLCP(\P_\III^\G,\Q_\III^\G)\},\] where $\G$ ranges over groups of $\tau$-runs in $S\$T$.

Computing any individual value of $\maxPairLCP$ can be expressed as an instance of the \problemtwo provided that all the first and second components of families are represented as nodes of compacted tries. 
We will use~\cref{lem:ctrie} to construct these compacted tries. $\LCE$ queries can be answered efficiently due to~\cref{thm:lce}, so it suffices to be able to sort all the first and second components of each pair of string pair sets lexicographically. 
Each of the sets $\P_I$ and $\Q_I$ can be ordered by the second components using~\cref{thm:suffixsort} since $A_I$ is a $\tau$-synchronising set, and by the first components via a straightforward tabulation approach because the number of possible $\tau$-length strings is $\sigma^\tau=\cO(n^{1/9})$. In a set $\P^\G_\II$, all first and all second components are prefixes of a single string (a power of the common Lyndon root). Hence, they can be sorted simply by comparing their lengths. This sorting is performed simultaneously for all the families $\P_\II^\G$, $\Q_\II^\G$ in $\Oh(n/\tau)$ time via Radix Sort. Finally, to sort the second components of the sets $\P_\III^\G$ $\Q_\III^\G$, instead of comparing strings of the form $S[y \dd |S|]$ (and same for $T$), we can equivalently compare strings $S[y-2\tau+1 \dd |S|]$ which are known to start at positions from a $\tau$-synchronising set by~\cref{lem:fact3.2}. This sorting is done across all groups using Radix Sort and~\cref{thm:suffixsort}.

Finally, we observe that $(\P_I,\Q_I)$ is a $(\tau,\Delta)$-family of size $N=\Oh(n/\tau)$, and thus the value $\maxPairLCP(\P_I,\Q_I)$ can be computed in $\Oh(n\log\sigma/\sqrt{\log n})$ time and $\Oh(n/\log_\sigma n)$ space using \cref{lem:main}. 
On the other hand, $(\P_\II^\G,\Q_\II^\G)$ and $(\P_\III^\G,\Q_\III^\G)$ are prefix families of total size $\Oh(n/\tau)$, so the corresponding instances of the \problemtwo can be solved in $\Oh(n/\log_\sigma n)$ total time using \cref{lem:suffix}.
\end{proof}

\section{Proof of Lemma \ref{lem:main}---Solution to a Special Case of the \problemtwo via Wavelet Trees}\label{sec:proof_lem_main}

\paragraph{Wavelet trees.}
For an arbitrary alphabet $\Sigma$, the \emph{skeleton tree} for $\Sigma$
is a full binary tree $\T$ together with a bijection between $\Sigma$ and the leaves of $\T$. 
For a node $v\in \T$, let $\Sigma_v$ denote the subset of $\Sigma$ that corresponds to the leaves in the subtree of $v$.

For a skeleton tree $\T$ and a string $T\in \Sigma^*$, the \emph{$\T$-shaped wavelet tree} of $T$ is $\T$ augmented with bit-vectors assigned to its internal nodes (inspect~\cref{fig:WT}(a) for a wavelet tree with a ``standard'' skeleton tree). 
For an internal node $v$ with left child $v_L$ and right child $v_R$, let $T_v$ denote the maximal subsequence of $T$ that consists of letters from $\Sigma_v$; the bit-vector $B_v[1\dd |T_v|]$ is defined so that $B_v[i]=0$ if $T_v[i]\in \Sigma_{v_L}$ and $B_v[i]=1$ if $T_v[i]\in \Sigma_{v_R}$.

\begin{figure}[htpb]
  \begin{center}
  \begin{subfigure}[b]{0.55\textwidth}
    \includegraphics[width=7cm]{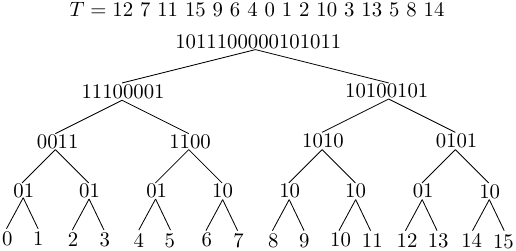}
    \caption{Standard skeleton tree with $|T|=\sigma=16$.}\label{fig:std}
  \end{subfigure}
  \begin{subfigure}[b]{0.4\textwidth}
  \centering
     \includegraphics[width=3cm]{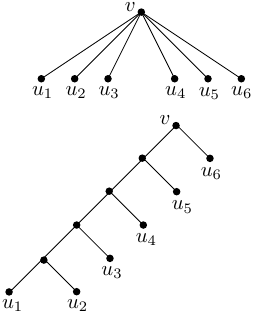}
     \caption{Binarisation with height $\cO(\alpha + \log N)$.}\label{fig:binN}
  \end{subfigure}
\end{center}
\caption{(a) Let $v$ be left child of the root node. Then $\Sigma_v=\{0,1,\ldots,7\}$, $T_v=7,6,4,0,1,2,3,5$ and so $B_v=11100001$: $7,6,4,5$ belong to the right subtree of $v$ and $0,1,2,3$ to the left. (b) For each $i$, let the size of the subtree rooted at $u_i$ be $2^{i}$. The binarisation from~\cite{DBLP:conf/soda/BabenkoGKS15} leads to height $\cO(\alpha+\log N)$, favouring heavier children.}\label{fig:WT}
\end{figure}

Wavelet trees were introduced in~\cite{DBLP:conf/soda/GrossiGV03}. 
One can derive an $\Oh(n \log \sigma)$-time construction algorithm directly from their definition;
more efficient construction algorithms were presented in~\cite{DBLP:journals/tcs/MunroNV16,DBLP:conf/soda/BabenkoGKS15}.

\begin{theorem}[see {\cite[Theorem 2]{DBLP:journals/tcs/MunroNV16}}]\label{thm:wavelet}
Given the packed representation of a string $T\in [0\dd \sigma)^n$ and a skeleton tree $\T$ of height $h$, the $\T$-shaped wavelet tree of $T$ takes $\Oh(nh/\log n+\sigma)$ space and can be constructed in $\Oh(nh/\sqrt{\log n}+\sigma)$ time.
\end{theorem}

Wavelet trees are sometimes constructed for sequences $T\in \M^*$
over an alphabet $\M \sub \Sigma^*$ that itself consists of strings (see e.g.,~\cite{DBLP:conf/stoc/KempaK19}).
In this case, the skeleton tree $\T$ is often chosen to resemble the compacted trie of $\M$.
Formally, we say that a skeleton tree $\T$ for some $\M \sub \Sigma^*$ is \emph{prefix-consistent}  if
each node $v\in \T$ admits a label $\val(v)\in \Sigma^*$ such that:
\begin{itemize}
  \item the restriction of the function $\val(v)$ to the leaves $v \in \T$ is a bijection from those leaves to~$\M$;
  \item if $v$ is a node with children $v_L,v_R$, then, for all leaves $u_L,u_R$ in the subtrees of $v_L$ and $v_R$, respectively, the string $\val(v)$ is the longest common prefix of $\val(u_L)$ and $\val(u_R)$.
\end{itemize}
We observe that if $\M\sub \{0,1\}^{\alpha}$ for some integer $\alpha$, then the compacted trie $\T(\M)$
is a prefix-consistent skeleton tree for $\M$. For larger alphabets, we binarise $\T(\M)$ as follows.

\begin{restatable}{lemma}{lemwav}\label{lem:C2}
Given the compacted trie $\T(\M)$ of a set $\M\sub \Sigma^\alpha$, 
a prefix-consistent skeleton tree for~$\M$ of height $\Oh(\alpha+\log |\M|)$ can be constructed in $\Oh(|\M|)$ time, with each node $v$ associated to a node $v'$ of $\T(\M)$ such that $\val(v)=\val(v')$.
\end{restatable}
\begin{proof}
We use~\cite[Corollary 3.2]{DBLP:conf/soda/BabenkoGKS15}, where the authors showed that any rooted tree with $m$ leaves and of height $h$ can be binarised in $\cO(m)$ time so that the resulting tree is of height $\cO(h + \log m)$. For $\T(\M)$, we obtain height $\cO(\alpha + \log |\M|)$ and time $\cO(|\M|)$ (inspect~\cref{fig:WT}(b)).
\end{proof}

\lemmain*
\begin{proof}
If $\log \beta = \Omega(\log N)$, the stated bounds are not better than those of \cref{lem:problem}. We henceforth assume that $\log \beta = o(\log N)$.

Let $\R$ be the list obtained by sorting $\P \cup \Q$ with respect to the lexicographic order of the second components.
For any sublist $\mathcal X = (U_1,V_1),\dots,(U_m,V_m)$ of~$\R$, let $\LCPs(\mathcal X)$ denote the representation of the list $0,\LCP(V_1,V_2),\dots,\LCP(V_{m-1},V_{m})$ as a packed string over alphabet $[0\dd \beta]$ in space $\Oh(1+N/\log_\beta N)$.
For each node $v$ of the wavelet tree, let $\R_v$ be the sublist of $\R$ composed of elements whose first component is in the leaf list $\Sigma_v$ of $v$.

\begin{claim}\label{clm:useful}
Consider an instance of the \problemtwo in which $\P$ and~$\Q$ are $(\alpha,\beta)$-families with $\log \beta = o(\log N)$.
We can construct, in $\cO(N(\alpha+\log N)/ \sqrt{\log N})$ time and $\cO(N + N\alpha/\log N)$ space, a wavelet tree of height $\cO(\alpha + \log N)$ for the first components of $\R$ (some possibly padded with a $\$\not\in\Sigma$).

Moreover, in $\cO(N(\alpha+\log N) \log \beta / \log N)$ time and $\cO(N)$ space, we can compute a bit-vector~$G_v$ specifying the origin ($\P$ or $\Q$) of each element of $\R_v$ and the list $\L_v = \LCPs(\R_v)$ for each node~$v$ of the wavelet tree in the BFS order, such that after computing $G_u$ and $\L_u$ for each child $u$ of a node~$v$, $G_v$ and $\L_v$ are deleted.
\end{claim}
\begin{proof}
By traversing $\T(\F)$ we can sort in $\Oh(N)$ time all elements of $\P \cup \Q$ by the second components, obtaining the list $\R$. We also store a bit-vector~$G$ of length $|\R|$ that specifies, for each element of $\R$, which of the sets $\P$, $\Q$ it originates from.
We then construct the prefix-consistent skeleton tree for the sequence of strings being the first components of pairs from~$\R$ with~\cref{lem:C2} and use it to construct the wavelet tree of the same sequence using~\cref{thm:wavelet}.
Before said skeleton and wavelet trees are constructed, we pad each string with a letter $\$\not\in\Sigma$ to make them all of length $\alpha$; in what follows, we ignore the nodes of the wavelet tree whose path-label contains a \$.

The list $\LCPs(\R)$ can be computed in $\Oh(N)$ time when constructing $\R$. For each node~$v$ of the wavelet tree, we wish to compute $\L_v=\LCPs(\R_v)$ and the corresponding bit-vector $G_v$ of origins ($\P$ or $\Q$) of the elements of $\R_v$. We will construct the lists $\LCPs(\R_v)$ without actually computing~$\R_v$. We process the nodes in the BFS order.

\bigskip
\newcommand{\rd}[1]{#1.\texttt{read}()}
 \newcommand{\wrt}[2]{#1.\texttt{write}(#2)}
 \newcommand{\eof}[1]{#1.\mathtt{end}\texttt{\_}\mathtt{of}\texttt{\_}\mathtt{stream}()}

 \begin{algorithm}[H]
  \caption{Computation of $\LCPs$ lists.}\label{algo:computeLCPs}
  $\mu_0 := \mu_1 := \rd{\L}$\;
  \While{\KwSty{not} $\eof{B}$}{
    $b:= \rd{B}$\;
    $\wrt{G_{b}}{\rd{G}}$\;
    $\wrt{\L_{b}}{\mu_b}$\;
    $\mu_{b} := \rd{\L}$\;
    $\mu_{1-b} := \min(\mu_{1-b},\mu_b)$\;
  }
 \end{algorithm}%

\bigskip
\paragraph{Computation of $\LCPs$ lists.} We apply \cref{algo:computeLCPs}, that we describe below, to each node~$v$ of the wavelet tree, with $\L$ being a stream representing the $\LCPs$ list $\L_v$, $G$ a stream of the corresponding bit-vector of origins~$G_v$, and $B$ a stream of the bit-vector $B_v$ stored in the wavelet tree node. The algorithm outputs streams $\L_0$ and $\L_1$ representing the lists $\L_{u_L}$ and $\L_{u_R}$ of the left and right children $u_L$, $u_R$ of $v$ and, analogously, streams $G_0$ and $G_1$ representing $G_{u_L}$ and $G_{u_R}$.

Let us describe the algorithm and argue for its correctness. Variable $b$ is used to store each subsequent bit of the stream $B$. It is then used to forward the next bit of the stream of origins $G$ to the respective output stream.

Variables $\mu_0$ and $\mu_1$ are used for processing the $\LCPs$ list. Let $\R_v=(U_1,V_1),\dots,(U_m,V_m)$ and assume $V_0$ is a sentinel empty string. We next show that before the step of the while-loop in which the $i$th bit of $B$ is read, we have
\begin{equation}\label{eq:mu}
    \mu_b=\LCP(V_{p_b},V_i),\text{ where }p_b=\max(\{0\} \cup \{p \in [1\dd i)\,:\,B_v[p]=b\}).
\end{equation}
In particular, if $p_b=0$, i.e., no $b$-bit in $B$ was encountered yet, then $\mu_b=0$.

Let us recall that the list $\L_v$ starts with a 0, so initially $\mu_0=\mu_1=0$ satisfy \eqref{eq:mu}. Let us consider the $i$th step of the loop with $b=B[i]$. If $p_b=0$, then this is the first $b$-bit in $B$, so $\mu_b=0$ is correctly output to $\L_b$. Otherwise, $\mu_b=\LCP(V_{p_b},V_i)$ is output to $\L_b$, as expected. Next, $\mu_b$ becomes $\L_v[i+1]=\LCP(V_i,V_{i+1})$ which is correct since when $i$ is incremented, $p_b=i$. Finally, $\mu_{1-b}$ becomes $\min(\LCP(V_{p_{1-b}},V_i),\LCP(V_i,V_{i+1}))$, which is $\LCP(V_{p_{1-b}},V_{i+1})$ by \cref{fact:trivial} and the fact that the strings $V_1,\ldots,V_m$ are ordered lexicographically. When $i=m$, the next bit on stream~$\L$ is not defined, which is not an issue as then the lines setting $\mu_0$ and $\mu_1$ are not relevant anymore.

\paragraph{Analysis.}
We analyze the complexity of \cref{algo:computeLCPs} using \cref{prp:streaming}. The algorithm works on $\Oh(1)$ data streams, uses $s=\Oh(\log (\beta+1))$ bits of space and among every $t = \Oh(1)$ instructions performs a read or write. For a node $v$ of the wavelet tree, the total size of input and output streams in bits is proportional to the number of bits in $\L_v$, i.e., $\Oh(|B_v| \log \beta)$. By \cref{prp:streaming}, for $\tau=\frac12 \log N$, after $o(N)$ preprocessing (as $\log\beta=o(\log N)$), the algorithm can be executed in $\Oh(1+|B_v| \log \beta / \log N)$ time. Whenever \cref{algo:computeLCPs} is applied, the input streams are removed from memory. This guarantees that the additional space required to store the streams for the at most $N$ topmost nodes that have not been processed yet is bounded by~$\Oh(N)$.
 
The wavelet tree, which is of height $h = \cO(\alpha + \log N)$, can be built in $\Oh(N h/\sqrt{\log N}+N)$ time using space $\Oh(N h/\log N+N)$ by~\cref{thm:wavelet}.
Over all nodes, \cref{algo:computeLCPs} consumes $\Oh(Nh\log \beta/\log N)$ time and $\Oh(N)$ space.
The overall complexities follow.
\end{proof}

\paragraph{Application of $\LCPs$ lists.}
For each node $v$ of the wavelet tree, let us define $f(v)$ as the sum of the string-depth $|\val(v)|$ and the maximum $\LCP$ between second components of pairs in $\R_v$ of different origins.
For each $(U,V) \in \P$ and $(U',V') \in \Q$ such that $\LCP(U,U')=d$, the pairs $(U,V)$ and $(U',V')$ are in the same list $\R_v$ of one node $v$ per each depth in $[0 \dd d]$.
Thus, to compute $\maxPairLCP(\P,\Q)$, it suffices to compute $f(v)$ for each node $v$ of the wavelet tree and return the maximum of such values.

Greedily, again by \cref{fact:trivial}, the maximum $\LCP$ is obtained by two \emph{consecutive} elements in the list~$\R_v$, i.e., as $\max\{\L_v[i]\,:\,G_v[i] \ne G_v[i-1],\,i \in [2 \dd |\L_v|]\}$.
This maximum LCP is computed in \cref{algo:applyLCPs} exactly from this formula with $\L$ and $G$ being streams representing $\L_v$ and $G_v$, as before.
The algorithm is used whenever the lists $\L_v$ and $G_v$ have been computed for a node $v$ in \cref{clm:useful}.

 \medskip
\newcommand{\rd}[1]{#1.\texttt{read}()}
\newcommand{\wrt}[2]{#1.\texttt{write}(#2)}
 \newcommand{\eof}[1]{#1.\mathtt{end}\texttt{\_}\mathtt{of}\texttt{\_}\mathtt{stream}()}
\begin{algorithm}[H]
  \caption{Application of $\LCPs$ lists.}\label{algo:applyLCPs}
  $r := g':= 0$\;
  \While{\KwSty{not} $\eof{G}$}{
    $g:= \rd{G}$\;
    \lIf{$g \ne g'$}{$r  := \max(r, \rd{\L})$}
    \lElse{$\rd{\L}$}
    $g' := g$\;
  }
  \Return{$r$}\;
 \end{algorithm}%

 \medskip
Let $\R_v=(U_1,V_1),\dots,(U_m,V_m)$. We next show that before the $i$th step of the while-loop,
\[r=\max(\{0\} \cup \{\LCP(V_{j-1},V_j)\,:\,j \in [2\dd i-1],\,G_v[j-1] \ne G_v[j]\}).\]
Moreover, $g'=G_v[i-1]$ if $i>1$ and $g'=0$ if $i=1$.

Let us consider the $i$th step of the loop with $g=G[i]$. If $i=1$, then $\L[i]=0$ by definition, so the value of $r$ does not change regardless of the result of comparing $g$ and $g'$. If $i>1$, we have $\L[i]=\LCP(V_{i-1},V_i)$. Therefore, if $g \ne g'$, i.e., $G[i-1] \ne G[i]$, variable $r$ is updated correctly, and otherwise the value $\L[i]$ can be discarded. Finally, $g'$ becomes $g$.

The complexities of \cref{algo:applyLCPs} are asymptotically dominated by those of \cref{algo:computeLCPs} and hence the result follows.
\end{proof}

\section{Sublinear-time LCS of Several Strings}
We show how to modify the results from~\cref{sec:0LCS,sec:proof_lem_main} to compute the LCS of $\lambda$ strings $T^{(1)},\ldots,T^{(\lambda)}$ of total length $n$ in $\Oh(n\log \sigma / \sqrt{\log n})$ time if $\lambda = \Oh(\sqrt{\log n}/\log \log n)$; in particular, for $\lambda=\Oh(1)$.

By $\ell$ we denote the length of an LCS. As before, the solution is divided into three cases, depending on whether $\ell$ is small (\cref{sec:lambda_short}), large (\cref{sec:lambda_long}), or medium (\cref{sec:lambda_medium}). Instead of instances of \problemtwo, we will obtain instances of its generalisation called \problemlambda:

\defproblem{\problemlambda}{%
A compacted trie $\T(\F)$ of $\F\sub \Sigma^*$
and $\lambda$ non-empty sets $\P^{(1)},\ldots,\P^{(\lambda)}\sub \F^2$, with $|\F|+\sum_{t=1}^\lambda |\P^{(t)}| = N$.}
 {$\maxMultiLCP(\P^{(1)},\ldots,\P^{(\lambda)})=\max\{\LCP(P^{(1)}_1,\ldots,P^{(\lambda)}_1)+\LCP(P^{(1)}_2,\ldots,P^{(\lambda)}_2) : (P^{(t)}_1,P^{(t)}_2)\in \P^{(t)}\text{ for all }t \in [1\dd \lambda]\}$.}

Before we proceed with the solution, in \cref{sec:lambda_strings} we provide generalisations of some of the tools used for two input strings to the setting with multiple input strings. We also discuss efficient predecessor/successor data structures (\cref{sec:pred}).

\subsection{Toolbox for Handling Multiple Input Strings}\label{sec:lambda_strings}
Let us recall that $(D,h)$ is a ($\lambda, d$)-\emph{cover} if $D \sub \mathbb{Z}_+$ is a set and $h : \Zp^\lambda \to [0\dd d)$ is an efficiently computable function such that, for any $i_1,\dots,i_\lambda\in \mathbb{Z}_+$ and $t\in [1\dd \lambda]$, we have $i_t+h(i_1,\dots,i_\lambda)\in D$.
We show an efficient construction of ($\lambda,d$)-covers; the theorem below is restated for convenience.

\lambdadcover*
\begin{proof}
Let $c = \floor{\sqrt[\lambda]d}$. 
Our initialization procedure determines $c$ in $\Oh(\lambda \cdot \log c)=\Oh(\log d)$ time using exponential search.

We define 
\[D = \bigcup_{t=1}^{\lambda} \{j\in \Zp : j \bmod c^{t} \ge c^t- c^{t-1}\}.\]
Note that $|D\cap [1\dd n]| \le \sum_{t=1}^{\lambda}|\{j\in [1\dd n] : j \bmod c^{t} \ge c^t- c^{t-1}\}| \le \lambda \cdot n/c = \Oh(\lambda \cdot n/\sqrt[\lambda]{d})$.
Moreover, the set $D\cap [1\dd n]$ can be easily constructed in $\Oh(\lambda \cdot n/\sqrt[\lambda]{d})$ time.

For any $i_1,\ldots,i_\lambda\in \Zp$, we define $h(i_1,\ldots,i_\lambda)=h_\lambda$ based on a sequence $0=h_0 \le \cdots \le h_{\lambda}$ specified using the following recursion for $t\in [1\dd \lambda]$:
\begin{align*}
h_t &= h_{t-1}+c^{t-1}\cdot \left(c-1- \left\lfloor \frac{i_t+h_{t-1}}{c^{t-1}} \right\rfloor \bmod c\right)\\
&=h_{t-1}+c^t-c^{t-1}+(i_t+h_{t-1})\bmod c^{t-1}-(i_t+h_{t-1})\bmod c^{t}.
\end{align*}
The values $(h_0,\ldots,h_\lambda)$ can be computed in $\Oh(\lambda)$ time along with the values $(c^0,\ldots, c^{\lambda})$.

A simple inductive argument proves that $h_t \in [0\dd c^{t})$ holds for every $t\in [0\dd \lambda]$ and, in particular, $h_\lambda \in [0\dd c^{\lambda})\subseteq [0\dd d)$. 
It remains to prove that $i_t + h_\lambda\in D$ holds for every $t\in [1\dd \lambda]$.
For this, observe that
\begin{align*}i_t + h_t &= i_t+h_{t-1}+c^t-c^{t-1}+(i_t+h_{t-1})\bmod c^{t-1}-(i_t+h_{t-1})\bmod c^{t}\\
&\equiv c^t-c^{t-1}+(i_t+h_{t-1})\bmod c^{t-1}  \pmod{c^t}.\end{align*}
Since $h_t \equiv h_{t+1}\equiv \cdots \equiv h_\lambda \pmod{c^t}$, we conclude that  $(i_t+h_\lambda)\bmod c^{t}=(i_t+h_t)\bmod c^{t} \in [c^t-c^{t-1}\dd c^t)$, and thus $i_t + h_\lambda \in D$ holds as claimed.
\end{proof}

\newcommand{\Color}{\mathsf{col}}
\newcommand{\X}{\mathcal{X}}
Let $\L$ be a multiset of strings and assume that each string $S\in \L$ is assigned a color $\Color(S) \in [1\dd \lambda]$. Note that, in general, multiple strings from $\L$ may be assigned the same color.
Let $\$,\#$ be auxiliary letters that are smaller and greater than all other letters of the alphabet, respectively. For a string~$U$ and a color $c \in [1\dd \lambda]$, by $\nnext_c(U,\L)$ we denote the lexicographically smallest string $S \in \L$ such that $S \ge U$ and $\Color(S)=c$. If no such string $S$ exists, we assume that $\nnext_c(U,\L)=\#$. Symmetrically, by $\pprev_c(U,\L)$ we denote the lexicographically greatest string $S \in \L$ such that $S \le U$ and $\Color(S)=c$, or $\$$ if no such string $S$ exists. We also denote
\begin{align*}\pprev(U,\L)&=\min\{\pprev_c(U,\L)\,:\,c \in [1\dd \lambda]\},\\
\nnext(U,\L)&=\max\{\nnext_c(U,\L)\,:\,c \in [1\dd \lambda]\}.
\end{align*}
A subset $\X$ of $\L$ is called a \emph{$\lambda$-rainbow} if $|\X|=\lambda$ and all strings in $\X$ have different colors.

The next lemma can be viewed as an extension of \cref{fact:trivial}.

\begin{lemma}\label{lem:colored_lcp}
Let $\L$ be a multiset of strings and assume that each string $S\in \L$ is assigned a color $\Color(S) \in [1\dd \lambda]$ so that no two equal strings in $\L$ have the same color and each color in $[1\dd \lambda]$ is present in $\L$.
\begin{enumerate}[(a)]
    \item\label{col_b} The maximum $\LCP$ over all $\lambda$-rainbows $\X \subseteq \L$ equals $\max_{S \in \L} \LCP(S, \pprev(S,\L))$.
    \item\label{col_a} Assume that $S \in \L$. Let $c_1,\ldots,c_\lambda$ be a permutation of $[1\dd \lambda]$ such that $\pprev_{c_1}(S,\L) \le \cdots \le \pprev_{c_\lambda}(S,\L)$ and $\Color(S)=c_\lambda$. Let us denote $c_0=c_\lambda$. Then, the maximum $\LCP$ of a $\lambda$-rainbow $\X \subseteq \L$ such that $S \in \X$ equals
    \begin{equation}\label{eq:maxLCPeq}
        \maxLCP(S,\L) := \max_{i=1}^\lambda \LCP(\pprev_{c_i}(S,\L),\ \max \{\nnext_{c_j}(S,\L)\,:\,j \in [0\dd i)\}).
    \end{equation}
\end{enumerate}
\end{lemma}
\begin{proof}
Proof of part \eqref{col_a}: Let us fix $S \in \L$, let $\zeta=\maxLCP(S,\L)$, and assume that the maximum in formula~\eqref{eq:maxLCPeq} is attained for $i=i'$.

Let $\X' \subseteq \L$ be a $\lambda$-rainbow such that $S \in \X'$ and $\LCP(\X')$ is maximal among such multisets. Among possibly many multisets $\X'$ that satisfy the requirements, we select the one with the largest $\min \X'$. By \cref{fact:trivial}, $\LCP(\X')=\LCP(\min \X', \max \X')$.
We will prove that $\zeta=\LCP(\X')$ by showing two inequalities.

$\zeta \le \LCP(\X')$: Since $\LCP(\X')\ge 0$ holds trivially, we can assume $\zeta > 0$. It suffices to construct a $\lambda$-rainbow $\X \subseteq \L$ such that $S \in \X$ and $\LCP(\X)\ge \zeta$. The required set is
\[\X=\{\pprev_{c_i}(S,\L)\,:\,i \in  [i'\dd \lambda]\} \cup \{\nnext_{c_i}(S,\L)\,:\,i \in  [1\dd i')\}.\]
Indeed, $S=\pprev_{c_\lambda}(S,\L) \in \X$.
Moreover, we have $\Color(\pprev_{c_i}(S,\L))=c_i$ for $i \in [i'\dd \lambda]$ and $\Color(\nnext_{c_j}(S,\L))=c_j$ for all $j \in [1\dd i')$; this is because $\zeta>0$, and hence $\pprev_{c_{i'}}(S,\L) \ne \$$ and $\max_{j \in [1\dd i')} \nnext_{c_j}(S,\L) \ne \#$. 
Consequently, $\X$ is a $\lambda$-rainbow. We have $\min \X = \pprev_{c_{i'}}(S,\L)$. Further,
\[\max \X = \max (\{\pprev_{c_\lambda}(S,\L)\} \cup \{\nnext_{c_j}(S,\L)\,:\,j \in [1\dd i')\}) = \max \{\nnext_{c_j}(S,\L)\,:\,j \in [0\dd i)\}.\]
By \cref{fact:trivial}, $\LCP(\X)=\LCP(\min \X,\max \X)=\zeta$, as required.

$\zeta \ge \LCP(\X')$: 
Recall that $\X'$ was selected as a multiset that maximises $\min \X'$.
Consider the multiset $\mathcal{Y}$ of minimal elements in $\X'$.
We next argue that $\mathcal{Y}$ has a non-empty subset $\mathcal{Y'}$ such that for all $Y \in \mathcal{Y'}$, $Y=\pprev_{\Color(Y)}(S,\L)$. Suppose to the contrary that this is not the case. Then, for each $Y \in \mathcal{Y'}$, let us replace $Y$ in $\X'$ with $\pprev_{\Color(Y)}(S,\L)$.
This way, we would obtain a $\lambda$-rainbow~$\X''$ containing $S$ and satisfying $\min \X'' > \min \X'$.
Moreover, \[\LCP(\X'') = \LCP(\min \X'', \max \X'') = \LCP(\min \X'', \max \X') \ge \LCP(\min \X', \max \X') = \LCP(\X')\]
by \cref{fact:trivial} as $\min X' < \min \X'' \le \max \X'$; a contradiction.

Now, consider an element of $\mathcal{Y}'$ with minimal color, and suppose that this color is $c_i$. Now, it suffices to note that
\begin{equation}\label{eq:next_ineq}
    \max(\X') \ge \max \{\nnext_{c_j}(S,\L)\,:\,j \in [0\dd i)\}
\end{equation}
Indeed, if $i>1$, by the order of $c_1,\ldots,c_\lambda$, no string $V$ in $\L$ such that $\pprev_{c_i}(S,\L) \le V \le S$ has a color in $\{c_1,\ldots,c_{i-1}\}$, and the smallest strings in $\L$ that are greater than $S$ and have colors in $\{c_1,\ldots,c_{i-1}\}$ are exactly $\nnext_{c_j}(S,\L)$ for $j \in [1\dd i)$. If $i=1$, the right-hand side of \eqref{eq:next_ineq} is just $S$.

By \eqref{eq:next_ineq} and \cref{fact:trivial},
\[\LCP(\X') = \LCP(\min \X', \max \X') \le \LCP(\min \X',\max \{\nnext_{c_j}(S,\L)\,:\,j \in [0\dd i)\}) = \zeta,\]
as required.

\medskip
Proof of part \eqref{col_b}: From part \eqref{col_a}, it follows that, for every string $S$,
\[\maxLCP(S,\L) = \LCP(\nnext(U,\L),U)\]
for some string $U$. More precisely, the string $U$ is defined as the string $\pprev_{c_i}(S,\L)$ for $i$ which yields the maximal value in \eqref{eq:maxLCPeq}.
Hence, the maximum LCP of a $\lambda$-rainbow $\X \subseteq \L$ does not exceed $\max_{U \in \L} \LCP( \nnext(U,\L),U)$, which is equal by symmetry to $\max_{U \in \L} \LCP( \pprev(U,\L),U)$. The same value is clearly a lower bound for this maximum LCP. This concludes the proof.
\end{proof}

We often answer $\LCE$ queries between substrings of pairs of input strings $T^{(1)},T^{(2)},\ldots,T^{(\lambda)}$. To this end, we use the $\LCE$ data structure of \cref{thm:lce} in a string $T^{(1)}\$ T^{(2)} \$\cdots \$T^{(\lambda)}$ (where $\$ \not \in \Sigma$). 
Whenever we perform a longest common extension query for substrings of two strings $T^{(x)}$ and $T^{(y)}$, we ensure that no delimiter $\$$ is crossed by capping the result -- as the substrings are given in the form $T^{(z)}[i \dd j]$, we can infer the distances of $i$ and $j$ from the beginning and end of~$T^{(z)}$ in constant time. A symmetric solution can be used for answering $\LCE$ queries on substrings of the reversals of the input strings. Let us note that we deliberately do not use unique delimiters when concatenating the input strings as this would blow up the alphabet size if $\sigma$ is small.

\subsection{Dynamic Predecessor Data Structures}\label{sec:pred}

For a set of integers $A$ and integer $x$, the predecessor and successor of $x$ in $A$ are defined as
$\max\{a \in A\,:\,a < x\}$ and $\min\{a \in A\,:\,a>x\}$, respectively. (We assume that $\max \emptyset = -\infty$ and $\min \emptyset = \infty$.)
In a dynamic version of the problem, a set $A \subseteq [1\dd u]$ of size up to $n$ is to be maintained under insertions and deletions so that predecessor and successor queries for $A$ can be answered efficiently.
The classic solution by van Emde Boas \cite{DBLP:journals/ipl/Boas77} achieves $\Oh(\log \log u)$ time per operation using space $\Oh(u)$. We use a slightly more space-efficient solution.

\begin{lemma}[\cite{DBLP:journals/siamcomp/Willard00,DBLP:journals/jda/LagogiannisMT06}]\label{lem:compact_pred}
The dynamic predecessor/successor problem on an instance of size up to $n$ over a universe $[1\dd u]$ can be solved in $\Oh(\log \log u)$ time per operation using space $\Oh(u/\log u+n)$.
\end{lemma}

The dynamic predecessor/successor problem can also be solved in $\Oh(n)$ space and $\Oh(\log \log u)$ amortised expected time per operation using a y-fast trie~\cite{DBLP:journals/ipl/Willard83}
or in $\Oh(\log^2 \log u/\log \log \log u)$ worst-case time per operation using an exponential search tree~\cite{DBLP:journals/jacm/AnderssonT07}.
For simplicity, an AVL tree with $\Oh(n)$ space and $\Oh(\log u)$ time cost of operations can also be used.
By $\pi_u$ we denote the time of an operation on a predecessor/successor data structure of size $\Oh(n)$ over universe $[1\dd u]$.

\subsection{Short LCS}\label{sec:lambda_short}
We assume that $\ell\leq \frac13 \log_\sigma n$.

The algorithm follows rather closely the proof of \cref{lem:short}. Each of the input strings is split into fragments of length (up to) $2m$, where $m = \lfloor\frac13 \log_\sigma n\rfloor$. We compute all distinct substrings corresponding to such fragments for each of the strings in $\Oh(n/m+\lambda+n^{2/3})$ total time using Radix Sort of pairs of the form: (an integer encoding a fragment, id number of the string it originates from). Each string $T^{(t)}$ is then replaced with a concatenation of its distinct fragments with unique delimiters. Thus, $\lambda$ strings of total length $\Oh(n^{2/3}m+\lambda)$ are obtained, and their LCS can be computed in $\Oh(n^{2/3}m+\lambda)$ time~\cite{DBLP:conf/cpm/Hui92}. The overall time complexity is $\Oh(n/m+n^{2/3}m+\lambda)=\Oh(n/\log_\sigma n)$.

\subsection{Long LCS}\label{sec:lambda_long}
Now, we assume that $\ell=\Omega(\frac{\log^{4\lambda} n}{\log^\lambda \sigma})$.

We consider  a $(\lambda,d)$-cover $(D,h)$ for $d \le \ell$ and then take:

\begin{itemize}
    \item $\P^{(t)}:= \{((T^{(t)}[1 \dd i))^R, T^{(t)}[i \dd |T^{(t)}|]) : i \in [1\dd |T^{(t)}|] \cap D\}$ for $t \in [1\dd \lambda]$,
    \item $\F := \{U : (U,V) \in \bigcup_{t=1}^\lambda \P^{(t)} \text{ or } (V,U) \in \bigcup_{t=1}^\lambda \P^{(t)} \text{ for some string }V\}$.
\end{itemize}

Thus, the LCS problem for $\lambda$ strings is reduced to an instance of \problemlambda for $\P^{(1)},\ldots,\P^{(\lambda)}$ and $\F$.
By \cref{thm:lambda_dcover}, $|\P^{(t)}|=\Oh(\lambda |T^{(t)}|/\sqrt[\lambda]{d})$, so $N=|\F|+\sum_{t=1}^\lambda |\P^{(t)}|=\Oh(\lambda \cdot n / \sqrt[\lambda]{d})$.
By \cref{thm:lambda_dcover}, the sets $\P^{(1)},\ldots,\P^{(\lambda)}$ can be constructed in $\Oh(\lambda \cdot n / \sqrt[\lambda]{d})$ time.
Construction of the compacted trie $\TT(\F)$ in $\Oh(N \log N+n / \log_\sigma n)$ time follows the construction for two strings.

The solution is completed by the following lemma. In its proof, we adapt the technique of merging AVL trees from~\cite{DBLP:journals/ipl/FlouriGKU15}. A main modification is that we do not need to use the very efficient merging algorithm of Brown and Tarjan~\cite{DBLP:journals/jacm/BrownT79} in which case we would have to ensure that it can handle a more complex type of queries; instead, we pay an extra $\Oh(\log n)$ factor that stems from using a simple smaller-to-larger trick.

\newcommand{\lcpstr}{\mathsf{LCPStr}}
\newcommand{\iden}{\mathsf{id}}
\newcommand{\result}{\mathsf{result}}
\begin{lemma}\label{lem:problemlambda}
The \problemlambda can be solved in $\Oh(N\lambda\pi_N\log N)$ time and $\Oh(N\lambda)$ space.
\end{lemma}
\begin{proof}
For strings $U_1,\ldots,U_k$, by $\lcpstr(U_1,\ldots,U_k)$ let us denote their longest common prefix. For each node $w$ of $\T(\F)$, we will compute a list of pairs ${(P^{(1)}_1,P^{(1)}_2)\in \P^{(1)},\ldots,(P^{(\lambda)}_1,P^{(\lambda)}_2)\in \P^{(\lambda)}}$ such that $\val(w)$ is a prefix of $\lcpstr(P^{(1)}_1,\ldots,P^{(\lambda)}_1)$ (where $\val(w)$ is the string corresponding to node~$w$) and $\LCP(P^{(1)}_2,\ldots,P^{(\lambda)}_2)$ is maximal; the latter will be denoted as $\result(w)$ (if there is no such list, we set $\result(w)=-\infty$). In the end, we will return the maximum of values $|\val(w)|+\result(w)$ over nodes $w$. Let us observe that the prefix $\lcpstr(P^{(1)}_1,\ldots,P^{(\lambda)}_1)$ corresponds to an explicit node of $\T(\F)$, so it is sufficient to consider explicit nodes $w$.

We will iterate through all explicit nodes of $\T(\F)$ in a bottom-up manner. For each node $w$, we will compute the multiset
\[\L(w)=\{V\,:\,(U,V) \in \bigcup_{t=1}^\lambda \P^{(t)},\,\val(w)\text{ is a prefix of }U\}.\]

\paragraph{Computing the sets $\L(w)$.}
We order the second components of pairs across all the sets $\P^{(t)}$ lexicographically. This can be done in $\Oh(N)$ time by a left-to-right traversal of $\T(\F)$. For each string pair $(U,V)$ in a set~$\P^{(t)}$, we assign to $V$ an integer identifier $\iden(V)$ that indicates the position of $V$ in the sorted list of strings and color $\Color(V)=t$.

The strings in $\L(w)$ will be stored in $\lambda$ dynamic predecessor data structures, one per string color. The strings~$V$ in a predecessor data structure will be stored using their identifiers $\iden(V)$.
We will ensure that, at any time, there is an injection from the elements stored in the predecessor structures to the elements of the multiset obtained as the union of $\P^{(t)}$ for $t \in [1\dd \lambda]$. Hence, the total space required for the predecessor structures is~$\Oh(N\lambda)$.

The multisets $\L(w)$ can be initialised to empty in $\Oh(N\lambda)$ time.
For an explicit node $w$, we compute $\L(w)$ as $\bigcup_{j=1}^r \L(v_j)$ where $v_1,\ldots,v_r$ are all explicit children of $w$; we may then need to insert additional strings if $w$ is a terminal node (that is, we insert strings $V$ from pairs $(U,V)$ from $\bigcup_{t=1}^\lambda \P^{(t)}$ such that $U=\val(w)$). The multisets $\L(v_j)$ are not kept; instead, as usual in applications of the smaller-to-larger trick, the union is computed by choosing the largest among the multisets, say $\L(v_{j'})$, and inserting to it all elements from $\L(v_j)$, for $j \ne j'$ (elements are moved from source predecessor data structure with color $c$ to a target predecessor data structure of the same color), and potentially strings obtained for a terminal node. In total, every element will be moved between predecessor data structures at most $\Oh(\log N)$ times, as the size of its multiset $\L$ at least doubles after a move, and the final multiset $\L$ for the root of $\T(\F)$ has $\Oh(N)$ elements. Thus, the total cost of insertions is $\Oh(N\lambda+N \pi_N\log N)$.

\paragraph{Computing the maximum LCP of a rainbow.}
The multiset $\L(w)$ can be viewed as a multiset of strings with assigned colors. When the multiset $\L(w)$ has been computed, for each string $V$ that was inserted to it (that is, $V \in \L(v_j)$ for some $j \ne j'$), we will compute the maximum $\LCP$ of a $\lambda$-rainbow $\X \subseteq \L(w)$ such that $V \in \X$. We compute $\result(w)$ as the maximum of these $\LCP$s (or $-\infty$ if there is no such $\LCP$) and the values $\result(v_j)$ over all explicit children $v_j$ of $w$.

The maximum $\LCP$ of a $\lambda$-rainbow $\X$ containing a given $V \in \L(w)$ is computed using the formula from \cref{lem:colored_lcp}\eqref{col_a}. Our representation of the multiset $\L(w)$ allows answering each query $\pprev_c(S,\L(w))$ or $\nnext_c(S,\L(w))$ for $S \in \L(w)$ in $\Oh(\pi_N)$ time, the time of an operation on a predecessor/successor data structure of linear size over $[1\dd N]$. The formula takes $\Oh(\lambda \pi_N)$ time to compute the values $\pprev$ and $\nnext$, $\Oh(\lambda \log \log \lambda)$ time to compute the permutation $c_1,\ldots,c_\lambda$ by sorting~\cite{DBLP:journals/jal/Han04}, and $\Oh(\lambda)$ time for the remaining operations; in total, it works in $\Oh(\lambda (\log \log \lambda+\pi_N))=\Oh(\lambda \pi_N)$ time as sorting reduces to dynamic predecessors and $\lambda \le N$. The formula is used $\Oh(N \log N)$ times, whenever a string is inserted into a multiset $\L$. This yields $\Oh(N\lambda \pi_N\log N)$ time. 
\end{proof}

Recall that $N=\Oh(\lambda \cdot n / \sqrt[\lambda]{d})$.
Overall, the time complexity of the solution to long LCS (\cref{lem:problemlambda}) is $\Oh(n \lambda^2 \pi_n \log n / \sqrt[\lambda]{d}+n / \log_\sigma n)$, which is $\Oh(n / \log_\sigma n)$ if $d=\Omega(\frac{\log^{4\lambda} n}{\log^\lambda \sigma})$ and $\lambda=\Oh(\sqrt{\log n})$ (and $\pi_n = \Oh(\log n)$). Clearly, the space is also bounded by $\Oh(n / \log_\sigma n)$ in this case.

\subsection{Medium-length LCS}\label{sec:lambda_medium}
Finally, we consider the case that $\frac13\log_\sigma n\le\ell\le 2^{\Oh(\sqrt{\log n})}$. Then, all possible values of $\ell$ will have been considered, as $\log^{4\lambda} n = 2^{\Oh(\sqrt{\log n})}$ for $\lambda = \Oh(\sqrt{\log n}/\log \log n)$.

The subsets of positions $A_j$ for $j\in \{I,\II,\III\}$ are defined as in \cref{sec:medium} for a string \[T^{(1)}\$ \cdots T^{(\lambda-1)}\$ T^{(\lambda)}\] where $\$ \not\in \Sigma$. We then denote
\[A^{(t)}_j=\{a-\mathit{sum}_{t-1}\,:\,a \in A_j \cap (\mathit{sum}_{t-1}\dd \mathit{sum}_t)\}\]
for $j\in \{I,\II,\III\}$ and $\mathit{sum}_t=|T^{(1)}|+\dots+|T^{(t)}|+t$.
Note that given the sets $A_j$ for $j\in \{I,\II,\III\}$ and the lengths of the strings $T^{(t)}$ for $t \in [1 \dd \lambda]$, we can compute all sets $A^{(t)}_j$ in time $\cO(\lambda + \sum_{j \in \{I,\II,\III\}} |A_j| \log\log n)$ using integer sorting~\cite{DBLP:journals/jal/Han04}. This yields a time complexity of $\cO(n \log\log n /  \log_\sigma n)$.

\cref{lem:anchors} can be directly generalised as follows.

\begin{lemma}\label{lem:anchors_lambda}
If an LCS of $T^{(1)},\ldots,T^{(\lambda)}$ has length $\ell \ge 3\tau$, then there exist positions $i^{(t)} \in [1\dd |T^{(t)}|]$ for $t \in [1\dd \lambda]$, a shift $\delta \in [0\dd  \ell)$, and $j \in \{I,\II,\III\}$ such that all strings $T^{(t)}[i^{(t)} \dd i^{(t)}+\ell)$ for $t \in [1\dd \lambda]$ are equal, $i^{(t)}+\delta\in A^{(t)}_j$ for all $t \in [1\dd \lambda]$, and
\begin{itemize}
    \item if $j=I$, then we can choose $\delta\in [0\dd \tau)$;
    \item if $j=\II$, then $T^{(1)}[i^{(1)} \dd i^{(1)}+\ell)$ is contained in the $\tau$-run from which $i^{(1)}+\delta\in A^{(1)}$ originates;
    \item if $j=\III$, then $\delta \ge 3\tau-2$ and $T^{(1)}[i^{(1)} \dd i^{(1)}+\delta]$ is a suffix of the $\tau$-run from which $i^{(1)}+\delta\in A^{(1)}$ originates.
\end{itemize}
\end{lemma}

The proof of \cref{lem:anchors_lambda} follows the proof of \cref{lem:anchors}; we choose the indices $i^{(1)},\dots,i^{(\lambda)}$ that minimise the sum $i^{(1)}+\dots+i^{(\lambda)}$. Let us note that if $\ell \ge 3\tau$, the equal delimiters do not create any $\tau$-runs in $T^{(1)}\$ \cdots T^{(\lambda-1)}\$ T^{(\lambda)}$ that would not be present in the strings $T^{(1)},\ldots,T^{(\lambda)}$.

Let us now show how to generalise the solutions to the special cases of \problemtwo to $\lambda$ strings.

\begin{lemma}\label{lem:suffix2}
An instance of the \problemlambda in which $\bigcup_{t=1}^{\lambda} \P^{(t)}$ is a prefix family can be solved in $\Oh(N\lambda\log \log N)$ time and $\Oh(N\lambda/\log N + N)$ space.
\end{lemma}
\begin{proof}
We first show how \cref{clm:gamma} from the proof of \cref{lem:suffix} can be adapted to the case of~$\lambda$ input strings.

We construct an array $\R$ that contains all string pairs from $\bigcup_{t=1}^\lambda \P^{(t)}$, sorted by their second components. For a string pair $e=(U,V) \in \P^{(t)}$, where $t \in [1\dd \lambda]$, we denote a list of strings
\[\L_{|U|}=\{V'\,:\,(U',V') \in \R,\,|U'| \ge |U| \}.\]
We assume that strings in $\L_{|U|}$ are sorted lexicographically and each string $V'$ in $\L_{|U|}$ has a color $c \in [1\dd \lambda]$ such that the original string pair $(U',V')$ belonged to $\P^{(c)}$. This allows us to apply the techniques behind \cref{lem:colored_lcp}.

The following claim replaces \cref{clm:gamma}. We assume that $\maxLCP(V,\L_{|U|})=-\infty$ (cf.\ \cref{lem:colored_lcp}) if there is no $\lambda$-rainbow $\X \subseteq \L_{|U|}$ such that $V \in \X$.

\begin{claim}
$\maxMultiLCP(\P^{(1)},\ldots,\P^{(\lambda)}) = \max_{(U,V) \in \R} (|U|+\maxLCP(V,\L_{|U|}))$.
\end{claim}
\begin{proof}
We have $\maxMultiLCP(\P^{(1)},\ldots,\P^{(\lambda)}) \ge \max_{(U,V) \in \R}(|U|+\maxLCP(V,\L_{|U|}))$ since the family $\bigcup_{t=1}^{\lambda} \P^{(t)}$ is a prefix family. We proceed to show the other inequality.

For all $t \in [1\dd \lambda]$, let us pick $(P^{(t)}_1,P^{(t)}_2)\in \P^{(t)}$ such that 
\[\maxMultiLCP(\P^{(1)},\ldots,\P^{(\lambda)})=\LCP(P^{(1)}_1,\ldots,P^{(\lambda)}_1)+\LCP(P^{(1)}_2,\ldots,P^{(\lambda)}_2).\]
As in the proof of \cref{clm:gamma}, we choose an index $t' \in [1\dd \lambda]$ such that $\min\{|P^{(1)}_1|,\ldots,|P^{(\lambda)}_1|\}=|P_1^{(t')}|$. Let us denote $(U,V)=(P^{(t')}_1,P^{(t')}_2)$. Then \[\LCP(P^{(1)}_1,\ldots,P^{(\lambda)}_1)=|P_1^{(t')}|=|U|\text{ and }\LCP(P^{(1)}_2,\ldots,P^{(\lambda)}_2)\le\maxLCP(V,\L_{|U|}).\]
Hence, indeed $\maxMultiLCP(\P^{(1)},\ldots,\P^{(\lambda)}) \le \max_{(U,V) \in \R}(|U|+\maxLCP(V,\L_{|U|}))$.
\end{proof}

\cref{lem:colored_lcp}\eqref{col_a} gives a formula for $\maxLCP(V,\L_{|U|})$ that can be evaluated efficiently if strings $\pprev_c(V,\L_{|U|})$ (and $\nnext_c(V,\L_{|U|})$) for all $c \in [1\dd \lambda]$ are at hand and strings $\pprev_c(V,\L_{|U|})$ can be sorted by $c$. 

We will process all pairs in $\mathcal{R}$ in the order of non-decreasing lengths of the first components. As before, this order can be computed in $\Oh(N)$ time using the compacted trie $\T(\F)$. All elements with equal lengths of the first components are processed simultaneously.

We store integer sets $L^{(1)},\ldots,L^{(\lambda)}$ such that $L^{(t)}$ contains an index $i \in [1\dd |\R|]$ if and only if $\R[i]$ originates from $\P^{(t)}$ and was not processed yet. Each set $L^{(t)}$ is stored in an instance of the dynamic predecessor/successor data structure of \cref{lem:compact_pred}, for a total of $\Oh(N \lambda/\log N + N)$ space.

Consider the processing of a string pair $\R[i]=(U,V)$. At the time of its processing, we have
$$L^{(t)}=\{j \in [1\dd |\R|]\,:\,\R[j]=(U',V'),\,|U'| \ge |U|,\,\Color(V')=t\}$$
for each $t \in [1\dd \lambda]$. Thus, to answer a $\pprev_c(V,\L_{|U|})$ query, it suffices to compute the predecessor of $i$ in $L^{(c)}$. If the predecessor is infinite, $\pprev_c(V,\L_{|U|})=\$$. If the predecessor equals $j<i$, then $\pprev_c(V,\L_{|U|})=V'$ where $\R[j]=(U',V')$.
Values $\nnext_c(V,\L_{|U|})$ can be computed symmetrically using successor queries.
Finally, all strings $\pprev_c(V,\L_{|U|})$, for $c \in [1\dd \lambda]$, can be sorted by sorting the corresponding predecessor values.

Hence, computing $\maxLCP(V,\L_{|U|})$ requires $\Oh(\lambda \log \log N)$ time for asking $\lambda$ predecessor and successor queries (\cref{lem:compact_pred}) and $\Oh(\lambda \log \log \lambda)$ time for sorting strings $\pprev_c(V,\L_{|U|})$, for $c \in [1\dd \lambda]$, represented as integers~\cite{DBLP:journals/jal/Han04}. LCPs on second components of pairs are computed in $\Oh(1)$ time using LCA queries on $\T(\F)$. Overall, this gives $\Oh(N \lambda \log \log N)$ time.
\end{proof}

\medskip
We now generalise the solution to the second special case of \problemlambda.

\begin{lemma}\label{lem:main2}
An instance of the \problemlambda in which $\P^{(t)}$, for all $t \in [1\dd \lambda]$, are $(\alpha,\beta)$-families can be solved in time 
$\cO(N(\alpha+\log N)(\log \beta + \sqrt{\log N})/ \log N)$ and space $\cO(N+N\alpha/\log N)$ provided that $\log \beta = o(\log N)$ and $\lambda = o(\sqrt{\log N})$.
\end{lemma}
\begin{proof}
To obtain \cref{lem:main2}, we use wavelet trees as in the proof of \cref{lem:main} in \cref{sec:proof_lem_main}.

Lists of $\LCPs$ are computed using  \cref{algo:computeLCPs}. The sole difference is that a local variable used for storing the value from stream $G$ now uses $\log \lambda$ bits instead of just one bit. The space grows to $s = \Oh(\log \beta+\log \lambda)$ bits.

\newcommand{\rd}[1]{#1.\texttt{read}()}
 \newcommand{\wrt}[2]{#1.\texttt{write}(#2)}
 \newcommand{\eof}[1]{#1.\mathtt{end}\texttt{\_}\mathtt{of}\texttt{\_}\mathtt{stream}()}

 \begin{figure}[htpb]
\begin{algorithm}[H]
  \caption{Application of $\LCPs$ lists in the case of $\lambda$ input strings.}\label{algo:applyLCPs_lambda}
    $\ell_1 := \cdots := \ell_\lambda := r := 0$\;
    \While{\KwSty{not} $\eof{G}$}{
        $\ell := \rd{\L}$\;
        \lFor{$t:=1$ \KwSty{to} $\lambda$}{$\ell_t := \min(\ell_t, \ell)$}
        $g:=\rd{G}$\;
        $\ell_g := \beta$\;
        $r := \max(r,\,\min(\ell_1,\ldots,\ell_\lambda))$\;
    }
  \Return{$r$}\;
 \end{algorithm}%
\end{figure}

For applying LCPs list, we use \cref{algo:applyLCPs_lambda}. Let $\R_v=(U_1,V_1),\ldots,(U_m,V_m)$ and $\V_v=V_1,\ldots,V_m$ be a list of strings with colors in $[1\dd \lambda]$ that correspond to their origins $G_v$. We next show that, after the $i$th step of the while-loop, we have
\begin{align*}
    \ell_t&=\begin{cases}\LCP(V_i,\pprev_t(V_i,\mathcal{V}_v)) & \text{ if }G[i] \ne t,\\\beta&\text{ otherwise},\end{cases}\\
    r&=\max(\{0\} \cup \{\LCP(V_j,\pprev(V_j,\mathcal{V}_v))\,:\,j \in [1\dd i]\}).
\end{align*}
After the first step of the loop ($\ell=0$), we indeed have $\ell_t=0$ for all $t \in [1\dd \lambda] \setminus \{g\}$ and $\ell_g=\beta$, where $g=G[1]$. Let us consider the $(i+1)$th step of the loop with $g=G[i+1]$. We have $\ell=\LCP(V_i,V_{i+1})\le \beta$, so, for each $t \in [1\dd \lambda] \setminus \{g\}$, the variable $\ell_t$ becomes
\[\LCP(V_i,V_{i+1})\text{ if }G[i]=t\text{ and }\min(\LCP(V_i,\pprev_t(V_i,\mathcal{V}_v)),\LCP(V_i,V_{i+1}))\text{ otherwise}.\]
In either case, $\ell_t=\LCP(V_{i+1},\pprev_t(V_{i+1},\mathcal{V}_v))$; if $G[i] \ne t$, we use \cref{fact:trivial} and the fact that the strings in the list $\V_v$ are ordered lexicographically to argue for this. Then, $\ell_g$ is set to $\beta$, which is an obvious upper bound on $\LCP(V_{i+1},\pprev_g(V_{i+1},\mathcal{V}_v))=\LCP(V_{i+1},V_{i+1})$. In the end, the value of $r$ is maximised with
\[\min\{\LCP(V_{i+1},\pprev_t(V_{i+1},\mathcal{V}_v))\,:\,t \in [1\dd \lambda]\setminus\{g\}\}\ =\ \LCP(V_{i+1},\pprev(V_{i+1},\mathcal{V}_v)).\]
Hence, the value of $r$ is as claimed.
The correctness follows from \cref{lem:colored_lcp}\eqref{col_b}.

\cref{algo:applyLCPs_lambda} works on $\Oh(1)$ data streams, uses $s = \Oh(\lambda\log \beta + \log \lambda) = o(\log n)$ (for $\lambda = o(\sqrt{\log n})$) bits of space, and performs a read or write among every $t = \Oh(\lambda)$ instructions.

For a node $v$ of the wavelet tree, the total size of the input and output streams in bits is proportional to the number of bits in streams $\L_v$ and $G_v$, i.e., $\Oh(|B_v| (\log \beta + \log \lambda))$. By \cref{prp:streaming}, for $\tau=\frac12 \log n$, after $o(N)$ preprocessing, each of the algorithms can be executed in $\Oh(1+|B_v| (\log \beta + \log \lambda) / \log N)$ time. Over all nodes, \cref{algo:computeLCPs,algo:applyLCPs_lambda} consume $\Oh(Nh(\log \beta + \log \lambda)/\log N)$ time, where $h=\Oh(\alpha+\log N)$. As in \cref{lem:main}, the wavelet tree can be built in $\Oh(N h/\sqrt{\log N}+N)$ time using space $\Oh(N h/\log N+N)$. This yields \cref{lem:main2}.
\end{proof}

Finally, \cref{lem:anchors_lambda} leads to an $\cO(n \log\log n /  \log_\sigma n)$-time reduction of computing LCS of $\lambda$ strings if $\frac13\log_\sigma n \le\ell\le 2^{\Oh(\sqrt{\log n})}$ to several instances of \problemlambda of total size $N=\Oh(n/\tau)$, just as in the proof of \cref{lem:medium}. For each instance, either all sets~$\P^{(t)}$ are $(\frac13 \log_\sigma n,2^{\Oh(\sqrt{\log n})})$-families (and then \cref{lem:main2} can be used) or $\bigcup_{t=1}^{\lambda} \P^{(t)}$ is a prefix family (and then \cref{lem:suffix2} is applied). The total time complexity is $\Oh(n \lambda\log\log n\log \sigma/\log n + n\log \sigma/\sqrt{\log n})$, which is $\Oh(n\log \sigma/\sqrt{\log n})$ for $\lambda=\Oh(\sqrt{\log n}/\log \log n)$. The space complexity is $\Oh(n\lambda\log\sigma/\log^2 n+n/\log_\sigma n)$ which is $\Oh(n/\log_\sigma n)$ if $\lambda=\Oh(\log n)$.

Together with the $\Oh(n /\log_\sigma n)$-time algorithms of \cref{sec:lambda_short,sec:lambda_long}, we obtain the following result, which extends \cref{thm:sublinear}.

\begin{theorem}\label{thm:sublinear2}
    Given $\lambda=\Oh(\sqrt{\log n}/\log \log n)$ strings $T^{(1)},\ldots,T^{(\lambda)}$ of total length $n$ over an alphabet $[0\dd \sigma)$, their LCS can be computed in
    $\cO(n\log\sigma/\sqrt{\log n})$ time using $\Oh(n \log\sigma/ \log n)$ space.
\end{theorem}

\section{Faster \texorpdfstring{$k$}{k}-LCS}\label{sec:kLCP}
\newcommand{\MI}{\mathsf{MI}}
\newcommand{\MP}{\mathsf{MP}}
\newcommand{\Fam}{\mathbf{F}}

In this section, we present our $\Oh(n\log^{k-1/2}n)$-time algorithm for the $k$-LCS problem with $k=\Oh(1)$, that underlies \cref{thm:klcs}.
For simplicity, we focus on computing the length of a $k$-LCS; an actual pair of strings forming a $k$-LCS can be recovered easily from our approach.  
If the length of an LCS of $S$ and $T$ is $d$, then the length of a $k$-LCS of $S$ and $T$ is in $[d\dd (k+1)d+k]$. Below, we show how to compute a $k$-LCS provided that it belongs to an interval $(\lceil\ell/2\rceil\dd\ell]$ for a specified $\ell$; to solve the $k$-LCS problem it suffices to call this subroutine $\Oh(\log k)=\Oh(1)$ times.

\begin{definition}
For two strings $U,V\in \Sigma^m$, we define the \emph{mismatch positions} $\MP(U,V) = \{i\in [1\dd  m] :  U[i]\ne V[i]\}$.
Moreover, for $k\in \mathbb{Z}_{\ge 0}$, we write $U=_k V$ if $|\MP(U,V)|\le k$.
\end{definition}

Similarly to our solutions for long and medium-length LCS, we first distinguish anchors $A^S\sub [1\dd |S|]$ in $S$ and $A^T\sub [1\dd |T|]$ in $T$,
as summarised in the following lemma. 
\begin{lemma}\label{prp:anchors}
Consider an instance of the $k$-LCS problem for $k=\Oh(1)$ and let $\ell \in [1\dd n]$. In $\Oh(n)$ time, one can construct sets $A^S\sub [1\dd |S|]$ and $A^T\sub [1\dd |T|]$ of size $\Oh(\frac{n}{\ell})$ satisfying the following condition:
If a $k$-LCS of $S$ and $T$ has length $\ell' \in (\lceil\ell/2\rceil \dd \ell]$, then there exist positions $i^S \in [1\dd |S|]$, $i^T \in [1\dd |T|]$ and a shift  $\delta\in [0\dd \ell')$ such that $i^S+\delta \in A^S$, $i^T+\delta \in A^T$, and  the Hamming distance between $S[i^S\dd i^S+\ell')$ and $T[i^T\dd i^T+\ell')$ is at most $k$.
\end{lemma}
\begin{proof} 
  As in~\cite{DBLP:journals/jcss/Charalampopoulos21}, we say that position $a$ in a string $X$ is a \emph{misperiod} with respect to a substring $X[i \dd j)$ if $X[a] \ne X[b]$, where $b$ is the unique position such that $b \in [i\dd j)$ and ${(j-i) \mid (b-a)}$; for example, $j-i$ is a period of $X$ if and only if there are no misperiods with respect to $X[i\dd j)$.
  We define the set $\LeftMis_k(X,i,j)$ as the set of $k$ rightmost misperiods to the left of $i$ and $\RightMis_k(X,i,j)$ as the set of $k$ leftmost misperiods to the right of $j$.
  Either set may have fewer than $k$ elements if the corresponding misperiods do not exist.
  Further, let us define
  $\Misp_k(X,i,j)=\LeftMis_k(X,i,j) \cup \RightMis_k(X,i,j)$
  and $\Misp(X,i,j) = \bigcup_{k=0}^\infty \Misp_k(X,i,j)$.

  \begin{figure}[htpb]
\centering
\begin{tikzpicture}
  \begin{scope}
    \foreach \x in {-2.8,-2.1,...,7.7}{
        \draw[xshift=\x cm] (0.5*0.35,0.4) .. controls (0.35,0.7) and (2*0.35,0.7) .. (2.5*0.35,0.4);
    }
    \filldraw[white] (-0.175,0.4) rectangle (0.175,0.8);
    \filldraw[white,xshift=-7*0.35 cm] (-0.175,0.4) rectangle (0.175,0.8);
    \filldraw[white,xshift=20*0.35 cm] (-0.175,0.4) rectangle (0.175,0.8);
    \filldraw[white,xshift=23*0.35 cm] (-0.175,0.4) rectangle (0.175,0.8);
    \filldraw[white,xshift=24*0.35 cm] (-0.175,0.4) rectangle (0.175,0.8);
    \foreach \x in {-2.8,7.7}{
        \begin{scope}[xshift=\x cm]
        \clip (0.5*0.35,0.39) rectangle (1.5*0.35,0.8);
        \draw[densely dotted] (0.5*0.35,0.4) .. controls (0.35,0.7) and (2*0.35,0.7) .. (2.5*0.35,0.4);
        \end{scope}
    }
    \foreach \x in {-0.7,6.3}{
        \begin{scope}[xshift=\x cm]
        \clip (1.5*0.35,0.39) rectangle (2.5*0.35,0.8);
        \draw[densely dotted] (0.5*0.35,0.4) .. controls (0.35,0.7) and (2*0.35,0.7) .. (2.5*0.35,0.4);
        \end{scope}
    }
    
    \foreach \x/\c in {-6/a,-5/b,-4/a,-3/b,-2/a,-1/b, 1/b,2/a,3/b,4/a,5/b,6/a,7/b,8/a,9/b,10/a,11/b,12/a,13/b,14/a,15/b,16/a,17/b,18/a,19/b, 21/b,22/a}{
      \draw (\x*0.35,0) node[above] {\texttt{\c}};
    }
    \draw (-8.5*0.35,0) node[above] {$\cdots$};
    \draw (-11*0.35,0) node[above] {$Y=$};
    \draw (24.5*0.35,0) node[above] {$\cdots$};
    \foreach \x/\c in {-7/a,0/c,20/d,23/c}{
      \draw (\x*0.35,0) node[above] {\textcolor{red}{\texttt{\c}}};
    }
    \draw[snake=brace] (19*0.35+0.175,0) -- node[below] {$\tau$-run} (1*0.35-0.175,0);
    \foreach \x in {-7*0.35,0,20*0.35,23*0.35}{
        \draw (\x,-0.3) node {\textcolor{blue}{$\III$}};
    }
    \foreach \x in {-6*0.35,-4*0.35,2*0.35,4*0.35}{
        \draw (\x,-0.3) node {\textcolor{green!50!black}{$\II$}};
    }
  \end{scope}
  \end{tikzpicture}
  \caption{A $\tau$-run for $\tau=6$ with period $p=2$ and Lyndon root \texttt{ab}. For $k=1$ and $i$ being the first position of the $\tau$-run, the set $\Misp_{k+1}(Y,i,i+p)$ contains the letters shown in red. Positions that are added to the sets $A_{\II}$ and $A_{\III}$ because of this $\tau$-run are shown.
  Intuitively, if a length-$\ell$ substring $U$ of~$Y$ starts within $p$ next positions of a position in $\LeftMis_{k+1}(Y,i,i+p)$, it has to contain a position from $A_{\II}$, and otherwise, if~$U$ fits within the approximately periodic substring of~$Y$, the substring $U'$ shifted by $p$ positions to the left contains the same (or smaller) set of misperiods.}\label{fig:A_approx}
\end{figure}
  
  Similar to \cref{lem:anchors}, we construct three subsets of positions in $Y=\#S\$T$, where $\#,\$ \not \in \Sigma$.
  For $\tau=\floor{\ell/(6(k+1))}$, let $A_I$ be a $\tau$-synchronising set of $Y$. 
  Let $Y[i \dd j]$ be a $\tau$-run with period $p$ and assume that the first occurrence of its Lyndon root is at a position $q$ of $Y$.
  Then, for $Y[i \dd j]$, for each $x \in \LeftMis_{k+1}(Y,i,i+p)$, we insert to $A_{\II}$ the two smallest positions in $[x+1\dd |Y|]$ that are equivalent to $q \pmod p$.
  Moreover, we insert to $A_{\III}$ the positions in $\Misp_{k+1}(Y,i,i+p)$.
  See \cref{fig:A_approx}.
  Finally, we denote $A=A_{I}\cup A_{\II}\cup A_{\III}$,
  as well as $A^S=\{a-1\,:\,a \in A \cap [2\dd |S|+1]\}$ and 
  $A^T=\{a-|S|-2\,:\,a \in A \cap [|S|+3\dd |Y|]\}$.
  The proof of the following claim resembles that of \cref{lem:anchors}.

  \begin{restatable}{claim}{anchorsTwo}\label{lem:anchors2}
    The sets $A^S$ and $A^T$ satisfy the condition stated in \cref{prp:anchors}.
  \end{restatable}
  \begin{proof}
By the assumption of \cref{prp:anchors}, a $k$-LCS of $S$ and $T$ has length $\ell'\in (\ell/2,\ell]$.
Consequently, there exist $i^S \in [1\dd |S|]$ and $i^T \in [1\dd |T|]$ such that $U=_kV$, for $U=S[i^S \dd i^S+\ell')$ and $V=T[i^T \dd i^T+\ell')$. Let us choose any such pair $(i^S,i^T)$ minimising the sum $i^S+i^T$.
We shall prove that there exists $\delta \in [0\dd \ell')$ such that $i^S+\delta\in A^S$ and $i^T+\delta \in A^T$.
In each of the three main cases, we actually show that $i^S+\delta\in A^S_j$ and $i^T+\delta \in A^T_j$
for $\delta \in [0\dd \ell')$ and for some $j\in \{I,\II,\III\}$, where $A^S_j=\{a-1\,:\,a \in A_j \cap [2\dd |S|+1]\}$ and
  $A^T_j=\{a-|S|-2\,:\,a \in A_j\cap [|S|+3\dd |Y|]\}$ (recall that $Y=\#S\$T$).

Recall that $\tau=\floor{\ell/(6(k+1))}$. By the pigeonhole principle, there exists a shift $s \in [0\dd \ell'-3\tau+2]$ such that 
\[W:=S[i^S+s \dd i^S+s+3\tau-2] = T[i^T+s \dd i^T+s+3\tau-2].\]
First, assume that $\per(W)>\frac13\tau$. By the definition of a $\tau$-synchronising set, in this case there exist some elements $a^S \in A_I^S \cap [i^S+s\dd i^S+s+\tau)$ and $a^T \in A_I^T \cap [i^T+s\dd i^T+s+\tau)$. Let us choose the smallest such elements. By~\cref{lem:synch}, we have $a^S-i^S=a^T-i^T$.

From now on we consider the case that $p=\per(W)\le\frac13\tau$. 
Thus, each of the two distinguished fragments of $S$ and $T$ that match $W$ belongs to a $\tau$-run with period $p$.

In the case where $\Misp_{k+1}(U,1+s,1+s+p) \cap \Misp_{k+1}(V,1+s,1+s+p) \neq \emptyset$, there exist $a^S \in A^S_{\III}$ and $a^T\in A^T_{\III}$ that satisfy the desired condition. This is because, for any string $Z$ and its run $Z[i\dd j]$ with period $p$ and $i'\in [i\dd j-p+1]$, $\Misp_{k+1}(Z,i,i+p)=\Misp_{k+1}(Z,i',i'+p)$.

In order to handle the complementary case, we rely on the following claim. We recall its short proof for completeness.

\begin{claim}[{\cite[Lemma 13]{DBLP:journals/jcss/Charalampopoulos21}}]\label{claim:circ}
  Assume that $U =_k V$ and that $U[i \dd j) = V[i \dd j)$.
  Let
    \[I=\Misp_{k+1}(U,i,j) \text{ and } I'=\Misp_{k+1}(V,i,j).\]
  If $I \cap I' = \emptyset$, then $\MP(U,V)=I \cup I'$, $I=\Misp(U,i,j)$, and $I'=\Misp(V,i,j)$.
\end{claim}
\begin{proof}
Let $J=\Misp(U,i,j)$ and $J'=\Misp(V,i,j)$.
We first observe that $I \cup I' \subseteq \MP(U,V)$ since $I \cap I' = \emptyset$.
Then, $U =_k V$ implies that $|\MP(U,V)| \leq k$ and hence $|I|\leq k$ and $|I'| \leq k$, which in turn implies that $I=J$ and $I'=J'$.
The observation that $\MP(U,V) \subseteq J \cup J'$ concludes the proof.
\end{proof}

Towards a contradiction, let us suppose that there are no $a^S \in A^S_\II$ and $a^T \in A^T_\II$ satisfying the desired condition.
By \cref{claim:circ}, we have \[|\LeftMis_{k+1}(U,1+s,1+s+p)|,|\LeftMis_{k+1}(V,1+s,1+s+p)|\le k.\]
Therefore, at least one of the following holds: \[i^S \in \LeftMis_{k+1}(\#S,1+i^S+s,1+i^S+s+p)\text{ or }i^T \in \LeftMis_{k+1}(\$T,1+i^T+s,1+i^T+s+p);\]
otherwise, $S[i^S-1]=T[i^T-1]$ and $(i^S-1,i^T-1)$ has a smaller sum and could be used instead of $(i^S,i^T)$; a contradiction.
Let us assume that the first of the two above conditions holds; the other case is symmetric.
Then, $[i^S\dd i^S+p)$ contains an element of $A^S_\II$, say $a^S$.
If $A^T_\II \cap [i^T\dd i^T+p) = \emptyset$, then this implies that
\[[i^T-p+1\dd i^T]\cap \LeftMis_{k+1}(\$T,1+i^T+s,1+i^T+s+p) = \emptyset.\]
In particular, we have $i^T>p$.
Now, let us consider $U$ and $V'=T[i^T-p \dd i^T+\ell'-p)$.
By \cref{claim:circ}, we have that $|\Misp(U,1+s,1+s+p)|+|\Misp(V,1+s,1+s+p)|\leq k$.
Further, we have $|\Misp(V',1+s+p,1+s+2p)|\leq |\Misp(V,1+s,1+s+p)|$.
Thus, $|\MP(U,V')|\le |\Misp(U,1+s,1+s+p)|+|\Misp(V',1+s,1+s+p)|\leq k$.
This contradicts our assumption that $i^S+i^T$ was minimum possible.

Consequently, there exists an element $a^T \in A_{\II}^T \cap [i^T\dd i^T+p)$. By modular arithmetics, we have $\delta:=a^S-i^S=a^T-i^T$. Then $\delta$ certainly belongs to $[0\dd \ell')$, which concludes the proof.
\end{proof}

  It remains to show that the sets $A^S$ and $A^T$ can be constructed efficiently.
  A $\tau$-synchronising set can be computed in $\Oh(n)$ time by~\cref{thm:synch_packed} and all the $\tau$-runs, together with the position of the first occurrence of their Lyndon root, can be computed in $\Oh(n)$ time~\cite{DBLP:journals/siamcomp/BannaiIINTT17}.
  After an $\cO(n)$-time preprocessing, for every $\tau$-run, we can compute the set of the $k+1$ misperiods of its period to either side in $\cO(1)$ time; see~\cite[Claim 18]{DBLP:journals/jcss/Charalampopoulos21}.
\end{proof}

The next step in our solutions to long LCS and medium-length LCS was to construct an instance of the \problemtwo.
To adapt this approach, we generalise the notions of $\LCP$ and $\maxPairLCP$ so that they allow for mismatches.
By $\LCP_k(U,V)$, for $k\in \mathbb{Z}_{\ge 0}$, we denote the maximum length $\ell$
such that $U[1\dd \ell]$ and $V[1\dd \ell]$ are at Hamming distance at most $k$.
\begin{definition}
  Given two sets $\U,\V\sub \Sigma^*\times \Sigma^*$ and two integers $k_1,k_2\in \mathbb{Z}_{\ge 0}$,
  we define \[\maxPairLCP_{k_1,k_2}(\U,\V)=\max\{\LCP_{k_1}(U_1,V_1)+\LCP_{k_2}(U_2,V_2) : (U_1,U_2)\in \U,(V_1,V_2)\in \V\}.\]
\end{definition}
\noindent Note that $\maxPairLCP(\U,\V)=\maxPairLCP_{0,0}(\U,\V)$.

By \cref{prp:anchors}, if a $k$-LCS of $S$ and $T$ has length $\ell' \in (\ell/2,\ell]$,
then
\begin{align*}
    \ell' = \max_{k'=0}^{k} \maxPairLCP_{k',k-k'}(\U,\V), \text{ for } & \U=\{((S[a-\ell \dd a))^R,S[a \dd a+\ell)) : a \in A^S\},\\ & \V = \{((T[a-\ell \dd a))^R,T[a \dd a+\ell)) : a \in A^T\}.
\end{align*}
Here, $k'$ bounds the number of mismatches between $S[i^S\dd i^S+\delta)$ and $T[i^T\dd i^T+\delta)$,
whereas $k-k'$ bounds the number of mismatches between $S[i^S+\delta \dd i^S+\ell')$ and  $T[i^T+\delta \dd i^T+\ell')$.
The following theorem, whose full proof is given in \cref{app:lcpk},
allows for efficiently computing the values $\maxPairLCP_{k',k-k'}(\U,\V)$. 
  \begin{restatable}{theorem}{maxpairlcpk}\label{thm:maxpairlcpk}
    Consider two $(\ell,\ell)$-families $\U,\V$ of total size $N$
    consisting of pairs of substrings of a given length-$n$ text over the alphabet $[0\dd n)$.
    For any non-negative integers $k_1,k_2 = \Oh(1)$,
    the value $\maxPairLCP_{k_1,k_2}(\U,\V)$ can be computed:
    \begin{itemize}
      \item in $\Oh(n+N\log^{k_1+k_2+1} N)$ time and $\Oh(n+N)$ space if $\ell > \log^{3/2} N$,
      \item in $\Oh(n + N \ell \log^{k_1+k_2-1/2} N)$ time and $\Oh(n+N\ell/\log N)$ space if $\log N < \ell \le \log^{3/2} N$,
      \item in $\Oh(n + N \ell^{k_1+k_2}\sqrt{\log N})$ time and $\Oh(n+N)$ space if $\ell \le \log N$.
    \end{itemize}
  \end{restatable}
  Here we just give an outline of the proof.
    We reduce the computation of $\maxPairLCP_{k_1,k_2}(\U,\V)$ into multiple computations of $\maxPairLCP(\U',\V')$ across a family $\Pam$ of pairs $(\U',\V')$ with $\U',\V'\sub \Sigma^{*}\times \Sigma^{*}$.
    Each pair $(U'_1,U'_2)\in \U'$ is associated to a pair $(U_1,U_2)\in \U$, with the string~$U'_i$ represented as a pointer to the \emph{source} $U_i$ and up to $k_i$ substitutions needed to transform $U_i$ to~$U'_i$.
    Similarly, each pair $(V'_1,V'_2)\in \V'$ consists of \emph{modified strings} with sources $(V_1,V_2)\in \V$.
    In order to guarantee $\maxPairLCP_{k_1,k_2}(\U,\V) = \max_{(\U',\V')\in \Pam} \maxPairLCP(\U',\V')$,
    we require $\LCP(U'_i,V'_i)\le \LCP_{k_i}(U_i,V_i)$ for every $(U'_1,U'_2)\in \U'$ and $(V'_1,V'_2)\in \V'$ with $(\U',\V')\in \Pam$ and that, for every $(U_1,U_2)\in \U$ and $(V_1,V_2)\in \V$, there exists $(\U',\V')\in \Pam$ with $(U'_1,U'_2)\in \U'$ and $(V'_1,V'_2)\in \V'$, with sources $(U_1,U_2)$ and $(V_1,V_2)$, respectively, such that $\LCP(U'_i,V'_i)=\LCP_{k_i}(U_i,V_i)$.
    Our construction is based on a technique of~\cite{DBLP:journals/jcb/ThankachanAA16} which gives an analogous family for two subsets of $\Sigma^*$ (rather than $\Sigma^*\times \Sigma^*$) and a single threshold.
    We apply the approach of~\cite{DBLP:journals/jcb/ThankachanAA16} to $\U_i = \{U_i : (U_1,U_2)\in \U\}$ and $\V_i = \{V_i : (V_1,V_2)\in \V\}$ with threshold $k_i$, and then combine the two resulting families to derive~$\Pam$. 
    
    Strengthening the arguments of~\cite{DBLP:journals/jcb/ThankachanAA16}, we show that each string $F_i\in \U_i\cup \V_i$ is the source of $\Oh(1)$ modified strings $F'_i\in \U'_i\cup \V'_i$ for any single $(\U'_i,\V'_i)\in \Pam_i$
    and $\Oh(\min(\ell,\log N)^{k_i})$ modified strings across all $(\U'_i,\V'_i)\in \Pam_i$.
    This allows bounding the size of individual sets $(\U',\V')\in \Pam$ by $\Oh(N)$ 
    and the overall size by $\Oh(N\min(\ell,\log N)^{k_1+k_2})$. 
    In order to efficiently build the compacted tries required at the input of the \problemtwo, the modified strings  $F'_i\in \U'_i\cup \V'_i$ are sorted lexicographically, and the two derived linear orders (for $i\in \{1,2\}$) are maintained along with every pair $(\U',\V')\in \Pam$.
    Overall, the family $\Pam$ is constructed in $\Oh(n+N\min(\ell,\log N)^{k_1+k_2})$ time and $\Oh(n+N)$ space.

    We solve the resulting instances of the \problemtwo using \cref{lem:problem} if $\ell > \log^{3/2} N$ or \cref{lem:main} otherwise; note that $\U',\V'$ are $(\ell,\ell)$-families.

\medskip
Recall that the algorithm of \cref{thm:maxpairlcpk} is called $k+1=\Oh(1)$ times, always with $N=|A^S|+|A^T|=\cO(n/\ell)$. Overall, the value $ \max_{k'=0}^{k} \maxPairLCP_{k',k-k'}(\U,\V)$ is therefore computed in $\Oh(n\log^{k-1/2}n)$ time and $\Oh(n)$ space in each of the following cases:
    \begin{itemize}
      \item in $\Oh(n+\frac{n}{\ell}\log^{k+1}N)=\Oh(n\log^{k-1/2}n)$ time and $\Oh(n+\frac{n}{\ell})=\cO(n)$ space if $\ell > \log^{3/2} N$;
      \item in $\Oh(n + \frac{n}{\ell} \ell \log^{k-1/2}N)=\Oh(n\log^{k-1/2}n)$ time and $\Oh(n+\frac{n}{\ell} \ell/\log N)=\Oh(n)$ space if $\log N < \ell \le \log^{3/2} N$;
      \item in $\Oh(n + \frac{n}{\ell}\ell^{k}\sqrt{\log N})=\Oh(n\log^{k-1/2}n)$ time and $\Oh(n+\frac{n}{\ell})=\Oh(n)$ space if $\ell \le \log N$.
    \end{itemize}
Accounting for $\Oh(n)$ time and space to determine the length $d$ of an LCS between $S$ and $T$, and the $\Oh(\log k)$ values $\ell$ that need to be tested so that the intervals $(\lceil\ell/2\rceil\dd \ell]$ cover $[d\dd (k+1)d+k]$, this implies \cref{thm:klcs}.

We also obtain the following corollary that improves~\cite{DBLP:conf/cpm/Charalampopoulos18} for $k=\Oh(1)$. (We replace $\sqrt{\ell}$ by $\ell$.)

\begin{corollary}\label{cor:CPMbetter}
Given two strings $S$ and $T$ of total length $n$ and a constant integer $k>0$, the $k$-LCS problem can be solved in $\cO(n+\frac{n}{\ell} \log^{k+1} n)$ time using $\cO(n)$ space, where $\ell$ is the length of the $k$-LCS.
\end{corollary}
\begin{proof}
The three cases yield the following time complexities:
in $\Oh(n+\frac{n}{\ell}\log^{k+1}n)$ time if $\ell > \log^{3/2} N$, in $\Oh(n\log^{k-1/2}n) = \Oh(\frac{n}{\ell}\log^{k+1}n)$ time if $\log N < \ell \le \log^{3/2} N$, and in $\Oh(n\log^{k-1/2}n)=\Oh(\frac{n}{\ell}\log^{k+1}n)$ time if $\ell \le \log N$.
This gives an $\Oh(n+\frac{n}{\ell}\log^{k+1}n)$-time solution for any $\ell$. The space complexity is still $\Oh(n)$.
\end{proof}

\section{Computing \texorpdfstring{$\maxPairLCP_{k_1,k_2}(\U,\V)$}{maxPairLCP(k1,k2)(U,V)}---Proof of \texorpdfstring{Theorem~\ref{thm:maxpairlcpk}}{Theorem 6.2}}\label{app:lcpk}
For two strings $U,U'\in \Sigma^m$, let us define the set of \emph{mismatch positions} between $U$ and $U'$, ${\MP(U,U') = \{i\in [1\dd  n] :  U[i]\ne U'[i]\}}$, and the \emph{mismatch information}
$\MI(U,U') = \{(i,U'[i]) : i\in \MP(U,U')\}$.
Observe that $U$ and $\MI(U,U')$ uniquely determine $U'$.
This motivates the following definition.

\begin{definition}
Given $U\in \Sigma^m$ and $\Delta \sub [1\dd  m]\times \Sigma$, we denote by $U^\Delta$
the unique string $U'$ such that $\MI(U,U')=\Delta$. If there is no such string $U'$, then $U^\Delta$ is undefined.
We say that $U^\Delta$, represented using a pointer to $U$ and the set $\Delta$, is a \emph{modified string} with \emph{source} $U$.
\end{definition}

\begin{example}
Let $U=\mathtt{ababbab}$ and $\Delta=\{(2,\mathtt{a}),(3,\mathtt{b})\}$.
Then $U^{\Delta}=U'=\mathtt{aabbbab}$.
\end{example}

\begin{definition}[see~\cite{DBLP:journals/jcb/ThankachanAA16}]\label{def:pair}
Given strings $U,V\in \Sigma^*$ and an integer $k\in \mathbb{Z}_{\ge 0}$,
we say that two modified strings $(U^\Delta, V^\nabla)$ form a $(U,V)_k$-\emph{maxpair} if the following holds for every $i$:
\begin{itemize}
  \item if $i\in [1\dd  \LCP_k(U,V)]$ and $U[i]\ne V[i]$, then $U^\Delta[i]=V^\nabla[i]$
  \item otherwise, $U^\Delta[i]=U[i]$ (assuming $i\in [1\dd  |U|]$) and $V^\nabla[i]=V[i]$ (assuming $i\in [1\dd  |V|]$).
\end{itemize}
\end{definition}

\begin{example}\label{ex:maxpairs}
Let $U=\mathtt{ababbabb}$, $V=\mathtt{aacbaaab}$ and $k=3$.
Further let $\Delta=\{(2,\mathtt{a}),(3,\mathtt{b})\}$
and $\nabla=\{(3,\mathtt{b}),(5,\mathtt{b}))\}$.
We have $U^\Delta=\mathtt{aabbbabb}$ and $U^\nabla=\mathtt{aabbbaab}$.
Then $\LCP_k(U,V)=6$ and $(U^\Delta, V^\nabla)=(\mathtt{\textcolor{red}{\underline{aabbba}}ab},\mathtt{\textcolor{red}{\underline{aabbba}}bb})$ form a $(U,V)_3$-maxpair.
Additionally,  for $\Delta'=\{(3,\mathtt{c})\}$
and $\nabla'=\{(2,\mathtt{b}),(5,\mathtt{b}))\}$, $(U^{\Delta'}, V^{\nabla'})=(\mathtt{\textcolor{red}{\underline{abcbba}}ab},\mathtt{\textcolor{red}{\underline{abcbba}}bb})$ form a $(U,V)_3$-maxpair.
\end{example}

The following simple fact characterises this notion.

\begin{fact}\label{fct:pair}
Let $U^\Delta,V^\nabla$ be modified strings with sources $U,V\in \Sigma^*$ and let $k\in \mathbb{Z}_{\ge 0}$.
\begin{enumerate}[(a)]
  \item\label{it:pair:a} If $(U^\Delta,V^\nabla)$ is a $(U,V)_k$-maxpair, then $|\Delta \cup \nabla|\le k$ and $\LCP(U^\Delta,V^\nabla) \ge \LCP_k(U,V)$.
  \item\label{it:pair:b} If $|\Delta \cup \nabla|\le k$, then $\LCP_k(U,V)\ge \LCP(U^\Delta,V^\nabla)$.
\end{enumerate}
\end{fact}
\begin{proof}
\eqref{it:pair:a} Let $d=\LCP_k(U,V)$ and $M=\MP(U[1\dd d],V[1\dd d])$.
By~\cref{def:pair}, we have $U^\Delta[i] = V^\nabla[i]$ for $i \in [1\dd  d]$. Consequently, $\LCP(U^\Delta,V^\nabla)\ge d$ and, if $(i,a)\in \Delta$ and $(i,b)\in \nabla$ holds for some $i\in [1\dd  d]$, then $a=b$.
Furthermore,~\cref{def:pair} yields $\Delta,\nabla \sub M\times \Sigma$, and hence $|\Delta \cup \nabla|\le |M| \le k$.

\noindent\eqref{it:pair:b}
Let $d'=\LCP(U^\Delta,V^\nabla)$ and $M' = \MP(U[1\dd d'], V[1\dd d'])$.
For every $i\in M'$, we have $U[i]\ne V[i]$ yet $U^\Delta[i]=V^\nabla[i]$.
This implies $U[i]\ne U^\Delta[i]$ or $V[i]\ne V^\nabla[i]$,
i.e., $(i,U^\Delta[i])=(i,V^\nabla[i])\in \Delta\cup \nabla$.
Consequently, $|M'|\le |\Delta\cup \nabla|\le k$, which means that $\LCP_k(U,V)\ge d'$ holds as claimed.
\end{proof}

In particular, if $(U^\Delta,V^\nabla)$ is a $(U,V)_k$-maxpair, then by \cref{fct:pair}, $\LCP(U^\Delta,V^\nabla) = \LCP_k(U,V)$.

\begin{definition}\label{def:complete}
Consider a set of strings $\F\sub \Sigma^*$ and an integer $k\in \mathbb{Z}_{\ge 0}$.
A \emph{$k$-complete family} for $\F$ is a family $\Fam$ of sets of modified strings of the form $F^\Delta$ for $F\in \F$ and $|\Delta|\le k$ such that, for every $U,V\in \F$, there exists a set $\F'\in \Fam$
and modified strings $U^\Delta,V^\nabla\in \F'$ forming a $(U,V)_k$-maxpair.
\end{definition}

\begin{example}\label{ex:F1}
Let us consider the following string family:
\[\F_1=\{A_1=\mathtt{aaaab},\,
         A_2=\mathtt{abbac},\,
         A_3=\mathtt{acbab},\,
         A_4=\mathtt{bcaa},\,
         A_5=\mathtt{bcbac}
      \}.\]
Family $\Fam_1=\{\F'_{1,1},\F'_{1,2}\}$ is a 1-complete family for $\F_1$, where:

\vspace*{-0.4cm}
\twocol{
\begin{align*}
\F'_{1,1}=\{\quad\quad
         A_1&=\mathtt{aaaab},\\
         A_3^{\{(2,\mathtt{a})\}}&=\mathtt{a\textcolor{red}{a}bab},\\
         A_2^{\{(2,\mathtt{a})\}}&=\mathtt{a\textcolor{red}{a}bac},\\
         A_2&=\mathtt{abbac},\\
         A_4^{\{(1,\mathtt{a})\}}&=\mathtt{\textcolor{red}{a}caa},\\
         A_3&=\mathtt{acbab},\\
         A_5^{\{(1,\mathtt{a})\}}&=\mathtt{\textcolor{red}{a}cbac},\\
         A_4&=\mathtt{bcaa},\\
         A_5&=\mathtt{bcbac}
\,\}.
\end{align*}
}{
\begin{align*}
\F'_{1,2}=\{\quad\quad
    A_4&=\mathtt{bcaa},\\
    A_5^{\{(3,\mathtt{a})\}}&=\mathtt{bc\textcolor{red}{a}ac},\\
    A_5&=\mathtt{bcbac}
\,\}.
\end{align*}
}
For example, both strings in a $(A_2,A_3)_1$-maxpair $(A_2^{\{(2,\mathtt{a})\}}=\mathtt{a\textcolor{red}{a}bac},A_3^{\{(2,\mathtt{a})\}}=\mathtt{a\textcolor{red}{a}bab})$ belong to~$\F'_{1,1}$, and both strings in a $(A_4,A_5)_1$-maxpair $(A_4=\mathtt{bcaa},A_5^{\{(3,\mathtt{a})\}}=\mathtt{bc\textcolor{red}{a}ac})$ belong to $\F'_{1,2}$.

\cref{fig:HLD} in \cref{app:complete} shows how the family $\Fam_1$ was created.
\end{example}

\newcommand{\Y}{\mathcal{Y}}
\begin{example}\label{ex:thanks}
Let $\X$ and $\Y$ be the sets of all suffixes of strings $S$ and $T$, respectively, and $\Fam$ be a $k$-complete family for $\X \cup \Y$. Then
\[k\text{-LCS}(S,T) = \max_{\F'\in\Fam}\{\LCP(X^\Delta,Y^\nabla)\,:\,X \in \X,\, Y \in \Y,\,|\Delta\cup\nabla|\le k,\,X^\Delta,Y^\nabla \in \F'\}.\]
\end{example}

Our construction of a $k$-complete family follows the approach of Thankachan et al.~\cite{DBLP:journals/jcb/ThankachanAA16} (which, in turn, is based on the ideas behind $k$-errata trees of Cole et al.~\cite{DBLP:conf/stoc/ColeGL04})
with minor modifications.
For completeness, we provide a full proof of the following proposition in \cref{app:complete}.

\begin{restatable}{proposition}{prpThankachan}\label{prp:complete}
Let $\F\sub \Sigma^{\le \ell}$ and $k\in \mathbb{Z}_{\ge 0}$ with $k=\Oh(1)$.
There exists a $k$-complete family $\Fam$ for~$\F$ such that, for each $F\in \F$:
\begin{itemize}
  \item Every individual set $\F'\in \Fam$ contains $\Oh(1)$ modified strings with source $F$.
  \item In total, the sets $\F'\in \Fam$ contain $\Oh(\min(\ell,\log|\F|)^k)$ modified strings with source $F$.
\end{itemize}
Moreover, if $\F$ consists of substrings of a given length-$n$ text over the alphabet $[0\dd n)$,
then the family~$\Fam$ can be constructed in $\Oh(n+|\F|)$ space and $\Oh(n + |\F|\min(\ell,\log|\F|)^k)$ time
with sets $\F'\in \Fam$ generated one by one and modified strings within each set $\F'\in \Fam$ sorted lexicographically.    
\end{restatable}

Intuitively, in the approach of Thankachan et al.~\cite{DBLP:journals/jcb/ThankachanAA16}, a $k$-LCS was computed as the maximum $\LCP_k$ of any two suffixes originating from different strings $S$, $T$. 
Hence, using the $k$-complete family shown in \cref{ex:thanks} was sufficient. However, in our approach, a $k$-LCS is anchored at some pair of synchronised positions. This motivates the following generalised notion aimed to account for the parts of a $k$-LCS on both sides of the anchor.

\begin{definition}\label{def:bicomplete}
Consider a set $\G\sub \Sigma^* \times \Sigma^*$ of string pairs and integers $k_1,k_2\in \mathbb{Z}_{\ge 0}$.
A \emph{$(k_1,k_2)$-bicomplete family} for $\G$ is a family $\Gam$ of sets $\G'$ of modified string pairs
of the form $(F_1^{\Delta_1},F_2^{\Delta_2})$ for $(F_1,F_2)\in \G$, $|\Delta_1|\le k_1$, and $|\Delta_2|\le k_2$, such that, for every $(U_1,U_2),(V_1,V_2)\in \G$,
there exists a set $\G'\in \Gam$ with $(U_1^{\Delta_1},U^{\Delta_2}_2),(V^{\nabla_1}_1,V^{\nabla_2}_2)\in \G'$
such that $(U_1^{\Delta_1},V^{\nabla_1}_1)$ is a $(U_1,V_1)_{k_1}$-maxpair and $(U^{\Delta_2}_2,V^{\nabla_2}_2)$ is a $(U_2,V_2)_{k_2}$-maxpair.
\end{definition}

A concrete example of a bicomplete family (\cref{ex:G}) is presented after an efficient construction of such a family (\cref{lem:bicomplete}).

\begin{example}\label{ex:mpLCP}
Let $\U$ and $\V$ be the sets of string pairs as considered in \cref{thm:maxpairlcpk} and $\Gam$ be a $(k_1,k_2)$-bicomplete family for $\U \cup \V$. Then
\begin{multline*}\maxPairLCP_{k_1,k_2}(\U,\V) = \max_{\G'\in\Gam}\{\LCP(U_1^{\Delta_1},V_1^{\nabla_1})+\LCP(U_2^{\Delta_2},V_2^{\nabla_2}):\\
(U_1,U_2) {\in} \U, (V_1,V_2) {\in} \V,|\Delta_1\cup\nabla_1|\le k_1,|\Delta_2\cup\nabla_2|\le k_2,\,(U_1^{\Delta_1},U_2^{\Delta_2}),(V_1^{\nabla_1},V_2^{\nabla_2}){\in} \G'\}.
\end{multline*}
\end{example}

\noindent
Complete families can be used to efficiently construct a bicomplete family as shown in the following lemma.

\begin{lemma}\label{lem:bicomplete}
  Let $\G$ be an $(\ell,\ell)$-family and $k_1,k_2\in \mathbb{Z}_{\ge 0}$ with $k_1,k_2=\Oh(1)$.
  There exists a $(k_1,k_2)$-bicomplete family $\Gam$ for $\G$ such that, for each $(F_1,F_2)\in \G$:
  \begin{itemize}
    \item Every individual set $\G'\in \Gam$ contains $\Oh(1)$ pairs of the form $(F_1^{\Delta_1},F_2^{\Delta_2})$.
    \item In total, the sets $\G'\in \Gam$ contain $\Oh(\min(\ell,\log |\G|)^{k_1+k_2})$ pairs of the form $(F_1^{\Delta_1},F_2^{\Delta_2})$.
  \end{itemize}
\end{lemma}
\begin{proof}
  Let $\F_1 = \{F_1 : (F_1,F_2)\in \G\}$ and $\F_2 = \{F_2 : (F_1,F_2)\in \G\}$.
  Moreover, for $i\in \{1,2\}$, let $\Fam_i$ be the $k_i$-complete family for $\F_i$ obtained using~\cref{prp:complete}. The family $\Gam$ is defined as follows:
  \[\Gam = \{\{(F^{\Delta_1}_1,F^{\Delta_2}_2)\in \F'_1\times \F'_2 : (F_1,F_2)\in \G\} : (\F'_1,\F'_2)\in \Fam_1\times \Fam_2\}\sm \{\emptyset\}.\]

  Clearly, each set $\G'\in \Gam$ consists of modified string pairs of the form $(F^{\Delta_1}_1,F^{\Delta_2}_2)$ with $(F_1,F_2)\in \G$, $|\Delta_1|\le k_1$, and $|\Delta_2|\le k_2$.
  Let us fix pairs $(U_1,U_2),(V_1,V_2)\in \G$.
  By~\cref{def:complete}, for $i\in \{1,2\}$, there exists a set $\F'_i \in \Fam_i$ and modified strings $U_i^{\Delta_i},V_i^{\nabla_i}\in \F'_i$ that form a $(U_i,V_i)_{k_i}$-maxpair. 
  The family $\Gam$ contains a set $\G' = \{(F^{\Delta_1}_1,F^{\Delta_2}_2)\in \F'_1\times \F'_2 : (F_1,F_2)\in \G\}$ and this set $\G'$ contains both $(U_1^{\Delta_1},U_2^{\Delta_2})$ and $(V^{\nabla_1}_1,V^{\nabla_2}_2)$.
  Thus, $\Gam$ is a $(k_1,k_2)$-bicomplete family for $\G$.

  Now, let us fix a pair $(F_1,F_2)\in \G$ in order to bound the number of modified string pairs of the form $(F^{\Delta_1}_1,F^{\Delta_2}_2)$ contained in the sets $\G'\in \Gam$.
  If $(F^{\Delta_1}_1,F^{\Delta_2}_2)\in \G'$, where $\G'$ is constructed for $(\F'_1,\F'_2)\in \Fam_1\times \Fam_2$, then $F^{\Delta_i}_i\in \F'_i$. 
  By~\cref{prp:complete}, for $i\in \{1,2\}$, the set $\F'_i\in \Fam_i$ contains $\Oh(1)$ modified strings with the source $F_i$. Consequently, the set $\G'\in \Gam$ contains $\Oh(1)$ modified string pairs of the form $(F^{\Delta_1}_1,F^{\Delta_2}_2)$.
  Moreover,~\cref{prp:complete} implies that, for $i\in \{1,2\}$, the sets $\F'_i\in \Fam_i$ in total contain $\Oh(\min(\ell,\log |\F_i|)^{k_i})=\Oh(\min(\ell,\log |\G|)^{k_i})$ modified strings of the form~$F^{\Delta_i}_i$.
  Thus, the sets $\G'\in \Gam$ in total contain $\Oh(\min(\ell,\log |\G|)^{k_1+k_2})$ modified string pairs of the form $(F^{\Delta_1}_1,F^{\Delta_2}_2)$.
\end{proof}

\begin{example}\label{ex:G}
Let us consider the following string family
\[\F_2=\{B_3=\mathtt{aacba},\,
         B_2=\mathtt{abcba},\,
         B_1=\mathtt{babab},\,
         B_4=\mathtt{babc},\,
         B_5=\mathtt{bbcba}
      \}.\]
Family $\Fam_2=\{\F'_{2,1},\F'_{2,2}\}$ is a 1-complete family for $\F_2$, where:

\vspace*{-0.4cm}
\twocol{
\begin{align*}
\F'_{2,1}=\{\quad\quad
         B_3&=\mathtt{aacba},\\
         B_2&=\mathtt{abcba},\\
         B_4^{\{(4,\mathtt{a})\}}&=\mathtt{bab\textcolor{red}{a}},\\
         B_1&=\mathtt{babab},\\
         B_4&=\mathtt{babc},\\
         B_3^{\{(1,\mathtt{b})\}}&=\mathtt{\textcolor{red}{b}acba},\\
         B_5^{\{(2,\mathtt{a})\}}&=\mathtt{b\textcolor{red}{a}cba},\\
         B_5&=\mathtt{bbcba},\\
         B_2^{\{(1,\mathtt{b})\}}&=\mathtt{\textcolor{red}{b}bcba}
\,\}.
\end{align*}
}{
\begin{align*}
\F'_{2,2}=\{\quad\quad
    B_3&=\mathtt{aacba},\\
    B_2&=\mathtt{abcba},\\
    B_3^{\{(2,\mathtt{b})\}}&=\mathtt{a\textcolor{red}{b}cba}
\,\}.
\end{align*}
}

Now let us consider the following $(5,5)$-family, where $\F_1=\{A_1,A_2,A_3,A_4,A_5\}$ was considered in \cref{ex:F1}:
\begin{align*}
\G=\{\,&
(A_1=\mathtt{aaaab},B_1=\mathtt{babab}),\,
(A_2=\mathtt{abbac},B_2=\mathtt{abcba}),\,
(A_3=\mathtt{acbab},B_3=\mathtt{aacba}),\\
&(A_4=\mathtt{bcab},B_4=\mathtt{babc}),\,
(A_5=\mathtt{bcbac},B_5=\mathtt{bbcba})
\,\}.
\end{align*}
Family $\Gam=\{\G'_{1,1},\G'_{1,2},\G'_{2,1}\}$ constructed as in \cref{lem:bicomplete} using $\Fam_1$ and $\Fam_2$ is a $(1,1)$-bicomplete family for $\G$, where:
\begin{align*}
\G'_{1,1}=\,&\{(A_1=\mathtt{aaaab},B_1=\mathtt{babab})\}\ \cup\\
&\{A_2=\mathtt{abbac},A_2^{\{(2,\mathtt{a})\}}=\mathtt{a\textcolor{red}{a}bac}\} \times \{B_2=\mathtt{abcba},B_2^{\{(1,\mathtt{b})\}}=\mathtt{\textcolor{red}{b}bcba}\}\ \cup\\
&\{A_3=\mathtt{acbab},A_3^{\{(2,\mathtt{a})\}}=\mathtt{a\textcolor{red}{a}bab}\} \times \{B_3=\mathtt{aacba},B_3^{\{(1,\mathtt{b})\}}=\mathtt{\textcolor{red}{b}acba}\}\ \cup\\
&\{A_4=\mathtt{bcaa},A_4^{\{(1,\mathtt{a})\}}=\mathtt{\textcolor{red}{a}caa}\} \times \{B_4=\mathtt{babc},B_4^{\{(4,\mathtt{a})\}}=\mathtt{bab\textcolor{red}{a}}\}\ \cup \\
&\{A_5=\mathtt{bcbac},A_5^{\{(1,\mathtt{a})\}}=\mathtt{\textcolor{red}{a}cbac}\} \times \{B_5=\mathtt{bbcba},B_5^{\{(2,\mathtt{a})\}}=\mathtt{b\textcolor{red}{a}cba}\},\\
\G'_{1,2}=\,&\{A_2=\mathtt{abbac},A_2^{\{(2,\mathtt{a})\}}=\mathtt{a\textcolor{red}{a}bac}\} \times \{B_2=\mathtt{abcba}\}\ \cup\\
&\{A_3=\mathtt{acbab},A_3^{\{(2,\mathtt{a})\}}=\mathtt{a\textcolor{red}{a}bab}\} \times \{B_3=\mathtt{aacba},B_3^{\{(2,\mathtt{b})\}}=\mathtt{a\textcolor{red}{b}cba}\},\\
\G'_{2,1}=\,&\{A_4=\mathtt{bcaa}\} \times \{B_4=\mathtt{babc},B_4^{\{(4,\mathtt{a})\}}=\mathtt{bab\textcolor{red}{a}}\}\ \cup\\
&\{A_5=\mathtt{bcbac},A_5^{\{(3,\mathtt{a})\}}=\mathtt{bc\textcolor{red}{a}ac}\} \times \{B_5=\mathtt{bbcba},B_5^{\{(2,\mathtt{a})\}}=\mathtt{b\textcolor{red}{a}cba}\}
\end{align*}
Recall that in this example, $\Fam_i=\{\F'_{i,1},\F'_{i,2}\}$ for $i \in \{1,2\}$.
For each $i,j \in \{1,2\}$, we consider all pairs of modified strings in $\F'_{1,i} \times \F'_{2,j}$ and if a pair originates from a pair in $\G$, it is added to $\G'_{i,j}$. We do not create the set $\G'_{2,2}$ as it would be empty.

Let us check on two examples that $\Gam$ satisfies the conditions of \cref{def:bicomplete}.

For $(A_2=\mathtt{abbac},B_2=\mathtt{abcba}),\,(A_3=\mathtt{acbab},B_3=\mathtt{aacba}) \in \G$:
\begin{itemize}
    \item $(A_2^{\{(2,\mathtt{a})\}}=\mathtt{a\textcolor{red}{a}bac},A_3^{\{(2,\mathtt{a})\}}=\mathtt{a\textcolor{red}{a}bab})$ is a $(A_2,A_3)_1$-maxpair with $\LCP(A_2^{\{(2,\mathtt{a})\}},A_3^{\{(2,\mathtt{a})\}})=4$,
    \item $(B_2=\mathtt{abcba},B_3^{\{(2,\mathtt{b})\}}=\mathtt{a\textcolor{red}{b}cba})$ is a $(B_2,B_3)_1$-maxpair with $\LCP(B_2,B_3^{\{(2,\mathtt{b})\}})=5$,
    \item $(A_2^{\{(2,\mathtt{a})\}},B_2),(A_3^{\{(2,\mathtt{a})\}},B_3^{\{(2,\mathtt{b})\}}) \in \G'_{1,2}$.
\end{itemize}
Let us note that $A_2^{\{(2,\mathtt{a})\}},A_3^{\{(2,\mathtt{a})\}}$ contain modifications on only one position and the same applies to $B_2,B_3^{\{(2,\mathtt{b})\}}$. Let $\U=\{(A_1,B_1),(A_2,B_2),(A_5,B_5)\}$ and $\V=\{(A_3,B_3),(A_4,B_4)\}$. One can check that:
\[\maxPairLCP_{1,1}(\U,\V)=\LCP(A_2^{\{(2,\mathtt{a})\}},A_3^{\{(2,\mathtt{a})\}})+\LCP(B_2,B_3^{\{(2,\mathtt{b})\}})=9.\]
Let us note that even though the set $\G'_{1,1}$ contains pairs of modified strings with sources $A_2$ and $B_2$ and pairs with sources $A_3$ and $B_3$, no such pair of pairs yields a $(B_2,B_3)_1$-maxpair.

For $(A_3=\mathtt{acbab},B_3=\mathtt{aacba}),\,(A_5=\mathtt{bcbac},B_5=\mathtt{bbcba}) \in \G$:
\begin{itemize}
    \item $(A_3=\mathtt{acbab},A_5^{\{(1,\mathtt{a})\}}=\mathtt{\textcolor{red}{a}cbac})$ is a $(A_3,A_5)_1$-maxpair with $\LCP(A_3,A_5^{\{(1,\mathtt{a})\}})=4$,
    \item $(B_3^{\{(1,\mathtt{b})\}}=\mathtt{\textcolor{red}{b}acba},B_5=\mathtt{bbcba})$ is a $(B_3,B_5)_1$-maxpair with $\LCP(B_3^{\{(1,\mathtt{b})\}},B_5)=1$,
    \item $(A_3,B_3^{\{(1,\mathtt{b})\}}),(A_5^{\{(1,\mathtt{a})\}},B_5) \in \G'_{1,1}$.
\end{itemize}
Let us note that $(A_5^{\{(1,\mathtt{a})\}},B_5^{\{(2,\mathtt{a})\}}=\mathtt{b\textcolor{red}{a}cba})$ also belongs to $\G'_{1,1}$ and $\LCP(B_3^{\{(1,\mathtt{b})\}},B_5^{\{(2,\mathtt{a})\}})=5$. However, the two modified strings contain modifications at \emph{two different positions}, so they do not form a $(B_3,B_5)_1$-maxpair. In particular, they cannot be used for computing $\maxPairLCP_{1,1}(\U,\V)$ according to the formula from \cref{ex:mpLCP}.
\end{example}

The construction of a bicomplete family from \cref{lem:bicomplete} requires processing the complete families in batches to ensure that only linear space is used.

\begin{lemma}\label{lem:bicomplete2}
  Let $\G$ be an $(\ell,\ell)$-family and $k_1,k_2\in \mathbb{Z}_{\ge 0}$ with $k_1,k_2=\Oh(1)$.
  If $\G$ consists of pairs of substrings of a given length-$n$ text over the alphabet $[0\dd n)$,
  then a $(k_1,k_2)$-bicomplete family~$\Gam$ for $\G$ that satisfies the conditions of \cref{lem:bicomplete} can be constructed in $\Oh(n+|\G|)$ space and $\Oh(n + |\G|\min(\ell,\log |\G|)^{k_1+k_2})$ time  with sets $\G'\in \Gam$ generated one by one and each set $\G'$ sorted in two ways:
  according to the lexicographic order of the modified strings on the first and the second coordinate, respectively.
\end{lemma}
\begin{proof}
  We will show how to construct the bicomplete family from the proof of \cref{lem:bicomplete}.
  We first generate the family $\Fam_1$ and batch it into subfamilies $\Fam_1'\sub \Fam_1$ using the following procedure applied on top of the algorithm of~\cref{prp:complete} after it outputs a set $\F'_1\in \Fam_1$.
  For each modified string $F_1^{\Delta_1}\in \F'_1$, we augment each pair of the form $(F_1,F_2)\in \G$ with a triple consisting of $F_1^{\Delta_1}$, the lexicographic rank of $F_1^{\Delta_1}$ within $\F'_1$, and the index of $\F'_1$ within $\Fam'_1$.
  After processing a set $\F'_1\in \Fam_1$ in this manner, we resume the algorithm of~\cref{prp:complete}
  provided that the total number of triples stored is smaller than $n+|\G|$.
  Otherwise (and when the algorithm of~\cref{prp:complete} terminates), 
  we declare the construction of the current batch $\Fam_1'\sub \Fam_1$ complete.

  For each batch $\Fam_1'\sub \Fam_1$, we generate the family $\Fam_2$ and batch it into subfamilies $\Fam_2'\sub \Fam_2$ using the following procedure applied on top of the algorithm of~\cref{prp:complete} after it outputs a set $\F'_2\in \Fam_2$.
  For each modified string $F_2^{\Delta_2}\in \F'_2$, we retrieve all pairs of the form $(F_1,F_2)\in \G$
  and iterate over the triples stored at $(F_1,F_2)$.
  For each such triple, consisting of a modified string~$F_1^{\Delta_1}$, the lexicographic rank of $F_1^{\Delta_1}$ within $\F'_1\in \Fam_1'$, and the index of $\F'_1$ within $\Fam'_1$, 
  we generate the corresponding triple for $F_2^{\Delta_2}$ and combine the two triples into a 6-tuple.
  After processing a set $\F'_2\in \Fam_2$ in this manner, we resume the algorithm of~\cref{prp:complete}
  provided that the total number of 6-tuples stored is smaller than $n+|\G|$.
  Otherwise (and when the algorithm of~\cref{prp:complete} terminates), 
  we declare the construction of the current batch $\Fam_2'\sub \Fam_2$ complete.

  For each pair of batches $\Fam'_1,\Fam'_2$, we group the 6-tuples by the index of $\F'_1$ within $\Fam'_1$ and the index of $\F'_2$ within $\Fam'_2$, and we sort the 6-tuples in each group in two ways: by the lexicographic rank of $F_1^{\Delta_1}$ within $\F'_1$ and by the lexicographic rank of $F_2^{\Delta_2}$ within $\F'_2$.
  The keys used for sorting and grouping are integers bounded by $n^{\Oh(1)}$, so we implement this step using radix sort. 
  Finally, for each group, we create a set $\G'$ by preserving only the modified strings $(F_1^{\Delta_1},F_2^{\Delta_2})$ out of each 6-tuple. We output $\G'$ along with the two linear orders: according to $F_1^{\Delta_1}$ and $F_2^{\Delta_2}$.
  It is easy to see that this yields $\G' = \{(F^{\Delta_1}_1,F^{\Delta_2}_2)\in \F'_1\times \F'_2 : (F_1,F_2)\in \G\}$. Consequently, for the two batches $\Fam'_1,\Fam'_2$, the algorithm produces
  \[\Gam' = \{\{(F^{\Delta_1}_1,F^{\Delta_2}_2)\in \F'_1\times \F'_2 : (F_1,F_2)\in \G\} : (\F'_1,\F'_2)\in \Fam'_1\times \Fam'_2\}\sm \{\emptyset\}.\]
  Since, for $i\in \{1,2\}$, each set $\F'_i\in \Fam_i$ belongs to exactly one batch $\Fam'_i$,
  across all pairs of batches we obtain the family $\Gam$ defined above.

  We conclude with the complexity analysis. The generators of $\Fam_1$ and $\Fam_2$ use $\Oh(n+|\F_1|)=\Oh(n+|\G|)$ and $\Oh(n+|\F_2|)=\Oh(n+|\G|)$ space, respectively.
  Each set $\F'_1\in \Fam_1$ contains $\Oh(1)$ modified strings with the same source, so the number of triples generated for $\F_1'$ is $\Oh(|\G|)$, so the number of triples generated for each batch $\Fam_1'$ is $\Oh(n+|\G|)$.
  Similarly, each set $\F'_2\in \Fam_2$ contains $\Oh(1)$ modified strings with the same source, so the number of 6-tuples generated for each batch pair $\Fam'_1,\Fam'_2$ is $\Oh(n+|\G|)$.
  The triples are removed after processing each batch $\Fam'_1$ and the 6-tuples are removed after processing each pair of batches $\Fam'_1,\Fam'_2$, so the space complexity of the entire algorithm is $\Oh(n+|\G|)$.

  As for the running time, note that the generator of $\Fam_1$ takes $\Oh(n +  |\F_1|\min(\ell,\log|\F_1|)^{k_1})=
  \Oh(n +  |\G|\min(\ell,\log|\G|)^{k_1})$ time. In the post-processing of the sets $\F'_1\in \Fam_1$, by \cref{prp:complete}, for each $(F_1,F_2) \in \G$ we generate $\Oh(\min(\ell,\log|\G|)^{k_1})$ triples. This gives $\Oh(|\G|\min(\ell,\log|\F_1|)^{k_1})=\Oh(|\G|\min(\ell,\log|\G|)^{k_1})$ triples in total in $\Oh(1)$ time per triple. 
  The number of batches $\Fam'_1$ is therefore $\Oh(1 + \frac{|\G|}{n+|\G|}\min(\ell,\log|\G|)^{k_1})$.
  For each such batch, we run the generator of $\Fam_2$, which takes $\Oh(n +  |\F_2|\min(\ell,\log|\F_2|)^{k_2})=
  \Oh(n +  |\G|\min(\ell,\log|\G|)^{k_2})$ time. 
  Across all batches $\Fam'_1$, this sums up to \[\Oh((1 + \tfrac{|\G|}{n+|\G|}\min(\ell,\log|\G|)^{k_1}\!)(n+|\G|\min(\ell,\log|\G|)^{k_2}\!))\!=\Oh(n+|\G|\min(\ell,\log|\G|)^{k_1+k_2}).\]
  In the post-processing of the sets $\F'_2\in \Fam_2$, we generate $\sum_{\G'\in \Gam}|\G'|=\Oh(|\G|\min(\ell,\log|\G|)^{k_1+k_2})$ 6-tuples, in $\Oh(1)$ time per tuple. 
  The number of batch pairs $\Fam'_1,\Fam'_2$
  is therefore
  \[\Oh(1 + \tfrac{|\G|}{n+|\G|}\min(\ell,\log|\G|)^{k_1} + \tfrac{|\G|}{n+|\G|}\min(\ell,\log|\G|)^{k_1+k_2})=\Oh(1+\tfrac{|\G|}{n+|\G|}\min(\ell,\log|\G|)^{k_1+k_2}).\]
  Each batch pair is processed in $\Oh(n+|\G|)$ time, so this yields $\Oh(n+|\G|\min(\ell,\log|\G|)^{k_1+k_2})$ time in total.
\end{proof}

Next we will use bicomplete families to reduce the computation of $\maxPairLCP_{k_1,k_2}(\U,\V)$ for $(\ell,\ell)$-families of substrings of a given text to a number of computations of $\maxPairLCP(\U',\V')$ for $(\ell,\ell)$-families of modified substrings. Similarly to~\cite{DBLP:journals/jcb/ThankachanAA16}, we need to group the elements of a bicomplete family for $\U \cup \V$ by subsets of modifications so that, if two modified strings have a common modification, it is counted as one mismatch between them. Intuitively, we primarily want to avoid checking the conditions $|\Delta_i \cup \nabla_i| \le k_i$ for $i \in \{1,2\}$ in the formula in \cref{ex:mpLCP}.

\begin{proposition}\label{prp:maxpairlcpk}
  Let $\U,\V$ be $(\ell,\ell)$-families and $k_1,k_2\in \mathbb{Z}_{\ge 0}$ with $k_1,k_2=\Oh(1)$.
  There exists a family $\Pam$ such that:
  \begin{enumerate}
    \item\label{it1} each element of $\Pam$ is a pair of sets $(\U',\V')$, where the elements of $\U'$ are of the form $(U_1^{\Delta_1},U_2^{\Delta_2})$ for $(U_1,U_2)\in \U$ with $|\Delta_1|\le k_1$ and $|\Delta_2|\le k_2$, whereas the elements of $\V'$ are of the form $(V_1^{\nabla_1},V_2^{\nabla_2})$ for $(V_1,V_2)\in \V$ with $|\nabla_1|\le k_1$ and $|\nabla_2|\le k_2$;
    \item\label{it2} $\max_{(\U',\V')\in \Pam} (|\U'|+|\V'|) = \Oh(|\U|+|\V|)$;
    \item\label{it3} $\sum_{(\U',\V')\in \Pam} (|\U'|+|\V'|) = \Oh((|\U|+|\V|)\min(\ell,\log(|\U|+|\V|))^{k_1+k_2})$;
    \item\label{it4} $\max_{(\U',\V')\in \Pam} \maxPairLCP(\U',\V') = \maxPairLCP_{k_1,k_2}(\U,\V)$.
  \end{enumerate}
\end{proposition}
\begin{proof}
  Let $\G = \U\cup \V$ and let $\Gam$ be the $(k_1,k_2)$-bicomplete family for $\G$ constructed as in \cref{lem:bicomplete}.
  Given $\G'\in \Gam$ and $\delta_1,\delta_2\sub \mathbb{Z}_{\ge 0} \times \Sigma$, we define a subset of $\G'$:
  \[\G'_{\delta_1,\delta_2}  = \{(F_1^{\Delta_1},F_2^{\Delta_2})\in \G' : \delta_1 \sub \Delta_1, \delta_2\sub \Delta_2\}.\]
  For all non-empty sets $\G'_{\delta_1,\delta_2}$ with $\G'\in \Gam$, and all integers $d_1\in [0\dd  k_1]$ and $d_2\in [0\dd  k_2]$, we insert to $\Pam$ a pair $(\U',\V')$, where
  \begin{align*}
    \U' &= \{(U_1^{\Delta_1},U_2^{\Delta_2})\in \G'_{\delta_1,\delta_2}: (U_1,U_2)\in \U,\, |\Delta_1|\le d_1,\, |\Delta_2|\le d_2\},\\
    \V' &= \{(V_1^{\nabla_1},V_2^{\nabla_2})\in \G'_{\delta_1,\delta_2}: (V_1,V_2)\in \V,\, |\nabla_1|\le k_1+|\delta_1|-d_1,\, |\nabla_2|\le k_2 + |\delta_2|-d_2\}.
  \end{align*}
  Clearly, the elements of $\U'$ and $\V'$ satisfy requirement \eqref{it1}.
  Moreover, $|\U'|+|\V'|\le 2|\G'| = \Oh(|\G|) =  \Oh(|\U|+|\V|)$, so requirement \eqref{it2} is fulfilled.
  Furthermore, each pair $(F_1^{\Delta_1},F_2^{\Delta_2})\in \G'$ belongs to 
  $2^{|\Delta_1|+|\Delta_2|}\le 2^{k_1+k_2}=\Oh(1)$ sets $\G'_{\delta_1,\delta_2}$
  and thus to $\Oh(k_1k_2)=\Oh(1)$ sets $\U'$ or $\V'$ created for $\G'$.
  Consequently, due to $\sum_{\G'\in \Gam} |\G'| = \Oh(|\G|\min(\ell,\log|\G|)^{k_1+k_2})$,
  we have $\sum_{(\U',\V')\in \Pam}(|\U'|+|\V'|) = \Oh(|\G|\min(\ell,\log|\G|)^{k_1+k_2})=\Oh((|\U|+|\V|)\min(\ell,\log(|\U|+|\V|))^{k_1+k_2})$, which yields requirement \eqref{it3}.

  Our next goal is to prove $\maxPairLCP_{k_1,k_2}(\U,\V) = \max_{(\U',\V')\in \Pam} \maxPairLCP(\U',\V')$ (requirement \eqref{it4}).

\begin{example}
As in~\cref{ex:maxpairs}, consider $U=\mathtt{ababbabb}$, $V=\mathtt{aacbaaab}$, $k=3$, $\Delta=\{(2,\mathtt{a}),(3,\mathtt{b})\}$,
and $\nabla=\{(3,\mathtt{b}),(5,\mathtt{b})\}$.
Recall that $(U^\Delta, V^\nabla)$ form a $(U,V)_3$-maxpair.
We have $\delta=\Delta \cap \nabla=\{(3,\mathtt{b})\}$.
$U^\Delta$ has $d=2$ modifications, and $V^\nabla$ has $k+|\delta|-d=2$ modifications.
Note that, if we instead had $\nabla=\{(3,\mathtt{b}),(5,\mathtt{b}),(7,\mathtt{b})\}$, then $\LCP(U^\Delta,V^\nabla)$ would be greater than $\LCP_3(U,V)$, and we would thus not want to consider the pair $(U^\Delta, V^\nabla)$.
\end{example}
  
  Let us fix $(\U',\V')\in \Pam$ generated for $\G'$, $\delta_1$, $\delta_2$, $d_1$, and $d_2$.
  For every $(U_1^{\Delta_1},U_2^{\Delta_2})\in \U'$ and $(V_1^{\nabla_1},V_2^{\nabla_2})\in \V'$,
  we have $|\Delta_i\cup \nabla_i| = |\Delta_i|+|\nabla_i|-|\Delta_i\cap \nabla_i|\le d_i + {(k_i + |\delta_i| - d_i)} - |\delta_i| \le k_i$ for $i\in \{1,2\}$ and thus, by~\cref{fct:pair}\eqref{it:pair:b}, $\LCP(U_i^{\Delta_i},V_i^{\nabla_i})\le \LCP_{k_i}(U_i,V_i)$.
  Consequently, we have $\LCP(U_1^{\Delta_1},V_1^{\nabla_1})+\LCP(U_2^{\Delta_2},V_2^{\nabla_2}) \le \LCP_{k_1}(U_1,V_1) + \LCP_{k_2}(U_2,V_2)$, and thus $\maxPairLCP(\U',\V')\le \maxPairLCP_{k_1,k_2}(\U,\V)$.

  As for the converse inequality, suppose that $(U_1,U_2)\in \U$ and $(V_1,V_2)\in \V$ satisfy
  $$\maxPairLCP_{k_1,k_2}(\U,\V) = \LCP_{k_1}(U_1,V_1) + \LCP_{k_2}(U_2,V_2).$$
  By the definition of a bicomplete family (\cref{def:bicomplete}), there exist $\G'\in \G$ and pairs of modified strings $(U_1^{\Delta_1},U_2^{\Delta_2}),(V_1^{\nabla_1},V_2^{\nabla_2})\in \G'$ such that $(U_i^{\Delta_i},V_i^{\nabla_i})$ is a $(U_i,V_i)_{k_i}$-maxpair for $i\in \{1,2\}$.
  By~\cref{fct:pair}, this yields $\LCP_{k_i}(U_i,V_i) = \LCP(U_i^{\Delta_i},V_i^{\nabla_i})$ and $|\Delta_i\cup \nabla_i|\le k_i$.
  We define $\delta_i = \Delta_i\cap \nabla_i$ and $d_i = |\Delta_i|$ for $i\in \{1,2\}$,
  and consider the pair $(\U',\V')$ constructed for $\G'$, $\delta_1$, $\delta_2$, $d_1$, and~$d_2$.
  Note that $|\Delta_i|\le d_i$ and $\delta_i \sub \Delta_i$, so $(U_1^{\Delta_1},U_2^{\Delta_2})\in \U'$.
  Moreover, $|\nabla_i| = {|\Delta_i\cup \nabla_i| + |\delta_i| - |\Delta_i|} \le {k_i + |\delta_i|-d_i}$
  and $\delta_i\sub \nabla_i$, so $(V_1^{\nabla_1},V_2^{\nabla_2})\in \V'$.
  Consequently, $\maxPairLCP_{k_1,k_2}(\U,\V)\le \maxPairLCP(\U',\V')$.
\end{proof}

\begin{example}\label{ex:P}
Let us consider sets $\U=\{(A_1,B_1),(A_2,B_2),(A_5,B_5)\}$ and $\V=\{(A_3,B_3),(A_4,B_4)\}$ and the $(1,1)$-bicomplete family $\Gam$ of $\U \cup \V$ from \cref{ex:G}. Let $\G'$ be the following set $\G'_{1,2}$ from that example:
\begin{align*}
\G'=\,&\{A_2=\mathtt{abbac},A_2^{\{(2,\mathtt{a})\}}=\mathtt{a\textcolor{red}{a}bac}\} \times \{B_2=\mathtt{abcba}\}\ \cup\\
&\{A_3=\mathtt{acbab},A_3^{\{(2,\mathtt{a})\}}=\mathtt{a\textcolor{red}{a}bab}\} \times \{B_3=\mathtt{aacba},B_3^{\{(2,\mathtt{b})\}}=\mathtt{a\textcolor{red}{b}cba}\}.
\end{align*}
In the construction of \cref{prp:maxpairlcpk}, in particular, the following set is constructed:
\begin{align*}
\G'_{\{(2,\mathtt{a})\},\emptyset}=\,&\{A_2^{\{(2,\mathtt{a})\}}=\mathtt{a\textcolor{red}{a}bac}\} \times \{B_2=\mathtt{abcba}\}\ \cup\\
&\{A_3^{\{(2,\mathtt{a})\}}=\mathtt{a\textcolor{red}{a}bab}\} \times \{B_3=\mathtt{aacba},B_3^{\{(2,\mathtt{b})\}}=\mathtt{a\textcolor{red}{b}cba}\}.
\end{align*}
For $\G'_{\{(2,\mathtt{a})\},\emptyset}$, $d_1=1$ and $d_2=0$ the following pair of sets $(\U',\V')$ is inserted to $\Pam$:
\begin{align*}
\U'&=\{A_2^{\{(2,\mathtt{a})\}}=\mathtt{a\textcolor{red}{a}bac}\} \times \{B_2=\mathtt{abcba}\},\\
\V'&=\{A_3^{\{(2,\mathtt{a})\}}=\mathtt{a\textcolor{red}{a}bab}\} \times \{B_3=\mathtt{aacba},B_3^{\{(2,\mathtt{b})\}}=\mathtt{a\textcolor{red}{b}cba}\}.
\end{align*}
We have $\maxPairLCP_{1,1}(\U,\V) = \maxPairLCP(\U',\V')$, as discussed in \cref{ex:G}.

As another example, $\G'_{\emptyset,\emptyset}=\G'$. For this set, no $d_1,d_2 \in [0\dd 1]$ produce a pair of sets $(\U',\V')$ such that $(A_2^{\{(2,\mathtt{a})\}},B_2)\in \U'$, $(A_3^{\{(2,\mathtt{a})\}},B_3^{\{(2,\mathtt{b})\}}) \in \V'$. Intuitively, this is because in $\G'_{\emptyset,\emptyset}$, the modifications $\{(2,\mathtt{a})\}$ in the first modified strings in the pairs are counted twice.
\end{example}

  \begin{lemma}\label{lem:maxpairlcpk2}
  Let $\U,\V$ be $(\ell,\ell)$-families and $k_1,k_2\in \mathbb{Z}_{\ge 0}$ with $k_1,k_2=\Oh(1)$.
  Moreover, if $\U$ and~$\V$ consist of pairs of substrings of a given length-$n$ text over the alphabet $[0\dd n)$,
  then the family~$\Pam$ from \cref{prp:maxpairlcpk} can be constructed in $\Oh(n+(|\U|+|\V|)\min(\ell,\log(|\U|+|\V|))^{k_1+k_2})$
  time and $\Oh(n+|\U|+|\V|)$ space, with pairs $(\U',\V')$ generated one by one 
  and each set $\U'\cup \V'$ sorted in two ways:   
  according to the lexicographic order of the modified strings on the first and the second coordinate, respectively.
\end{lemma}
\begin{proof}
  We use the construction from \cref{prp:maxpairlcpk}.
  We apply the construction of a bicomplete family (\cref{lem:bicomplete}) to generate the family $\Gam$ in batches $\Gam'$
  consisting of sets of total size $\Theta(n+|\G|)$ (the last batch might be smaller). In the algorithm $\Gam$ is processed in batches of size $\Omega(n)$ (and not, say, set by set) so that the time required to bucket sort integers of magnitude $n^{\Oh(1)}$ in the algorithm does not make the overall time complexity worse.
  For each batch $\Gam'\sub \Gam$, we first construct all the non-empty sets $\G'_{\delta_1,\delta_2}$ for $\G'\in \Gam'$. 
  For this, we iterate over sets $\G'\in \Gam'$, pairs $(F_1^{\Delta_1},F_2^{\Delta_2})\in \G'$,
  subsets $\delta_1\sub \Delta_1$, and subsets $\delta_2\sub \Delta_2$,
  creating 5-tuples consisting of the index of $\G'$ in $\Gam'$, as well as $\delta_1$, $\delta_2$, $F_1^{\Delta_1}$, and $F_2^{\Delta_2}$.
  We group these 5-tuples according to the first three coordinates;
  the key consists of $\Oh(1+k_1+k_2)=\Oh(1)$ integers bounded by $n^{\Oh(1)}$, so we use radix sort.
  Moreover, since radix sort is stable, the sets $\G'_{\delta_1,\delta_2}$ can be constructed along with both linear orders derived from $\G'$.
  Finally, for each non-empty set $\G'_{\delta_1,\delta_2}$ with $\G'\in \Gam'$,
  we iterate over $d_1\in [0\dd  k_1]$ and $d_2\in [0\dd  k_2]$, generating the subsets $\U'$ and $\V'$ of $\G'_{\delta_1,\delta_2}$ according to the formulae above.
  The two linear orders of $\U'\cup \V'$ are derived from $\G'_{\delta_1,\delta_2}$.

  We conclude with the complexity analysis.
  The construction of the $(k_1,k_2)$-bicomplete family of~\cref{lem:bicomplete} requires $\Oh(n+|\G|\min(\ell,\log|\G|)^{k_1+k_2})$ time and $\Oh(n+|\G|)$ space.
  The total size of the sets in each batch $\Gam'$ is $\Oh(n+|\G|)$ and the number of batches is $\Oh(1+\frac{|\G|}{n+|\G|}\min(\ell,\log|\G|)^{k_1+k_2})$.
  For each modified string $(F_1^{\Delta_1},F_2^{\Delta_2})\in \G'$,
  we construct $2^{|\Delta_1|+|\Delta_2|}\le 2^{k_1+k_2}=\Oh(1)$ tuples corresponding to elements of the sets $\G'_{\delta_1,\delta_2}$, in $\Oh(1)$ time per tuple. Consequently, this phase uses $\Oh(n+|\G|)$ space and $\Oh(|\G|\min(\ell,\log|\G|)^{k_1+k_2})$ time.
  For each batch, we use $\Oh(n+|\G|)$ time and space for radix sort, yielding a total of $\Oh(n+|\G|)$ space and $\Oh(n+|\G|\min(\ell,\log|\G|)^{k_1+k_2})$ time.
  Finally, for each non-empty set $\G'_{\delta_1,\delta_2}$ with $\G'\in \Gam$,
  each $d_1\in [0\dd  k_1]$, and each $d_2\in [0\dd  k_2]$, we spend $\Oh(|\G'_{\delta_1,\delta_2}|)$ time and space to generate $\U'$ and $\V'$. 
  Since $k_1,k_2=\Oh(1)$, the overall time of this final phase is proportional to the total size of the sets $\G'_{\delta_1,\delta_2}$, which is $\Oh(|\G|\min(\ell,\log|\G|)^{k_1+k_2})$ as argued in the proof of \cref{prp:maxpairlcpk}.
\end{proof}

We are ready to provide a proof of \cref{thm:maxpairlcpk}. With condition \eqref{it4} of \cref{prp:maxpairlcpk}, instead of computing $\maxPairLCP_{k_1,k_2}(\U,\V)$, it suffices to compute $\max_{(\U',\V')\in \Pam} \maxPairLCP(\U',\V')$. We process the family $\Pam$ in batches of size $\Omega(N)$ to use the fact that the complexity in \cref{lem:main} has a denominator $\log N$ (instead of, say, a logarithm of the size of a set in $\Pam$). Still, the batch size is bounded by $\Oh(N)$ to ensure small space complexity. By \cref{lem:maxpairlcpk2}, there are $\Oh(\min(\ell,\log N)^k)$ batches. For each of them we create an instance of \problemtwo, which we solve using~\cref{lem:problem} if $\ell > \log^{3/2}N$
or~\cref{lem:main} otherwise to obtain the desired complexity.

Let us restate the theorem for convenience.

\maxpairlcpk*
\begin{proof}
First, we augment the given text in $\Oh(n)$ time with a data structure for $\Oh(1)$-time $\LCE$ queries (see~\cref{thm:lce}).
Next, we apply~\cref{prp:maxpairlcpk,lem:maxpairlcpk2} to generate a family $\Pam$
such that $\maxPairLCP_{k_1,k_2}(\U,\V)=\max_{(\U',\V')\in \Pam}\maxPairLCP(\U',\V')$.
We process pairs ${(\U',\V')\in \Pam}$ in batches $\Pam'\sub \Pam$ consisting of pairs of sets of total size $\Theta(N)$ (in the last batch, the total size might be smaller).
Each batch $\Pam'=\{(\U'_j,\V'_j): j\in [1\dd p]\}$ is processed as follows.

First, for each $j\in [1\dd p]$, we construct the compacted trie $\T(\F_{j})$ of the set $\F_j = \F_{1,j}\cup \F_{2,j}$, where 
\[\F_{i,j} = \{F_i^{\Delta_i} : (F_1^{\Delta_1},F_2^{\Delta_2})\in \U'_j\cup \V'_j\}.\]
The two linear orders associated with the set $\U'_j\cup \V'_j$ yield the lexicographic order of the sets $\F_{1,j}$ and $\F_{2,j}$. These orders can be merged to a linear order of $\F_j$ using a classic linear-time algorithm with the aid of $\LCE$ queries used to compare any two modified substrings. In order to build the compacted tries $\T(\F_{j})$ using~\cref{lem:ctrie},
it suffices to determine the longest common prefixes between any two consecutive modified strings in $\F_{j}$,
which reduces to $\LCE$ queries on the text. 

Then we construct the compacted trie $\T(\F)$ of the set 
\[\F = \bigcup_{j=1}^{p} \{j\cdot F^{\Delta} : F^{\Delta}\in \F_j\}\sub [1\dd p]\cdot \Sigma^{\le \ell}.\]
For this, we create a new root node and, for $j\in [1\dd m]$, attach the compacted trie $\T(\F_{j})$ with a length-$1$ edge; if the root of $\T(\F_{j})$ has degree $1$ and $\varepsilon \notin \F_{j}$, we also dissolve the root. 
Finally, we construct two sets $\P,\Q\sub \F^2$:
\vspace{-1.5ex}
\begin{align*}
  \P &= \bigcup_{j=1}^{p} \{(j\cdot U_1^{\Delta_1},j\cdot U_2^{\Delta_2}) : (U_1^{\Delta_1},U_2^{\Delta_2})\in \U'_j\},\\
  \Q &=  \bigcup_{j=1}^{p} \{(j\cdot V_1^{\nabla_1},j\cdot V_2^{\nabla_2}) : (V_1^{\nabla_1},V_2^{\nabla_2})\in \V'_j\}.
\end{align*}%
This yields an instance of the \problemtwo. We solve this instance using~\cref{lem:problem} if $\ell > \log^{3/2}N$
or~\cref{lem:main} otherwise.
The obtained value satisfies \[\maxPairLCP(\P,\Q)=2+\max_{j=1}^{m} \maxPairLCP(\U'_j,\V'_j)\]
and, taking the maximum across all batches $\Pam'\sub \Pam$,
we retrieve $\maxPairLCP_{k_1,k_2}(\U,\V) = \max_{(\U',\V')\in \Pam} \maxPairLCP(\U',\V')$ (cf.\ condition \eqref{it4} of \cref{prp:maxpairlcpk}).

We conclude with the complexity analysis.
Applying the algorithm of~\cref{lem:maxpairlcpk2} costs $\Oh(n+N\min(\ell,\log N)^{k_1+k_2})$ time
and $\Oh(n+N)$ space.
According to \cref{prp:maxpairlcpk}, the resulting family~$\Pam$ satisfies $\max_{(\U',\V')\in \Pam} (|\U'|+|\V'|)=\Oh(N)$
and $\sum_{(\U',\V')\in \Pam} = \Oh(N \min(\ell,\log N)^{k_1+k_2})$.
Hence, the number of batches $\Pam'$ is $\Oh(\min(\ell,\log N)^{k_1+k_2})$.
Let us now focus on a single batch $\Pam'=\{(\U'_j,\V'_j): 1\le j \le p\}$; recall that it satisfies
$\sum_{j=1}^p (|\U'_j|+|\V'_j|)=\Oh(N)$. 
Each set $\F_{j}$ is obtained from $\U'_j\cup \V'_j$ in $\Oh(|\U'_j|+|\V'_j|)$ time.
Each modified string in $\F_{j}$ contains up to $k_1+k_2$ modifications,
so the $\LCP$ computation for two such modified strings requires up to $2(k_1+k_2)+1=\Oh(1)$ $\LCE$ queries on $T$.
Consequently, each trie $\T(\F_{j})$ is constructed in $\Oh(|\U'_j|+|\V'_j|)$ time.
Merging these tries into $\T(\F)$ requires $\Oh(p)$ additional time, i.e., 
the construction of $\T(\F)$ takes $\Oh(N)$ time in total.
The sets $\P$ and $\Q$ are of size $\Oh(N)$ and also take $\Oh(N)$ time to construct.
Overall, preparing the instance $(\T(\F),\P,\Q)$ of \problemtwo
takes $\Oh(N)$ time and space.
Solving this instance takes $\Oh(N\log N)$ time and $\Oh(N)$ space if we use~\cref{lem:problem}
or $\Oh(N(\ell+\log N)(\log \ell + \sqrt{\log N})/\log N)$ time and $\Oh(N+N\ell/\log N)$ space if we use~\cref{lem:main}.
In either case, this final step dominates both the time and the space complexity of processing a single batch $\Pam'$.
Across all $\Oh(\min(\ell,\log N)^{k_1+k_2})$ batches, we obtain the following trade-offs:
\begin{itemize}
  \item $\Oh(N \log N \cdot \log^{k_1+k_2} N)$ time and $\Oh(N)$ space if $\ell > \log^{3/2} N$;
  \item $\Oh((N\ell/\sqrt{\log N})\cdot \log^{k_1+k_2} N)$ time and $\Oh(N\ell/\log N)$ space if $\log N < \ell \le \log^{3/2} N$;
  \item $\Oh(N\sqrt{\log N}\cdot \ell^{k_1+k_2})$ time and $\Oh(N)$ space if $\ell\le \log N$.
\end{itemize}
Accounting for $\Oh(n)$ space and construction time of the data structure for $\LCE$ queries on the text, we retrieve the claimed trade-offs.
\end{proof}

\section*{Acknowledgments}
Panagiotis Charalampopoulos was partly supported by the Israel Science Foundation grant 592/17.

Tomasz Kociumaka was partly supported by NSF 1652303, 1909046, and HDR TRIPODS 1934846 grants, and an Alfred P. Sloan Fellowship.

Jakub Radoszewski was supported by the Polish National Science Center, grants no. 2018/31{\allowbreak}/D/ST6/03991 and 2022/46/E/ST6/00463.

Solon P. Pissis was supported in part by the PANGAIA and ALPACA projects that have received funding from the European Union's Horizon 2020 research and innovation programme under the Marie Skłodowska-Curie grant agreements No 872539 and 956229, respectively.

\bibliographystyle{alphaurl}
\bibliography{references2}

\appendix

\section{Proof of Proposition \ref{prp:streaming}}\label{app:prpstream}
Throughout the appendix, we restate the propositions/lemmas for convenience.
\prpstream*
\begin{proof}
The sequence of computations in the algorithm depends only on the next instruction to be performed, its current memory state, and the data that appear on the input streams. More precisely, in the preprocessing:
\begin{itemize}
\item for every possible instruction in the algorithm's code ($\Oh(1)$ options; we can assume that it is a single-bit read or write instruction or the first instruction of the algorithm), and 
\item for every possible state of memory used in the algorithm ($2^s$ options), and
\item for every possible combination of the next (up to) $\floor{\tau/x}$ bits that are available on each of the input streams, where $x$ is the number of input streams ($\Oh(2^\tau)$ possible combinations),
\end{itemize}
the algorithm is simulated starting from the selected instruction with the given memory state. The simulation lasts until at least one of the following events takes place:
\begin{itemize}
\item the algorithm terminates, or
\item the algorithm requires to read a bit of an input stream that is not given to it (that is, the $(\floor{\tau/x}+1)$th bit of a stream), or
\item the algorithm requires to write the $(\tau+1)$th bit.
\end{itemize}
The computations take $\Oh(\tau t)$ time. When the simulation ends, the following data is memoised:
\begin{itemize}
\item the (read or write) instruction on which the simulation ended, or information that the algorithm terminated ($\Oh(1)$ options),
\item the final state of the algorithm's memory ($s$ bits),
\item the number of bits read from each of the input streams ($\Oh(x\log \tau) = \Oh(\log \tau)$ bits),
\item the number of bits written to each of the output streams and the bits written to all output streams, ordered by the stream and then in the order of writing ($\Oh(y \log \tau + \tau)=\Oh(\tau)$ bits, where $y$ is the number of output streams).
\end{itemize}
In total, in the preprocessing, an array consisting of $\Oh(2^{s+\tau})$ cells is constructed. Each cell is represented using $\Oh(1+(\tau+s)/w)=\Oh(1)$ machine words. The total preprocessing time is $\Oh(2^{\tau+s}(\tau t+s))$, as required.

The constructed array allows us to simulate the algorithm in $\Oh(1+Lx/\tau)=\Oh(1+L/\tau)$ time, where $L$ is the total length (in bits) of all input and output streams, in a straightforward way. In the simulation, we store:
\begin{itemize}
\item the current instruction to be executed in the algorithm, and
\item the current state of the algorithm's memory, and
\item for each input stream, the position of the next bit to be read, and
\item the data on the output streams.
\end{itemize}
In each step of the simulation, $\Oh(\tau)$ bits of the input are processed or $\Oh(\tau)$ bits are output. Each step is performed in $\Oh(1)$ time thanks to the fact that the streams are given in a packed form and $s,\tau \le w$. Hence the total processing complexity.
\end{proof}

\section{Proof of Lemma \ref{lem:tauruns}}\label{app:taurons}
\lemtauruns*
\begin{proof}
The first claim of the lemma follows from the periodicity lemma (\cref{lem:perlemma}). Recall that a $\tau$-run is a run of length at least $3\tau-1$ with period at most $\frac13\tau$.
Suppose, towards a contradiction, that two different $\tau$-runs overlap by more than $\frac23\tau$ positions. By the fact that their periods are at most $\frac13\tau$ and an application of \cref{lem:perlemma}, we obtain that these two $\tau$-runs cannot be distinct as they share the same smallest period; a contradiction. This implies that two distinct $\tau$-runs can overlap by no more than $\frac23\tau$ positions. In turn, this implies that $T$ contains $\Oh(n/\tau)$ runs.

Let us now show how to efficiently compute $\tau$-runs.
We first compute a $\tau$-synchronising set $A$ in $\cO(n/\tau)$ time using~\cref{thm:synch_packed}.
By the definition of such a set, a position $i\in [1\dd  n-3\tau +2]$ in~$T$ is a starting position of a $\tau$-run if and only if $[i\dd i+\tau)\cap A =\emptyset$ and $i-1 \in A$.
The period of this $\tau$-run is equal to $p=\per(T[i \dd i+3\tau -2])$.
Then, by~\cref{lem:fact3.2}, the longest prefix of $T[i \dd n]$ with period $p$ is $T[i \dd \suc_A(i) + 2\tau - 2]$. 
Using these conditions, we can compute all $\tau$-runs in $\cO(n/\tau)$ time in a single scan of $A$. (We assume that $T$ is concatenated with $3\tau$ occurrences of a letter not occurring in $T$ and ignore any $\tau$-run containing such letters.) 

We now show how to group all $\tau$-runs by Lyndon roots.
By the conditions of the statement, there are no more than $\sigma^{3\tau}\leq n^{3/4}$ distinct strings of length $3\tau-1$. We generate all of them, and for each of them that is periodic, we compute its Lyndon root in $\cO(n^{3/4} \log_\sigma n)$ time in total~\cite{DBLP:journals/siamcomp/BannaiIINTT17}.
We store these pairs in a table: the index is the string and the value is the Lyndon root. 
As the Lyndon root of a $\tau$-run $T[i \dd j]$ coincides with the Lyndon root of $T[i \dd i+3\tau-2]$, we can group all $\tau$-runs as desired using this table in $\cO(n/\tau)$ time. This is because $T[i \dd i+3\tau-2]$ can be read in $\Oh(1)$ time in the word RAM model, and so each $\tau$-run is processed in $\Oh(1)$ time.

Similarly, within the same tabulation process,  in $\cO(\tau)$ time, for each distinct string $S$ of length $3\tau-1$ with Lyndon root $R$, we compute and store the position $i_R$ of the leftmost occurrence of $R$ in~$S$~\cite{DBLP:journals/siamcomp/BannaiIINTT17}. Then, for each $\tau$-run $T[a \dd b]$ with Lyndon root $R$ and period $p=|R|$, we can obtain, due to periodicity, in $\Oh(1)$ time the starting positions $a+i_R-1$ and $a+i_R+p-1$ of the two leftmost occurrences of $R$ in $T[a \dd b]$. 
\end{proof}

\section{Proof of Proposition~\ref{prp:complete}}\label{app:complete}

\prpThankachan*

Our proof follows~\cite[Section 4]{DBLP:journals/jcb/ThankachanAA16} along with the efficient implementation from~\cite[Section 5.1]{DBLP:journals/jcb/ThankachanAA16}.
However, we cannot use these results in a black-box manner since we are solving a slightly more general problem:
In~\cite{DBLP:journals/jcb/ThankachanAA16}, the input consists of two strings $X$ and $Y$, and the output family~
$\Fam$ must satisfy the following condition: for every suffix $U$ of $X$ and every suffix $V$ of $Y$,
there exists a set $\F'\in \Fam$ and modified strings $U^{\Delta},V^{\nabla}\in \F'$ forming a $(U,V)_k$-maxpair.
Instead, in \cref{prp:complete}, we have a set $\F$ of (selected) substrings of a single text. The output family $\Fam$ must be a $k$-complete family for $\F$, i.e., satisfy the same condition as before:
for every $U,V\in \F$, there exists a set $\F'\in \Fam$ and modified strings $U^{\Delta},V^{\nabla}\in \F'$ forming a $(U,V)_k$-maxpair.
Additionally, we need stronger bounds on the sizes of sets $\F'\in \Fam$. The original argument~\cite[Lemma 3]{DBLP:journals/jcb/ThankachanAA16} yields $\max_{\F'\in \Fam}|\F'|=\Oh(|\F|)$ and 
$\sum_{\F'\in \Fam}|\F'| = \Oh(|\F|\log^{k}|\F|)$. Instead, we need to prove that each string $F\in \F$
is the source of $\Oh(1)$ modified strings in any single set $\F'\in \Fam$ and that there are $\Oh(\min(\ell,\log|\F|)^k)$ modified strings with source $F$ across all sets $\F'\in \Fam$ provided that $\F\sub \Sigma^{\le \ell}$.

On the high-level, the algorithm behind \cref{prp:complete} constructs $d$-complete families $\Fam_d$ for $\F$ for each $d\in [0\dd k]$. In this setting, the input set $\F$ is interpreted as a $0$-complete family $\Fam_0=\{\F\}$.
For the sake of space efficiency, all families are built in parallel, and the construction algorithm is organised into $2k+1$ \emph{levels}, indexed with $0$ through $2k$.
Each even level $2d$ is given a stream of sets $\F'\in \Fam_d$ and its task is to construct the compacted trie $\T(\F')$ for each given $\F'\in \Fam_d$ (which, in particular, involves sorting the modified strings in $\F'$).
Each odd level $2d-1$ is given a stream of tries $\T(\F')$ for $\F'\in \Fam_{d-1}$, and its task is to output a stream of sets $\F''\in \Fam_d$. We finish at an even level to ensure that modified strings within each set $\F' \in \Fam$ are sorted lexicographically.

The claimed properties of the constructed $k$-complete family $\Fam_k$ naturally generalise to analogous properties of the intermediate $d$-complete families $\Fam_d$. Moreover, efficient implementation of even levels relies on the following additional invariant:
\begin{invariant}\label{inv}
  Every set $\F'\in \Fam_d$ contains a \emph{pivot} $P^\nabla\in \F'$ such that,
  for every  $F^{\Delta}\in \F'$, we have $\Delta \sub [1\dd \LCP(F^\Delta,P^\nabla)]\times \Sigma$.
\end{invariant}
The invariant is trivially true for $d=0$ due to $\Delta=\emptyset$ for every $F^{\Delta}\in \F'$. Each set will store its pivot.

\subsection{Implementation and Analysis of Even Levels}
Recall that the goal of an even level $2d$ is to construct $\T(\F')$ for each $\F'\in \Fam_d$.
For this, we rely on \cref{lem:ctrie}, which requires sorting the modified strings in $F^{\Delta}\in \F'$ and computing the longest common prefixes between pairs of consecutive strings.

For the latter task, in the preprocessing, we augment the input text in $\Oh(n)$ time with a data structure for $\Oh(1)$-time $\LCE$ queries (\cref{thm:lce}). Since each modified string $F^{\Delta}\in \F'$ satisfies $|\Delta|\le d$, 
the longest common prefix of any two modified strings in $\F'$ can be computed in $\Oh(1)$ time using up to $2d+1$ $\LCE$ queries.
In particular, this yields $\Oh(1)$-time lexicographic comparison of any two modified strings in $\F'$,
which means that $\F'$ can be sorted in $\Oh(|\F'|\log |\F'|)$ time.
Nevertheless, this is too much for our purposes. A more efficient procedure requires buffering the sets $\F'\in \Fam_d$ into batches $\Fam'_d\sub \Fam_d$ of total size $\Theta(n+|\F|)$ and using the following result:%
\begin{theorem}[{\cite[Theorem 12]{Sorting}}]\label{thm:sort}
Any $m$ substrings of a given length-$n$ text over the alphabet $[0\dd n)$ can be sorted lexicographically in $\Oh(n+m)$ time.
\end{theorem}
Unfortunately, the possibility of generalising this result to modified substrings remains an open question.
A workaround proposed in~\cite{DBLP:journals/jcb/ThankachanAA16} relies on \cref{inv}. 
In the lexicographic order of~$\F'$, the strings satisfying  $F^{\Delta}\leq P^\nabla$ come before 
those satisfying  $F^{\Delta}>P^\nabla$. Moreover, in the former group, the values $\LCP(F^{\Delta},P^\nabla)$ form a non-decreasing sequence and, in the latter group, these values form a non-increasing sequence.
This way, the task of sorting $\F'$ reduces to partitioning $\F'$ into $\F'_{p}= \{F^{\Delta}\in \F' : \LCP(F^{\Delta},P^\nabla)=p\}$ and sorting each set $\F'_{p}$.
We first compute $\LCP(F^{\Delta},P^\nabla)$ for all $F^{\Delta}\in \F'$ in $\cO(|\F'|)$ time.
Then, in order to construct sets $\F'_{p}$, we use radix sort for the whole batch $\Fam'_d \sub \Fam_d$---the keys ($\LCP(F^{\Delta},P^\nabla)$ for $F^{\Delta}\in \F'$) are integers bounded by~$n$.
In order to sort the elements of each set $\F'_{p}$, we exploit the fact that the lexicographic order of a modified string $F^{\Delta}\in \F'_{p}$
is determined by $F^{\Delta}[p+1\dd |F|]$. Moreover, as $p=\LCP(F^{\Delta},P^\nabla)$ for the pivot $P^\nabla$,
we have $\Delta\sub [1\dd p]\times \Sigma$ by \cref{inv} and thus $F^{\Delta}[p+1\dd |F|]=F[p+1\dd |F|]$.
Thus, $F^{\Delta}[p+1\dd |F|]$ is a substring of the input text and hence we can use \cref{thm:sort} to sort all such substrings arising as we process the batch $\Fam'_d\sub \Fam_d$, and then classify them back into individual subsets $\F'_{p}$ using radix~sort.

The overall time and space complexity for processing a single batch $\Fam'_d\sub \Fam_d$ is, therefore, $\Oh(n+|\F|)$.
Across all batches, the space remains $\Oh(n+|\F|)$ and the running time becomes $\Oh(n+|\F|\min(\ell,\log|\F|)^d)$ due to $\sum_{\F'\in \Fam_d} |\F'| = \Oh(|\F|\min(\ell,\log|\F|)^d)$.

\subsection{Implementation and Analysis of Odd Levels}

Recall that the goal of an odd level $2d-1$ is to transform a stream of tries $\T(\F')$ for $\F'\in \Fam_{d-1}$ into a stream of sets $\F''\in \Fam_d$.
This process is guided by the \emph{heavy-light decomposition} of $\T(\F')$, which classifies the nodes of $\T(\F')$ into \emph{heavy} and \emph{light} so that the root is light and exactly one child of each internal node is heavy: the one whose subtree contains the maximum number of leaves (with ties broken arbitrarily).
Each light node $w$ is therefore associated with the \emph{heavy path} which starts at $w$ and repeatedly proceeds to the unique heavy child until reaching a leaf, which we denote~$h(w)$. The key property of the heavy-light decomposition is that each node has $\Oh(\log |\F'|)$ light ancestors. However, since the height of $\T(\F')$ is $\Oh(\ell)$, we can also bound the number of light ancestors by $\Oh(\min(\ell,\log|\F'|))$.

The algorithm constructs the heavy-light decomposition of $\T(\F')$ and, for each light node $w$,
creates a set $\F'_w\in \Fam_d$ constructed as follows.
We traverse the subtree of $w$ and, for each modified string $F^{\Delta}\in \F'$ with a prefix $\val(w)$,
we compute $p = \LCP(F^{\Delta},\val(h(w)))$, which is the string depth of the lowest common ancestor of $h(w)$ and the locus of $F^{\Delta}$.
If $\Delta \sub [1\dd p]\times \Sigma$, we add $F^{\Delta}$ to $\F'_w$.
If additionally $|F^{\Delta}|>p$, we also add $F^{\Delta\cup\{(p+1,\val(h(w))[p+1])\}}$ to $\F'_w$.
This way, we guarantee that $\val(h(w))\in \F'_w$ forms a pivot of $\F'_w$, as defined in \cref{inv}.
Note that each string $F^{\Delta}\in \F'$ has a prefix $\val(w)$ for $\Oh(\min(\ell,\log |\F'|))$ light nodes $w$.
Hence, each trie $\T(\F')$ for $\F'\in \Fam_{d-1}$ is processed in $\Oh(|\F'|\min(\ell,\log |\F'|))=\Oh(|\F'|\min(\ell,\log |\F|))$ time and $\Oh(|\F'|)=\Oh(|\F|)$ space. 
Across all $\F'\in \Fam_{d-1}$, this yields $\Oh(|\F|\min(\ell,\log |\F|)^d)$ time and $\Oh(|\F|)$ space.

\renewcommand{\tabcolsep}{0pt}
\newcommand{\dy}{0.6}
\begin{figure}[htpb]
\centering
\begin{tikzpicture}
\draw (0,0) -- node[sloped,left,rotate=270] {
\begin{tabular}{c}
$\mathtt{a}$
\end{tabular}
} (-2,-1*\dy);
\draw (-2,-1*\dy) -- node[sloped,left,rotate=270] {
\begin{tabular}{c}
$\mathtt{a}$\\$\mathtt{a}$\\$\mathtt{a}$\\$\mathtt{b}$
\end{tabular}
} (-3.5,-5*\dy);
\draw (-2,-1*\dy) -- node[sloped,right,rotate=90] {
\begin{tabular}{c}
$\mathtt{b}$\\$\mathtt{b}$\\$\mathtt{a}$\\$\mathtt{c}$
\end{tabular}
} (-2,-5*\dy);
\draw (-2,-1*\dy) -- node[sloped,right,rotate=90] {
\begin{tabular}{c}
$\mathtt{c}$\\$\mathtt{b}$\\$\mathtt{a}$\\$\mathtt{b}$
\end{tabular}
} (-0.5,-5*\dy);

\draw (0,0) -- node[sloped,right,rotate=90] {
\begin{tabular}{c}
$\mathtt{b}$\\$\mathtt{c}$
\end{tabular}
} (1,-2*\dy);
\draw (1,-2*\dy) -- node[sloped,left,rotate=270] {
\begin{tabular}{c}
$\mathtt{a}$\\$\mathtt{a}$
\end{tabular}
} (0.3,-4*\dy);
\draw (1,-2*\dy) -- node[sloped,right,rotate=90] {
\begin{tabular}{c}
$\mathtt{b}$\\$\mathtt{a}$\\$\mathtt{c}$
\end{tabular}
} (1.7,-5*\dy);

\foreach \x/\y in {-2/-1,-3.5/-5,0.3/-4}{
\filldraw (\x,\y*\dy) circle (0.08cm);
}
\foreach \x/\y in {0/0,-2/-5,-0.5/-5,1/-2,1.7/-5}{
\filldraw[white] (\x,\y*\dy) circle (0.08cm);
\draw (\x,\y*\dy) circle (0.08cm);
}
\foreach \x/\y/\i in {-3.5/-5/1,-2/-5/2,-0.5/-5/3,0.3/-4/4,1.7/-5/5}{
\draw (\x,\y*\dy) node[below] {$A_{\i}$};
}
\draw (-0.9,-6.5*\dy) node {$\T(\F_1)$};

\begin{scope}[xshift=6cm]
\draw (0,0) -- node[sloped,left,rotate=270] {
\begin{tabular}{c}
$\mathtt{a}$
\end{tabular}
} (-1,-1*\dy);
\draw (-1,-1*\dy) -- node[sloped,left,rotate=270] {
\begin{tabular}{c}
$\mathtt{a}$\\$\mathtt{c}$\\$\mathtt{b}$\\$\mathtt{a}$
\end{tabular}
} (-1.7,-5*\dy);
\draw (-1,-1*\dy) -- node[sloped,right,rotate=90] {
\begin{tabular}{c}
$\mathtt{b}$\\$\mathtt{c}$\\$\mathtt{b}$\\$\mathtt{a}$
\end{tabular}
} (-0.3,-5*\dy);

\draw (0,0) -- node[sloped,right,rotate=90] {
\begin{tabular}{c}
$\mathtt{b}$
\end{tabular}
} (2,-1*\dy);
\draw (2,-1*\dy) -- node[sloped,left,rotate=270] {
\begin{tabular}{c}
$\mathtt{a}$\\$\mathtt{b}$
\end{tabular}
} (1,-3*\dy);
\draw (1,-3*\dy) -- node[sloped,left,rotate=270] {
\begin{tabular}{c}
$\mathtt{a}$\\$\mathtt{b}$
\end{tabular}
} (0.5,-5*\dy);
\draw (1,-3*\dy) -- node[sloped,right,rotate=90] {
\begin{tabular}{c}
$\mathtt{c}$
\end{tabular}
} (1.5,-4*\dy);
\draw (2,-1*\dy) -- node[sloped,right,rotate=90] {
\begin{tabular}{c}
$\mathtt{b}$\\$\mathtt{c}$\\$\mathtt{b}$\\$\mathtt{a}$
\end{tabular}
} (2.5,-5*\dy);

\foreach \x/\y in {-0.3/-5,2/-1,1/-3,0.5/-5}{
\filldraw (\x,\y*\dy) circle (0.08cm);
}
\foreach \x/\y in {0/0,-1/-1,-1.7/-5,1.5/-4,2.5/-5}{
\filldraw[white] (\x,\y*\dy) circle (0.08cm);
\draw (\x,\y*\dy) circle (0.08cm);
}
\foreach \x/\y/\i in {-1.7/-5/3,-0.3/-5/2,0.5/-5/1,1.5/-4/4,2.5/-5/5}{
\draw (\x,\y*\dy) node[below] {$B_{\i}$};
}

\draw (0.4,-6.5*\dy) node {$\T(\F_2)$};
\end{scope}
\end{tikzpicture}
\caption{
The compacted tries $\T(\F_1)$ and $\T(\F_2)$ for families $\F_1$ and $\F_2$ from \cref{ex:F1} and \cref{ex:G}, respectively, with heavy (full circles) and light nodes (empty circles) marked. For $i \in \{1,2\}$, for each light non-leaf node in the trie $\T(\F_i)$, a subfamily of $\Fam_i$ is created.
}\label{fig:HLD}
\end{figure}

It remains to prove that $\Fam_d$ satisfies all the conditions of \cref{prp:complete}.
First, we argue about the bullet points.
Each modified string $F^{\Delta}\in \F'$ gives rise to at most two modified strings added to a single set $\F'_w$ and $\Oh(\min(\ell,\log|\F|))$ modified strings across sets $\F'_w$ for light nodes $w$ in $\T(\F')$.
Since each string $F\in \F$ is the source of $\Oh(1)$ modified strings in any single set $\F'\in \Fam_{d-1}$
and $\Oh(\min(\ell,\log|\F|)^{d-1})$ modified strings across all sets $\F'\in \Fam_{d-1}$,
it is the source of $\Oh(1)$ modified strings in any single set $\F'_w\in \Fam_{d}$
and $\Oh(\min(\ell,\log|\F|)^{d})$ modified strings across all sets $\F'_w\in \Fam_{d}$.

Finally, we shall argue that $\Fam_d$ is indeed a $d$-complete family.
The modified strings in $\F'\in \Fam_{d-1}$ have sources in $\F$ and up to $d-1$ modifications.
Each string inserted to $\F'_w\in \Fam_d$ retains its source and may have up to one more modification, for a total of up to $d$.
Next, consider two strings $U,V\in \F$.
Since $\Fam_{d-1}$ is a $(d-1)$-complete family, there is a set $\F'\in \Fam_{d-1}$ and modified strings $U^{\Delta},V^{\nabla}\in \F'$ forming a $(U,V)_{d-1}$-maxpair.
Let $v$ be the node of $\T(\F')$ representing the longest common prefix of $U^{\Delta}$ and $V^{\nabla}$,
and let $w$ be the lowest light ancestor of $v$ (so that $v$ lies on the path from $w$ to $h(w)$).
By \cref{fct:pair}, we have $p:= \LCP(U^{\Delta},V^{\nabla})=\LCP_{d-1}(U,V)$.
\cref{def:pair} further yields $\Delta,\nabla \sub [1\dd p]\times \Sigma$.
If $\LCP_{d}(U,V)=p$, then $U^{\Delta},V^{\nabla}\in \F'_w$ form a $(U,V)_d$-maxpair.
Otherwise, let $c = \val(h(w))[p+1]$ and note that $U[p+1]\ne V[p+1]$.
If $U[p+1]\ne c \ne V[p+1]$, then $U^{\Delta\cup\{(p+1,c)\}},V^{\nabla\cup\{(p+1,c)\}}\in \F'_w$ form a $(U,V)_d$-maxpair.
If $U[p+1]\ne c = V[p+1]$, then $U^{\Delta\cup\{(p+1,c)\}},V^{\nabla}\in \F'_w$ form a $(U,V)_d$-maxpair.
Symmetrically, if $U[p+1]=c \ne V[p+1]$, then $U^{\Delta},V^{\nabla\cup\{(p+1,c)\}}\in \F'_w$ form a $(U,V)_d$-maxpair.
Thus, in all four cases, $\F'_w\in \Fam_d$ contains a $(U,V)_d$-maxpair. This completes the proof of \cref{prp:complete}.

\end{document}